\documentclass[11pt,letterpaper]{article}
\usepackage[letterpaper, margin=1in]{geometry}
\usepackage{graphicx} %
\usepackage{style}
\usepackage{algorithm}
\usepackage[noend]{algpseudocode}
\usepackage{mdframed}

\providecommand{\Comments}{1}
\ifnum\Comments=1
\usepackage[colorinlistoftodos,prependcaption,textsize=scriptsize]{todonotes}
\definecolor{Gred}{RGB}{219, 50, 54}
\definecolor{Ggreen}{RGB}{60, 186, 84}
\definecolor{Gblue}{RGB}{72, 133, 237}
\definecolor{Gyellow}{RGB}{204, 204, 0}
\definecolor{Gpurple}{RGB}{204, 0, 204}
\definecolor{Gorange}{RGB}{255, 200, 120}
\definecolor{Gbrown}{RGB}{0, 204, 204}
\paperwidth=\dimexpr \paperwidth %
\marginparwidth=\dimexpr\marginparwidth %
\else
\usepackage[disable]{todonotes}
\fi

\usepackage{colortbl}

\newcommand{\mytodo}[1]{\ifnum\Comments=1{#1}\fi}

\let\epsilon\varepsilon

\newcommand{\AlgName}[1]{{\normalfont{\texttt{#1}}}}
\newcommand{\AlgCompare}{\AlgName{Compare}}
\newcommand{\AlgEstimateRatio}{\AlgName{EstimateRatio}}
\newcommand{\AlgQuicksortClustering}{\AlgName{QuicksortClustering}}
\newcommand{\AlgBuildEstimationForest}{\AlgName{BuildEstimationForest}}
\newcommand{\AlgGenerateWeights}{\AlgName{GenerateWeights}}
\newcommand{\AlgGenerateWeightsRec}{\AlgName{GenerateWeightsRec}}
\newcommand{\AlgClusterSort}{\AlgName{ClusterSort}}
\newcommand{\AlgBalancedEstimateRatio}{\AlgName{BalancedEstimateRatio}}
\newcommand{\AlgGetGeometric}{\AlgName{GetGeometric}}

\newcommand{\AlgBuildForestNonAdaptive}{\AlgName{BuildBalancedEstimationForest}}
\newcommand{\AlgBuildWeakForest}{\AlgName{BuildWeakEstimationForest}}
\newcommand{\AlgEstimateWeakRatio}{\AlgName{EstimateOneSidedRatio}}
\DeclareMathOperator*{\myassign}{=}

\newcommand{\dmul}{d_{\text{mul}}}

\newcommand{\orderingerror}{\varepsilon_o}

\newcommand{\cluster}{\gamma}

\DeclareMathOperator{\geom}{Geom}
\DeclareMathOperator{\bin}{Bin}
\DeclareMathOperator{\identityperm}{id}

\title{Learning Multinomial Logits in $O(n \log n)$ time}

\author{Flavio Chierichetti\thanks{Work done in part while at Sapienza University of Rome.} \\
Reddit\\
\texttt{\large flavio.chierichetti@gmail.com}
\and
Mirko Giacchini$^*$\\
Columbia University, Italian Academy Fellow\\
\texttt{\large mg5011@columbia.edu}
\and
Ravi Kumar\\
Google Research, Mountain View\\
\texttt{\large ravi.k53@gmail.com}
\and
Silvio Lattanzi\\
Google Research, Barcelona\\
\texttt{\large silviol@google.com}
\and
Alessandro Panconesi\\
Sapienza University of Rome\\
\texttt{\large ale@di.uniroma1.it}
\and
Erasmo Tani$^*$\\
Cornell University\\
\texttt{\large et484@cornell.edu}
\and
Andrew Tomkins\\
Google Research, Mountain View\\
\texttt{\large atomkins@gmail.com}
}

\date{}

\allowdisplaybreaks

\begin{document}

\maketitle
\thispagestyle{empty}
\addtocounter{page}{-1}

\begin{abstract}
A Multinomial Logit (MNL) model is composed of a finite universe of items $[n]=\{1,\dots, n\}$, each assigned a positive weight. A query specifies an admissible subset---called a \emph{slate}---and the model chooses one item from that slate with probability proportional to its weight. This query model is also known as the \emph{Plackett--Luce model} or \emph{conditional sampling} oracle in the literature. Although MNLs have been studied extensively, a basic computational question remains open: given query access to slates, how efficiently can we learn weights so that, for \emph{every} slate, the induced choice distribution is within total variation distance $\varepsilon$ of the ground truth? This question is central to MNL learning and has direct implications for modern recommender system interfaces.

We provide both adaptive and non-adaptive algorithms for this task. Our first adaptive algorithm makes $O(\frac{n \log n}{\epsilon^2})$ queries and learns every choice distribution in total variation distance $\epsilon$. We also provide a second adaptive algorithm that with $O(\frac{n \log n}{\epsilon^3})$ queries provides a stronger guarantee: each choice distribution is approximated to a multiplicative $1\pm\epsilon$ error (except for subsets of mass smaller than $\epsilon$). We complement these algorithms with lower bounds of $\Omega(\frac{n \log n}{\epsilon^2})$ for total variation distance and $\Omega(\frac{n \log n}{\epsilon^3})$ for the multiplicative error guarantee, showing the optimality of our adaptive algorithms. 

Our non-adaptive algorithm makes  $O\left(\tfrac{n^{2}}{\varepsilon^{3}}\log n \log\tfrac{n}{\varepsilon}\right)$ queries to obtain guarantees under the multiplicative error notion. We complement this result with an $\Omega(\frac{n^2 \log n}{\epsilon^2})$ lower bound for any non-adaptive algorithm, showing that our algorithm is tight in its dependence on the support size $n$, within a $\log n$ factor.

All algorithms query only slates of size two and run in time proportional to their query complexity.
\end{abstract}

\newpage

\hypersetup{hidelinks}
\tableofcontents
\hypersetup{
  colorlinks=false,            %
  pdfborder={0 0 1},           %
  linkbordercolor={1 0 0},     %
  citebordercolor={0 1 0},     %
  urlbordercolor={0 1 1},
}
\thispagestyle{empty}
\addtocounter{page}{-1}

\newpage

\section{Introduction}

Multinomial Logit models (MNLs), also known as softmax or Plackett--Luce models, are widely used to model choice behavior in machine learning and economics. They describe winning distributions over alternatives parameterized by item weights: given a universe $U$ of items, an MNL $M$ assigns each $i \in U$ a weight $w_i$, and for any non-empty subset $S \subseteq U$, defines $M_S(i) = w_i / \sum_{j \in S} w_j$ as the probability of selecting $i$ from $S$. Such models underlie diverse applications, from token prediction in large language models to content selection in recommender systems, where they capture how preferences depend on the available slate.  

Most prior work focuses on estimating MNL parameters or the induced distribution on the universal slate $S=U$, which suffices for identifying the globally most preferred items or obtaining a top-$k$ ranking. In contrast, we address the more challenging task of approximating the MNL distribution for \emph{all} slates, motivated by practical needs in modern recommender systems.  Consider a platform such as Netflix offering choices of movies. It is now broadly understood that simply displaying a very long list of top titles does not provide a compelling user experience. Instead, these platforms define a large and rapidly changing number of relevant subsets of the entire movie catalog: action movies, foreign movies, movies similar to a particular anchor movie the user recently watched, and so forth. The interface then shows a sequence of carousels, perhaps a carousel of ``top movies'' followed by ``movies similar to Ponyo'' then ``new arrivals'', each one ordered to show the user's best options from the class. To drive such an interface, it is important to approximate the winning distribution for \emph{every} one of these subsets simultaneously, to be ready to display it when needed. As the possible subsets of interest are constantly updated by the platform, it is critical to approximate the MNL's output on all possible subsets $S \subseteq U$.

Furthermore, for a particular subset, such as that containing all action movies, the platform will not show a single option, but will instead show a carousel with a moderate number of suggestions. While some previous work focused exclusively on ranking the items, practical recommender systems require scoring them for at least two key reasons. First, the number of items shown should depend on their scores: if there are four high-scoring movies, it might be better to only display those, rather than adding the next six, which may have very little chance of being selected. Second, the interface might have more richness than just the carousel itself. For instance, if the top movie of a carousel has a much higher score than the next one, the platform might feature this movie more prominently, for example, by using a specialized rendering or by allocating more space to it. Hence, it is crucial to obtain estimates of the weights that provide accurate winning distributions on all slates.
 
\paragraph{MNL and MNL Learning.}
A \emph{multinomial logit (MNL) model} supported on the universe $U=[n]=\{1,\dots, n\}$ is specified by a set $\{w_1, \ldots,w_n\}$ of $n$ positive values called \emph{weights}. 
A \emph{slate} is a non-empty subset of $[n]$.  An MNL $M$, for any given slate $S \subseteq [n]$, induces a conditional distribution denoted $M_S$ whose support is $S$ and where the probability of each item $i \in S$ is given by:\footnote{The terminology we adopt comes from the Economics literature \citep{t09}, this is called a \emph{logit} model because if we let $w_i=e^{\theta_i}$, then $M_S(i)=\text{softmax}(S)_i = e^{\theta_i} / \sum_{j\in S}e^{\theta_j}$.}
\[
    M_S(i) = \frac{w_i}{\sum_{j \in S} w_j}.  
\]
An MNL $M$ can be accessed by a $\maxsample$  oracle, which operates as follows: given a slate $S$, $\maxsample(S)$ returns $i \in S$ chosen according to the distribution $M_S$. Given MNLs $M$ and $M'$, we define two notions of distance between them:
\[
    d_{\infty}(M, M') := \max_{\substack{S\subseteq[n]\\S\neq \varnothing}} \,\norm{M_{S} - M'_{S}}_\infty
\quad \mbox{ and }\quad
    d_1(M, M') := \max_{\substack{S\subseteq[n]\\S\neq \varnothing}} \,\norm{M_{S} - M'_S}_1.
\]
We will also consider a multiplicative notion of error. Specifically, we say that $\dmul(M,M')\leq \epsilon$ if for each slate $S\subseteq [n]$ there exists $S'\subseteq S$ such that:
\begin{align*}
    &(1-\epsilon) \cdot M_S(i)\leq M'_S(i)\leq (1+\epsilon)\cdot M_S(i) \quad  \text{for each }i\in S',\\
    &\sum_{i\in S\setminus S'} M_S(i)\leq \epsilon,\\
    &\sum_{i\in S\setminus S'} M'_S(i)\leq \epsilon
\end{align*}
That is, $M'$ provides multiplicative $1\pm \epsilon$ approximation of the winning distributions, except for sub-slates that have total mass of at most $\epsilon$. Note that $d_1(M,M') \leq 3 \cdot \dmul(M,M')$, and therefore learning an MNL in $\dmul$ error directly implies learning for $d_1$ (and $d_\infty$). This type of multiplicative guarantee is also common in the distribution testing literature \citep{crs15}.

In this paper we obtain algorithms that approximate an unknown MNL $M$ in $\dmul$ error and $d_1$ distance, while our lower bounds apply even to the less challenging problem of obtaining estimates with small $d_\infty$ distance.
 
\begin{definition}[MNL Learning Problem]
Given $\maxsample$ oracle access to an MNL $M$ and an $\varepsilon \in (0, 1)$, the \emph{MNL learning problem for $d_1$} (resp. $\dmul$, and $d_\infty$) is to output an MNL $\hat{M}$ such that 
$d_1(M, \hat{M}) \leq \epsilon$ (resp. $\dmul(M, \hat{M}) \leq \epsilon$, and $d_\infty(M, \hat{M}) \leq \epsilon$). The MNL produced in output is represented using the logarithms of its weights.\footnote{Representing an MNL using the logarithms of its weights is standard in the ML and economics community \citep{sru20,t09}. Moreover, with a full representation of the $\hat{w}_i$'s, the weights could require $\Omega(n^2)$ bits just to be stored. For instance, consider the MNL on $[n]$ with weights $w_i=2^i$. Since $\frac{w_{i+1}}{w_i + w_{i+1}} = \frac{2}{3}$, any MNL $\hat{M}$ solving the MNL learning problem must satisfy $\hat{w}_{i+1} \geq \hat{w}_{i} \cdot (2-9\epsilon)$. Therefore, each weight $\{\hat{w}_{n/2}, \dots, \hat{w}_n\}$ requires $\Omega(n)$ bits for a total of $\Omega(n^2)$ bits. Hence, requiring that an algorithm outputs the weights, rather than their logarithms, would rule out the possibility of constructing any algorithm that runs in time $o(n^2)$. Other compact representations of the weights are also possible. %
} 
\end{definition}

\paragraph{Main Results.}
In this paper, we study algorithms for the MNL learning problem. We obtain three algorithms, two using adaptive queries and the other using non‑adaptive queries. All algorithms query only slates of size two and run in time proportional to their query complexity.
Our adaptive algorithm for $\dmul$ makes $O\left(\tfrac{n}{\varepsilon^{3}}\log n\right)$ queries while our adaptive algorithm for $d_1$ makes $O\left(\tfrac{n}{\varepsilon^{2}}\log n\right)$ queries. We give lower bounds matching our algorithmic results up to constant factors: we give lower bounds of $\Omega(\frac{n \log n}{\epsilon^3})$ for $\dmul$-error and $\Omega(\frac{n\log n}{\epsilon^2})$ for $d_\infty$-error (and therefore, also $d_1$). Summarizing:
\begin{theorem}[Informal]
    For any constant $\epsilon>0$, the complexity of learning an MNL within $\dmul$-error $\varepsilon$ by making $\maxsample$ queries adaptively is $\Theta(n\log n)$.
\end{theorem}

Our non-adaptive algorithm makes $O\left(\tfrac{n^{2}}{\varepsilon^{3}}\log n \log\tfrac{n}{\varepsilon}\right)$ queries to learn an MNL in $\dmul$-error; this is complemented by a lower bound of $\Omega({{n^2}\over \varepsilon^2}\log n)$ for $d_\infty$-error.  Summarizing:

\begin{theorem}[Informal]
    For any constant $\epsilon>0$, the complexity of learning an MNL within $\dmul$-error $\varepsilon$ by making $\maxsample$ queries non-adaptively is between $O(n^2\log^2 n)$ and $\Omega(n^2\log n)$.
\end{theorem}

Our results are surprising: our seemingly harder problem can be solved as fast as (noisy) sorting. Furthermore, our lower bounds hold for unit-time oracle queries of any slate size. Hence, restricting the algorithms to slates of size two incurs no loss in efficiency. 

\paragraph{Technical Challenges.}
Existing methods, especially ones that approximate the winning distribution over the universal slate $[n]$, do not seem to apply to our problem.  As a simple example of the difficulty, consider an algorithm that guarantees an $\ell_1$-estimate of the full slate distribution within an error of $\epsilon\in(0,1/2)$. Consider now the MNL on $\{1, 2, 3\}$ with weights $w_1=1-\epsilon$, $w_2=w_3=\epsilon/2$.  Suppose the algorithm returns the estimate $\hat{w}_1=1-\epsilon$, $\hat{w}_2=\frac{3\epsilon}{4}$, $\hat{w}_3=\frac{\epsilon}{4}$; clearly, $\norm{w-\hat{w}}_1 = \frac{\epsilon}{2} \leq \epsilon$. But, $\left|\frac{w_2}{w_2+w_3}- \frac{\hat{w}_2}{\hat{w}_2+\hat{w}_3}\right| \geq \frac{1}{4}$, and therefore the algorithm cannot guarantee small error on the slate $\{2, 3\}$. Similarly, as we discuss in Appendix~\ref{sec:additive-approximation-is-not-sufficient}, prior work that additively estimates the winning distributions on all size-two slates cannot be used to obtain a good approximation on every slate. Moreover, noisy sorting algorithms also do not imply our guarantees. Indeed, while our final algorithm will perform a noisy sorting step at the beginning, the key challenge lays in performing the additional comparisons that allow the accurate estimate of each conditional distribution, as we detail below. 

It is not difficult to obtain a cubic time algorithm for our problem.  Indeed, consider  the naive algorithm that works as follows. For each pair $\{i,j\}$ of items in the universe, estimate $w_i/w_j$ to within a $(1 \pm \epsilon)$-multiplicative error, or declare that their ratio (or its inverse) is larger than $\frac n \epsilon$. One can easily show that this algorithm will need to query each pair $\approx \frac{n \log n}{\epsilon^3}$ times to guarantee these bounds; the total cost would then be $\approx \frac{n^3 \log n}{\epsilon^3}$. From the output of this algorithm, it is easy %
to approximate the output distribution for any slate.\footnote{Indeed, if the items of this slate have weights that are within a $\frac n\epsilon$ factor of each other, the $(1\pm \epsilon)$-approximation error will make it possible to approximate the winning probability of any item to within a $(1 \pm O(\epsilon))$-factor (so that the total variation error will be at most $O(\epsilon)$). If, instead, the slate contains pairs $\{i,j\}$ of items such that $w_i / w_j < \frac{\epsilon} n$, then the lighter item $i$  will have a probability of winning in the slate not larger than $O(\frac{\epsilon}n)$, and hence we can estimate its winning probability to be zero---given that there are at most $n-1$ such light items in a slate, the total variation error is no larger than $O(\epsilon)$.}

With some effort, this algorithm can be improved.  The idea is to carefully control the pairs of items, querying only pairs that are nearby in the order induced by the weights (which can be approximately inferred via a noisy sorting algorithm). To avoid querying too many nearby pairs, one can first cluster items whose weights are within a constant factor, in a query-efficient manner, and then select a center from each cluster.  One can determine the weight ratio of each item to its cluster center, and the weight ratios of successive items in the sorted list of cluster centers. This method can be shown to produce an algorithm that compares $O(n)$ pairs of items, with each such comparison performing $\approx \frac{n}{\epsilon^3} \log^3 n$ queries.  While this yields a quadratic time algorithm, it is unclear how this can be further improved to being quasi-linear.

\paragraph{Overview of Methods.}

We focus on our algorithm for $\dmul$, which we find more intuitive to highlight our ideas. We construct our quasi-linear adaptive algorithm by building on the clustering idea described above. We first partition the universe into clusters of similar-weight items and select a center for each cluster. We then estimate the ratio of the weights $w_{i}/w_c$ for every item $i$ in the cluster with center $c$. Finally, we construct a forest on the cluster centers, where every edge is labeled with an estimate of the ratio between the weights of the two centers it connects. We call this data structure the \emph{estimation-forest}. This allows us to obtain estimates of the ratio of the weights for arbitrary pairs of items by combining these estimates along paths in the forest.

Following this strategy, the error compounds multiplicatively along the paths. To circumvent this issue without requiring more accurate ratio estimates---which would lead to a higher complexity--- we design the forest so that any two centers whose weights the algorithm might want to compare are at a short distance from each other. To achieve this property, the topology of the forest is constructed adaptively.
Additionally, to further improve the sample complexity (and runtime), we dynamically adjust the number of queries required to approximate the weight ratio between two centers. In particular, if the total weight of items lighter than a given item $i$ is not %
large enough with respect to 
the weight of $i$, then it becomes unimportant to estimate the   ratio of the weight of $i$ over the weight of any of these items.
Our estimation-forest data structure dynamically determines query sequences for weight-ratio estimation and enables all cluster-center–cluster-center and cluster-item–cluster-center comparisons in $O\left(\frac{n}{\epsilon^3}\log n\right)$ queries. To obtain a tight bound in $\epsilon$ for $d_1$-error, we relax the notion of estimation-forest and allow some edges to be approximated to less precision if this does not hurt the overall estimation. 

While the above algorithm is adaptive, we also obtain a non-adaptive version.  The idea is to first design a new adaptive algorithm that queries each pair of items only $O\left(\frac{1}{\epsilon^3} \log n\log{n\over \varepsilon}\right)$ times.  We then query every pair a fixed number of times and then simulate this new  adaptive algorithm on the precomputed answers; this leads to a non-adaptive algorithm with $O\left({n^2\over \varepsilon^3}\log n\log{n\over \varepsilon}\right)$ queries.

The three lower bounds in our paper are proved using a reduction from the problem of identifying the $k$ coins with the highest heads probability in a collection of $n$ biased coins. %
In particular, we construct an MNL supported on an even-sized universe whose items are divided into pairs, each pair representing the two sides of a coin. We then order the pairs so that the weights of the items in a pair are much larger than those in preceding pairs. This way, %
we can assume without loss of generality that any learning algorithm is only querying slates corresponding to our original pairs.

\paragraph{Organization.}
In Section~\ref{sec:relwork} we review related work. Section~\ref{sec:preliminaries} introduces key tools and notation that we use throughout the paper. Section~\ref{sec:overview} gives a more detailed technical overview of our algorithms and techniques. In Section~\ref{sec:algorithmic-primitives} we introduce two estimation primitives used by our main algorithms, which we present in the subsequent sections: in Section~\ref{sec:adaptive-algorithm} we analyze our adaptive algorithm for $\dmul$-error, and then we relax the notion of estimation-forest to obtain a tight dependence in $\epsilon$ for $d_1$-error in \Cref{sec:adaptive-optimal-eps-algorithm}, finally in Section~\ref{sec:non-adaptive-algorithm} we present our non-adaptive algorithm. Section~\ref{sec:lower-bounds} contains the proofs of our adaptive and non-adaptive lower bounds. We conclude in Section~\ref{sec:conc} with several open questions.

All the proofs missing from the main body of the paper can be found in the appendix.

\section{Related Work}\label{sec:relwork}

The problem of learning MNLs on \emph{all}  slates from $\maxsample$ queries arises naturally from several perspectives. Our work is related to, yet distinct from, the existing literature. First, prior work on MNL fitting has provided approximation guarantees only for the full slate or for pairs of items—both of which are strictly weaker than the guarantees we obtain. Second, our framework extends beyond classical MNL ranking and selection by capturing quantitative relationships among items, revealing how much and where certain items dominate, while preserving the $O(n \log n)$ efficiency of the best known ranking algorithms. Third, our problem can be viewed as a natural strengthening of distribution learning under conditional sampling, extending the “testing by learning” paradigm to recover all conditional distributions simultaneously in a more expressive and challenging setting. Finally, our results also strengthen prior work on learning Random Utility Models (RUMs), specialized to the MNL case. We elaborate on these connections below.

\paragraph{MNL Fitting.} 
A large body of the literature focuses on finding MNL weights maximizing the likelihood of a collected dataset. In this setting, usually, the queries are either fixed \citep{z29,f57,d60,n23} or sampled from a distribution \citep{oz24,nos12,nos17,mg15}. When the dataset actually comes from a hidden MNL model, some of these algorithms guarantee that the estimated (normalized) weights approximate the hidden (normalized) weights \citep{nos12,nos17,mg15,sru20,sbbprw15,spu19}. 

These works are not directly applicable to our setting because of the following two main issues.
(i) Approximately recovering the normalized weights is equivalent to providing a good estimate of the winning distribution of the full slate. However, this is insufficient to accurately estimate the winning distribution for smaller slates, as we discussed in the Introduction. 
(ii) Most of these works assume that the maximum ratio between two weights is upper bounded by a constant \citep{nos12}.  Note that such an assumption would greatly simplify our problem given that we could accurately estimate the weights with respect to any anchor item.  Therefore, the interesting setting is one where there is no \emph{a priori} bound on the ratio of the weights. Furthermore, as these algorithms are non-adaptive, they are subject to our lower bound of $\Omega(n^2 \log n)$ queries for our problem. 

Stepping outside the task of fitting the weights themselves, 
\cite{fjopr18} adaptively query $O\left(\frac{n \cdot \log (n) \cdot \min\{n, 1/\epsilon\}}{\epsilon^2}\right)$ slates of size two (i.e., pairs) and produce an \emph{additive} estimate within $\epsilon$ for all other \emph{pairs}, in a family of models that are more powerful than MNLs. However, in order for an additive approximation of the slates of size two to generalize to all other slates with an $\ell_\infty$-error of $\epsilon'$, one needs $\epsilon\leq \frac{\epsilon'}{n}$ (proved in \Cref{sec:additive-approximation-is-not-sufficient}). Therefore, applying their algorithm as a blackbox would require $\Omega(n^4 \log n)$ queries. Moreover, their algorithm heavily relies on providing an additive approximation---specifically, each estimated probability is rounded to the closest multiple of $\epsilon$. Hence, it appears hard to generalize their work to larger slates even in a non-blackbox manner.  

Other works considered minimizing the $\ell_\infty$-difference of the logarithm of the weights, or approximating the weight ratios of all the pairs \citep{yyx12,hytc20,sy99,lsr22}. These guarantees would be sufficient to obtain a good estimate for all slates, but are not necessary. Indeed, if $w_1=e^{1000}$, $w_2=e^\epsilon$ and $\hat{w}_1=e^{100}$, $\hat{w}_2=e^\epsilon$, then $\norm{\ln(w) - \ln (\hat{w})}_\infty$ can be made arbitrarily large by increasing the constant $1000$, yet, $\hat{w}$ provides good predictions: $|\frac{w_1}{w_1+w_2} - \frac{\hat{w}_1}{\hat{w}_1+\hat{w}_2}| \leq \epsilon$. Moreover, also these works analyze non-adaptive algorithms and assume a bound on the ratio of the weights, and therefore suffer our $\Omega(n^2 \log n)$ lower bound.

We also mention a separate line of work that focuses on developing statistical tests to determine whether a given dataset of comparisons is consistent with an MNL model \citep{su19,ms25,rbss22}. These works are complementary to ours in that they address model validation rather than estimation, and they do not provide algorithms for learning the underlying MNL parameters.

\paragraph{MNL Ranking/Selection.} Other classical problems involving MNLs include: (i) sorting/ranking the weights \citep{fops17, fjopr18, cgz22, sbph15, rls19}, (ii) finding the top-$k$ items with largest weight \citep{jkso17,cgms17,clm18,cs15, ktas12}, and (iii) finding the item of maximum weight \citep{emm02, mt04}. Some works make assumptions about the weights (e.g., adjacent weights are sufficiently separated) and seek an exact output with high probability \citep{jkso17}, while others do not make further assumptions but only require a probably approximately correct (PAC) output \citep{sbph15}. These problems have also been explored in the ``dueling bandits'' literature \citep{bbeh21}.

While we will use an $O(\frac{n \log n}{\epsilon^2})$ approximate sorting algorithm by \cite{fjopr18} as a first step in our algorithm, these results are not sufficient by themselves to learn an MNL under our definition. Our lower bound will be proved by showing that the top-$\frac{n}{2}$ problem (and the ranking problem) reduces to our MNL learning problem. Interestingly, despite this, we obtain an $O(\frac{n \log n}{\epsilon^2})$ learning algorithm which is (asymptotically) as efficient as the best possible algorithm for MNL ranking \citep{fjopr18}. 

\paragraph{Distribution Testing with Conditional Samples.} Our problem can also be described in the context of conditional sampling. Let $\mu$ be a hidden distribution over $[n]$. Algorithms can, adaptively, make the following types of queries: chosen a set $S\subseteq[n]$, an oracle returns an item of $S$ sampled according to distribution $\mu$ conditioned on $S$.\footnote{If $S$ has probability 0, the oracle returns a uniform at random item from $S$.} The goal in distribution testing is usually to make the smallest number of queries to establish whether $\mu$ satisfies certain properties, such as, e.g., uniformity~\citep{c20}. However, the problem of estimating the probability $\mu(i)$, for $i\in [n]$, has also been considered \citep{cfgm13,crs15,afl25}. 

Our problem, on the other hand, asks for the minimum number of queries to accurately estimate $\mu(i\mid S)$ for each $i\in S\subseteq [n]$.\footnote{In our proofs, we will assume that the weights are strictly positive for simplicity. However, the same algorithms also work when weights of zero are allowed, as we show in \Cref{sec:extension-to-pseudo-mnl}.} Indeed, the distribution $\mu$ can be seen as the weights of an MNL $M$ and therefore, $\mu(i\mid S) = M_S(i)$ for $i\in S\subseteq[n]$. Observe that having an estimate only for $\mu(i)$ is equivalent to an estimate of the winning distribution on the full slate---insufficient to estimate the winning distribution for smaller slates. Some algorithms provide a multiplicative $(1\pm \epsilon)$ estimate for $\mu(i)$ for $i\notin B$ where $B$ is a set such that $\sum_{b\in B}\mu(b)\leq \epsilon$; this is a stronger property than an additive approximation of the full slate. However, it still cannot provide accurate estimates for slates that are either subsets of $B$ or that span across $B$ and $[n]\setminus B$. Note also that it can be $|B|=\Theta(n)$ meaning that the distribution of most slates cannot be estimated. From a technical standpoint, \cite{cfgm13} achieve this guarantee by building a complete binary tree where edges are labeled with probabilities. This idea bears some high level similarities with our estimation-forest, but details differ. Indeed, their tree is static while the topology of our forest is adaptively chosen, which is crucial for a tight $O(n\log n)$ bound. Moreover, their tree is populated by querying slates of arbitrary size, while we only query slates of size two. Finally, as argued above, the guarantees provided by their tree are insufficient to estimate the winning distributions of all slates. Recent work has also focused on different query models~\citep{a25,pr25,gmp25}; however, their results are incomparable to ours.

We remark that a common paradigm for designing distribution testing algorithms in the traditional (unconditional) setting is that of \emph{testing by learning} (see, e.g.\ \cite{CanonneTopicsDT2022}), in which a property is tested by first approximately learning the underlying distribution, and then checking whether the learned distribution has the property in question. Our work serves as a conditional counterpart to this paradigm that works for the more challenging case in which the property being tested requires approximating the behavior of all conditional distributions.

\paragraph{RUM Learning.} MNLs are a special case of RUMs; hence, algorithms for learning RUMs on all slates could be used to learn an MNL. However, the best known algorithms for general RUM learning require exponentially many queries to slates of size $\Theta(\sqrt{n})$ \citep{cgkpt24}. In contrast, we show that MNLs can be learned using only $O(n \log n)$ queries to slates of size two.

\section{Technical Preliminaries}\label{sec:preliminaries}

Let $U = [n] = \{1, \ldots, n\}$ be a universe of items. For a probability distribution $P$ over $[n]$, let $P(i)$ denote the probability of the item $i \in [n]$.  For distributions $P, Q$, let $\norm{P - Q}_1 := \sum_{i \in [n]} |P(i) - Q(i)|$ be the \emph{$\ell_1$-distance}, which is also twice the total variation distance, and let 
$\norm{P - Q}_\infty := \max_{i \in [n]} |P(i) - Q(i)|$ be the \emph{$\ell_\infty$-distance}. Let $X \sim \Bern(\mu)$ denote a random variable following a Bernoulli distribution with mean $\mu$ and let $X\sim \bin(n,p)$ denote a random variable following a binomial distribution with $n$ trials and head probability $p$. Also let $X\sim\geom(p)$ denote a random variable following a geometric distribution with parameter $p\in(0,1]$; in particular, $\Pr_{X\sim\geom(p)}[X=k] = (1-p)^{k-1} p$, for $k\geq 1$.
For any $x,y \in \R$, we denote by $x\pm y$ the interval $[x-y, x+y]$ and for 
any $x \in \R, \varepsilon\in (0,1)$, we denote by $(1\pm \epsilon)x$ the interval $[(1-\varepsilon)x,(1+\varepsilon)x]$.

\paragraph{Ordered Clusterings and Directed Weightings.}

An \emph{ordered clustering} of $[n]$ is given by an ordered partition $(C_1, \dots, C_T)$ of $[n]$ and a corresponding list $(c_1, \dots, c_T)$ of \emph{centers} such that $c_i\in C_i$ for each $i$.  Here, for $v\in [n]$, let $\cluster(v)\in[T]$ be the unique index such that $v\in C_{\cluster(v)}$; we call $\cluster(v)$ the \emph{cluster index} of $v$. 

Let $F=([n], E)$ be an undirected forest supported on $[n]$.  For $u,v\in [n]$ in the same connected component of $F$, let $P(u,v)$ be the (unique) path in $F$ from $u$ to $v$.  We use $d(u,v)$ to denote the (unweighted/hop) distance in $F$ between vertices $u$ and $v$, where if $u$ and $v$ are in different connected components, we define $d(u,v)=\infty$. 

Let $\vec{E} := \{(u,v) \in V^2 \mid \{u,v\} \in E\}$.  A \emph{directed weighting} of the edges of $F$ is a function $r: \vec{E} \to \R_{>0}$
such that $r(u,v) = \nicefrac{1}{r(v,u)}$.  For a path $P=u_1, \ldots, u_t$ in $F$ define $r(P) = \prod_{i=1}^{t-1} r(u_i, u_{i+1})$, and if $t=1$, let $r(P)=1$.

\section{Overview of Results and Techniques}\label{sec:overview}

In this section, we provide a high-level overview of the results obtained in this paper. Our main focus is in showing how to get a tight dependence of $\Theta(n \log n)$ in the support size with adaptive algorithms. Therefore, we will mostly focus in \Cref{sec:overview-adaptive-dmul} on our adaptive algorithm for learning an MNL in $\dmul$-error which we believe is our most intuitive algorithm. We will then describe in \Cref{sec:overview-adaptive-d1} how to obtain a tight dependence in $\epsilon$ for the $d_1$-distance. In \Cref{sec:overview-non-adaptive} we describe the ideas to make our algorithms non-adaptive, and finally in \Cref{sec:overview-lower-bounds} we describe the main ideas to obtain our lower bounds.

\subsection{Learning MNLs Adaptively in $\dmul$-Error}\label{sec:overview-adaptive-dmul}

Our first result is an algorithm to learn an MNL $M$ by making $O\left({n\over\varepsilon^3} \log n\right)$ adaptive $\maxsample$ queries to output the weights of an MNL $\hat{M}$ such that $\dmul(M,\hat{M}) \leq \varepsilon$. Observe that $M_S(i)= \frac{w_i}{\sum_{s\in S} w_s} = \frac{1}{\sum_{s\in S} \frac{w_s}{w_i}}$. Therefore, if we had access to a multiplicative estimate of the ratio $\nicefrac{w_i}{w_j}$ for each pair $i,j\in [n]$, we could provide a good multiplicative estimate for $M_S$ for each slate $S$. Unfortunately, this has two issues.  (i) In general, this ratio can be unbounded and therefore, producing a multiplicative estimate could in principle cost an unbounded number of queries. (ii) If we aim to obtain an algorithm with query complexity $o(n^2)$, we simply cannot afford to query all the pairs. 

To circumvent these issues, we instead construct a \emph{sparse}  graph on the items of $[n]$ that contains estimates of the ratio $w_i/w_j$ along each edge $\{i,j\}$, and then use this graph to compute $\hat{M}$. At a high level, we produce a forest $F$ such that: (i) if two items are close to each other in $F$, we can get an estimate of their ratios, (ii) if two items are far away in $F$, then their ratio is negligible. We will also need some technical properties to ensure that we can obtain a valid MNL $\hat{M}$ from the forest. The following definition formalizes the properties we need.

\begin{definition}[$(t,\epsilon)$-Estimation-Forest]\label{def:estimation-forest}
Let $t \in \Z, t \geq 2$ and let $\epsilon\in(0,1)$. A $(t,\epsilon)$-\emph{estimation-forest} for an MNL supported on $[n]$ with weights $\{w_1, \dots, w_n\}$ is a tuple
$
    \cF = (F, r, (C_1, \dots, C_T), (c_1, \dots, c_T)),
$
where $F = ([n],E)$ is an undirected forest, $r$ is a directed weighting on $F $, and $(C_1, \dots, C_T), (c_1, \dots , c_T)$ is an ordered clustering over $[n]$. For any $u,v\in [n]$ such that $\cluster(u)\geq \cluster(v)$:
\begin{enumerate}[nosep]
    \item  if $d(u,v)\leq t$, then $r(P(u,v)) \in (1\pm \epsilon) \cdot \frac{w_u}{w_v}$ and $r(P(v,u)) \in (1\pm \epsilon) \cdot \frac{w_v}{w_u}$.
    \item If $d(u,v) \in (t, \infty)$, then:
    \[
        \sum_{\substack{s\in [n]\\ \cluster(s) \leq \cluster(v)}} \frac{w_s}{w_u} \leq \epsilon \quad \text{ and} \quad         \sum_{\substack{s\in \cC\\ \cluster(s) \leq \cluster(v)}} r(P(s,u)) \leq \epsilon,
    \]
    where $\mathcal{C}$ is the connected component containing both $u$ and $v$.
    \item if $d(u,v)=\infty$, then: 
    \[
        \sum_{\substack{s\in [n]\\ \cluster(s) \leq \cluster(v)}} \frac{w_s}{w_u} \leq \epsilon.
    \]
    Also, for any $u'$ (resp. $v'$) in the same connected component of $u$ (resp. $v$), it holds that $\cluster(u') > \cluster(v')$.
    \item if $\cluster(u)=\cluster(v)$, then $d(u,v)\leq t$.
\end{enumerate}
\end{definition}

In Section~\ref{sec:MNL-from-forest-adaptive}, we show that we can use a $(t,\epsilon)$-estimation-forest for an MNL $M$ to obtain an MNL $\hat{M}$ such that $\dmul(M, \hat{M})\leq O(\epsilon)$ (\Cref{thm:estimation-forest-produces-mnl}). 

\paragraph{On Choosing the Estimation-Forest Topology.} Interestingly, for the purpose of constructing $\hat{M}$, it turns out that the specific value of $t$ is irrelevant. This observation allows us to reduce the problem of learning $M$ to that of constructing a $(t,\varepsilon)$-estimation-forest for a single, arbitrary choice of $t$. The central challenge then lies in designing an efficient topology for the estimation-forest. %

The most natural topology would be a path on the items (after a noisy-sorting step). However, along a path, two items of comparable weight can be separated by a super-constant distance $d=\omega(1)$. To preserve property~2 of \Cref{def:estimation-forest}, one would then need to construct a $(d,\varepsilon)$-estimation-forest, which would incur a query cost of $\Omega(n d^2 \log n)$. Since $d$ can be as large as $\Theta(n)$, this  is clearly suboptimal, suggesting the need for a topology with low  diameter.
Note that even a complete binary tree also can yield super-constant length paths, implying an $\omega(n\log n)$ query cost. 

On the other hand, to achieve a very small diameter, one might consider a star or a tree topology with unbounded arity. However, in these cases, one would need to estimate extremely large weight ratios, leading to high query complexity. This, in fact, explains why a disconnected graph is required.

Another natural direction would be to consider general (non-acyclic) graphs. In fact, we could consider a path with skips to decrease the diameter (perhaps exploiting modern shortcutting results \citep{kp22}). The difficulty is that, in a cyclic graph, the weight of an item depends on the particular path chosen, and different paths can yield inconsistent estimates. Thus, acyclicity of the topology is essential to ensure that an explicit MNL can be extracted from it. %

In \Cref{sec:contructing-the-forest-adaptive}, we present an efficient algorithm for constructing an $(O(1),\varepsilon)$-estimation-forest. The resulting topology takes the form of a forest of \emph{lobster graphs}: items are first clustered together, as described below, and a forest of unbounded-arity trees is then constructed over the resulting cluster centers. The arity of each tree is not predetermined but is instead adaptively chosen as the algorithm progresses, in order to balance estimation accuracy and query efficiency. Interestingly, the diameter of the trees in our forest can be super-constant. However, each tree will have the property that if two items are at distance more than $O(1)$, then one of the two is so much larger than the other that their ratio can be taken to be infinite without incurring a large error. Thanks to this property, from the perspective of any single item, one can consider the tree to have constant diameter and lose at most $O(\epsilon)$ in the final estimate.

\paragraph{Building the Estimation-Forest.}
We now go more into the details of our solution to efficiently build an estimation-forest. When constructing the estimation-forest, some ratio estimates might be costlier to obtain than others. In order to maintain a low query complexity, we leverage the fact that if two items have similar weights, fewer queries are required to estimate the ratio of their weights. In the first step to build our estimation-forest, we exploit this observation via a pre-processing step, which sorts the items in approximately increasing order of weights, and produces clusters of similar items resulting in a \emph{cluster graph}, defined as follows. %

\begin{figure}
    \centering
    \tikzset{every picture/.style={line width=0.75pt}} %

\begin{tikzpicture}[x=0.75pt,y=0.75pt,yscale=-1,xscale=1]
\draw    (320.23,82.49) -- (321.44,130.36) ;
\draw    (104.23,84.4) -- (132.54,126.48) ;
\draw    (104.23,84.4) -- (91.63,135.26) ;
\draw    (104.23,84.4) -- (145.77,78.1) ;
\draw  [fill={rgb, 255:red, 0; green, 0; blue, 0 }  ,fill opacity=1 ] (84.87,135.26) .. controls (84.87,131.86) and (87.9,129.1) .. (91.63,129.1) .. controls (95.36,129.1) and (98.39,131.86) .. (98.39,135.26) .. controls (98.39,138.67) and (95.36,141.43) .. (91.63,141.43) .. controls (87.9,141.43) and (84.87,138.67) .. (84.87,135.26) -- cycle ;
\draw  [fill={rgb, 255:red, 0; green, 0; blue, 0 }  ,fill opacity=1 ] (102.92,34.2) .. controls (102.92,30.8) and (105.95,28.04) .. (109.68,28.04) .. controls (113.41,28.04) and (116.44,30.8) .. (116.44,34.2) .. controls (116.44,37.61) and (113.41,40.37) .. (109.68,40.37) .. controls (105.95,40.37) and (102.92,37.61) .. (102.92,34.2) -- cycle ;
\draw  [fill={rgb, 255:red, 0; green, 0; blue, 0 }  ,fill opacity=1 ] (206.22,36.4) .. controls (206.22,32.99) and (209.25,30.23) .. (212.98,30.23) .. controls (216.71,30.23) and (219.74,32.99) .. (219.74,36.4) .. controls (219.74,39.8) and (216.71,42.56) .. (212.98,42.56) .. controls (209.25,42.56) and (206.22,39.8) .. (206.22,36.4) -- cycle ;
\draw  [fill={rgb, 255:red, 0; green, 0; blue, 0 }  ,fill opacity=1 ] (208.63,129.26) .. controls (208.63,125.86) and (211.65,123.1) .. (215.39,123.1) .. controls (219.12,123.1) and (222.14,125.86) .. (222.14,129.26) .. controls (222.14,132.67) and (219.12,135.43) .. (215.39,135.43) .. controls (211.65,135.43) and (208.63,132.67) .. (208.63,129.26) -- cycle ;
\draw  [fill={rgb, 255:red, 0; green, 0; blue, 0 }  ,fill opacity=1 ] (57.2,64.51) .. controls (57.2,61.11) and (60.23,58.35) .. (63.96,58.35) .. controls (67.69,58.35) and (70.72,61.11) .. (70.72,64.51) .. controls (70.72,67.92) and (67.69,70.68) .. (63.96,70.68) .. controls (60.23,70.68) and (57.2,67.92) .. (57.2,64.51) -- cycle ;
\draw  [fill={rgb, 255:red, 0; green, 0; blue, 0 }  ,fill opacity=1 ] (345.61,50.04) .. controls (345.61,46.64) and (348.64,43.88) .. (352.37,43.88) .. controls (356.1,43.88) and (359.13,46.64) .. (359.13,50.04) .. controls (359.13,53.45) and (356.1,56.21) .. (352.37,56.21) .. controls (348.64,56.21) and (345.61,53.45) .. (345.61,50.04) -- cycle ;
\draw  [fill={rgb, 255:red, 0; green, 0; blue, 0 }  ,fill opacity=1 ] (314.68,130.36) .. controls (314.68,126.96) and (317.7,124.2) .. (321.44,124.2) .. controls (325.17,124.2) and (328.19,126.96) .. (328.19,130.36) .. controls (328.19,133.77) and (325.17,136.53) .. (321.44,136.53) .. controls (317.7,136.53) and (314.68,133.77) .. (314.68,130.36) -- cycle ;
\draw  [fill={rgb, 255:red, 0; green, 0; blue, 0 }  ,fill opacity=1 ] (290.82,52.74) .. controls (287.19,51.94) and (284.96,48.61) .. (285.84,45.3) .. controls (286.72,41.99) and (290.37,39.96) .. (294,40.76) .. controls (297.63,41.56) and (299.86,44.89) .. (298.98,48.2) .. controls (298.1,51.51) and (294.45,53.54) .. (290.82,52.74) -- cycle ;
\draw  [dash pattern={on 0.84pt off 2.51pt}] (50.08,84.4) .. controls (50.08,48.42) and (74.32,19.26) .. (104.23,19.26) .. controls (134.13,19.26) and (158.37,48.42) .. (158.37,84.4) .. controls (158.37,120.37) and (134.13,149.53) .. (104.23,149.53) .. controls (74.32,149.53) and (50.08,120.37) .. (50.08,84.4) -- cycle ;
\draw  [dash pattern={on 0.84pt off 2.51pt}] (182.26,86.23) .. controls (182.26,50.06) and (196.27,20.74) .. (213.54,20.74) .. controls (230.82,20.74) and (244.82,50.06) .. (244.82,86.23) .. controls (244.82,122.4) and (230.82,151.73) .. (213.54,151.73) .. controls (196.27,151.73) and (182.26,122.4) .. (182.26,86.23) -- cycle ;
\draw  [dash pattern={on 0.84pt off 2.51pt}] (263.69,85.99) .. controls (263.69,49.35) and (289.11,19.64) .. (320.47,19.64) .. controls (351.82,19.64) and (377.24,49.35) .. (377.24,85.99) .. controls (377.24,122.63) and (351.82,152.33) .. (320.47,152.33) .. controls (289.11,152.33) and (263.69,122.63) .. (263.69,85.99) -- cycle ;
\draw  [fill={rgb, 255:red, 0; green, 0; blue, 0 }  ,fill opacity=1 ] (554.96,43.46) .. controls (554.96,40.05) and (557.99,37.29) .. (561.72,37.29) .. controls (565.45,37.29) and (568.48,40.05) .. (568.48,43.46) .. controls (568.48,46.86) and (565.45,49.62) .. (561.72,49.62) .. controls (557.99,49.62) and (554.96,46.86) .. (554.96,43.46) -- cycle ;
\draw  [dash pattern={on 0.84pt off 2.51pt}] (461.46,84.71) .. controls (461.46,46.95) and (490.01,16.35) .. (525.23,16.35) .. controls (560.45,16.35) and (589,46.95) .. (589,84.71) .. controls (589,122.46) and (560.45,153.06) .. (525.23,153.06) .. controls (490.01,153.06) and (461.46,122.46) .. (461.46,84.71) -- cycle ;
\draw  [fill={rgb, 255:red, 0; green, 0; blue, 0 }  ,fill opacity=1 ] (555.57,119.31) .. controls (559.09,120.46) and (560.91,123.99) .. (559.65,127.19) .. controls (558.4,130.4) and (554.53,132.06) .. (551.01,130.91) .. controls (547.5,129.76) and (545.67,126.24) .. (546.93,123.03) .. controls (548.19,119.83) and (552.06,118.16) .. (555.57,119.31) -- cycle ;
\draw    (524.77,82.49) -- (561.72,43.46) ;
\draw    (524.77,82.49) -- (553.29,125.11) ;
\draw  [fill={rgb, 255:red, 0; green, 0; blue, 0 }  ,fill opacity=1 ] (508.04,35.77) .. controls (508.04,32.37) and (511.07,29.61) .. (514.8,29.61) .. controls (518.53,29.61) and (521.55,32.37) .. (521.55,35.77) .. controls (521.55,39.18) and (518.53,41.94) .. (514.8,41.94) .. controls (511.07,41.94) and (508.04,39.18) .. (508.04,35.77) -- cycle ;
\draw  [fill={rgb, 255:red, 0; green, 0; blue, 0 }  ,fill opacity=1 ] (471.94,68.7) .. controls (471.94,65.3) and (474.97,62.54) .. (478.7,62.54) .. controls (482.43,62.54) and (485.46,65.3) .. (485.46,68.7) .. controls (485.46,72.1) and (482.43,74.86) .. (478.7,74.86) .. controls (474.97,74.86) and (471.94,72.1) .. (471.94,68.7) -- cycle ;
\draw  [fill={rgb, 255:red, 0; green, 0; blue, 0 }  ,fill opacity=1 ] (485.18,121.38) .. controls (485.18,117.98) and (488.2,115.22) .. (491.94,115.22) .. controls (495.67,115.22) and (498.69,117.98) .. (498.69,121.38) .. controls (498.69,124.78) and (495.67,127.54) .. (491.94,127.54) .. controls (488.2,127.54) and (485.18,124.78) .. (485.18,121.38) -- cycle ;
\draw    (478.7,68.7) -- (524.77,82.49) ;
\draw    (514.8,35.77) -- (524.77,82.49) ;
\draw    (524.77,82.49) -- (491.94,121.38) ;
\draw  [fill={rgb, 255:red, 0; green, 0; blue, 0 }  ,fill opacity=1 ] (563.73,86.88) .. controls (563.73,83.48) and (566.76,80.72) .. (570.49,80.72) .. controls (574.22,80.72) and (577.25,83.48) .. (577.25,86.88) .. controls (577.25,90.29) and (574.22,93.05) .. (570.49,93.05) .. controls (566.76,93.05) and (563.73,90.29) .. (563.73,86.88) -- cycle ;
\draw    (524.77,82.49) -- (570.49,86.88) ;
\draw    (293.02,48.18) -- (320.23,82.49) ;
\draw    (320.23,82.49) -- (353.18,48.18) ;
\draw    (212.98,82.49) -- (215.39,129.26) ;
\draw    (212.98,36.4) -- (212.98,82.49) ;
\draw  [fill={rgb, 255:red, 0; green, 0; blue, 0 }  ,fill opacity=1 ] (125.78,126.48) .. controls (125.78,123.08) and (128.81,120.32) .. (132.54,120.32) .. controls (136.27,120.32) and (139.3,123.08) .. (139.3,126.48) .. controls (139.3,129.89) and (136.27,132.65) .. (132.54,132.65) .. controls (128.81,132.65) and (125.78,129.89) .. (125.78,126.48) -- cycle ;
\draw    (109.68,34.2) -- (104.23,84.4) ;
\draw    (63.96,64.51) -- (104.23,84.4) ;
\draw  [fill={rgb, 255:red, 0; green, 0; blue, 0 }  ,fill opacity=1 ] (139.02,78.1) .. controls (139.02,74.7) and (142.04,71.94) .. (145.77,71.94) .. controls (149.51,71.94) and (152.53,74.7) .. (152.53,78.1) .. controls (152.53,81.51) and (149.51,84.27) .. (145.77,84.27) .. controls (142.04,84.27) and (139.02,81.51) .. (139.02,78.1) -- cycle ;
\draw  [fill={rgb, 255:red, 255; green, 255; blue, 255 }  ,fill opacity=1 ] (205.59,77.28) -- (219,77.28) -- (219,90.69) -- (205.59,90.69) -- cycle ;
\draw  [fill={rgb, 255:red, 255; green, 255; blue, 255 }  ,fill opacity=1 ] (97.52,77.69) -- (110.93,77.69) -- (110.93,91.1) -- (97.52,91.1) -- cycle ;
\draw  [fill={rgb, 255:red, 255; green, 255; blue, 255 }  ,fill opacity=1 ] (313.53,75.79) -- (326.94,75.79) -- (326.94,89.2) -- (313.53,89.2) -- cycle ;
\draw  [fill={rgb, 255:red, 255; green, 255; blue, 255 }  ,fill opacity=1 ] (518.06,75.79) -- (531.48,75.79) -- (531.48,89.2) -- (518.06,89.2) -- cycle ;

\draw (80.09,79.85) node [anchor=north west][inner sep=0.75pt]    {$c_{1}$};
\draw (220.74,83.89) node [anchor=north west][inner sep=0.75pt]    {$c_{2}$};
\draw (330.23,78.38) node [anchor=north west][inner sep=0.75pt]    {$c_{3}$};
\draw (96.63,155.53) node [anchor=north west][inner sep=0.75pt]    {$C_{1}$};
\draw (205.95,155.53) node [anchor=north west][inner sep=0.75pt]    {$C_{2}$};
\draw (313.2,155.53) node [anchor=north west][inner sep=0.75pt]    {$C_{3}$};
\draw (518.05,155.53) node [anchor=north west][inner sep=0.75pt]    {$C_{T}$};
\draw (516.81,94.31) node [anchor=north west][inner sep=0.75pt]    {$c_{T}$};
\draw (401.86,76.13) node [anchor=north west][inner sep=0.75pt]  [font=\huge] [align=left] {{\fontfamily{ptm}\selectfont . . .}};

\end{tikzpicture}
    \vspace{-10mm}
    \caption{The structure of an $(A_1,A_2,\varepsilon)$-cluster graph. The vertices of the graph are the items $[n]$ of the MNL, the cluster centers are depicted as white-filled squares, while the other items are represented by black circles. Items in the same cluster have similar weight (within a factor of $A_1$ of each other). Clusters further to the right contain items of higher weights. Associated with each edge $\{u,v\}$, and each direction (say, $u \to v$), is an estimate $r(u,v)$ of the ratio $\nicefrac{w_u}{w_v}$.}
    \label{fig:cluster-graph}
\end{figure}
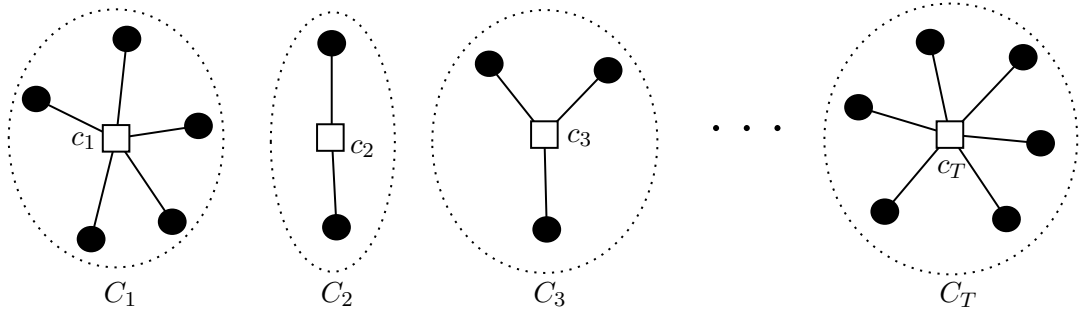
\begin{definition}[Cluster Graph]\label{def:cluster-graph}
An $(A_1,A_2,\varepsilon)$-cluster graph for an MNL supported on $[n]$ with weights $\{w_1, \dots, w_n\}$, is a tuple $\cG = (F,r,(C_1, \dots , C_T), (c_1, \dots , c_T))$, where $F=([n],E)$ is an undirected forest, $r$ is a directed weighting on $F$, and $(C_1,\dots , C_T), (c_1, \dots , c_T)$ is an ordered clustering over $[n]$, satisfying:
\begin{enumerate}[nosep]
    \item For any $i\in[T]$ and any item $u$ in the cluster $C_i$ we have:
    \[
        {1\over A_1} \leq {w_{u} \over w_{c_i}} \leq A_1.
    \]
    \item For any $i,j\in[T]$ with $i > j$ we have:
    \[
        {w_{c_i}\over w_{c_j}} \geq A_2.
    \]
    \item $E$ consists of all the edges of the form $\{c_i,u\}$ for all choices of $i$ and of $u \in C_i$.  Moreover the weight $r(u,v)$ of any edge $\{u,v\}\in E$ satisfies:
    \begin{align*}    
    &r(u,v) \in (1\pm \varepsilon) {w_u \over w_v} \quad \text{ and} \quad  r(v,u) = \frac{1}{r(u,v)} \in (1\pm \epsilon) \frac{w_v}{w_u}.
    \end{align*}
\end{enumerate}
\end{definition}

In Section~\ref{sec:cluster-graph-from-sorting}, we show how to obtain a cluster graph.  Our algorithms employs a noisy sorting procedure of \cite{fjopr18} as a subroutine and builds on it to partition the vertices and compute the edge weights $r$. We show in \Cref{fig:cluster-graph} a cluster graph produced by our algorithm.

\begin{figure}
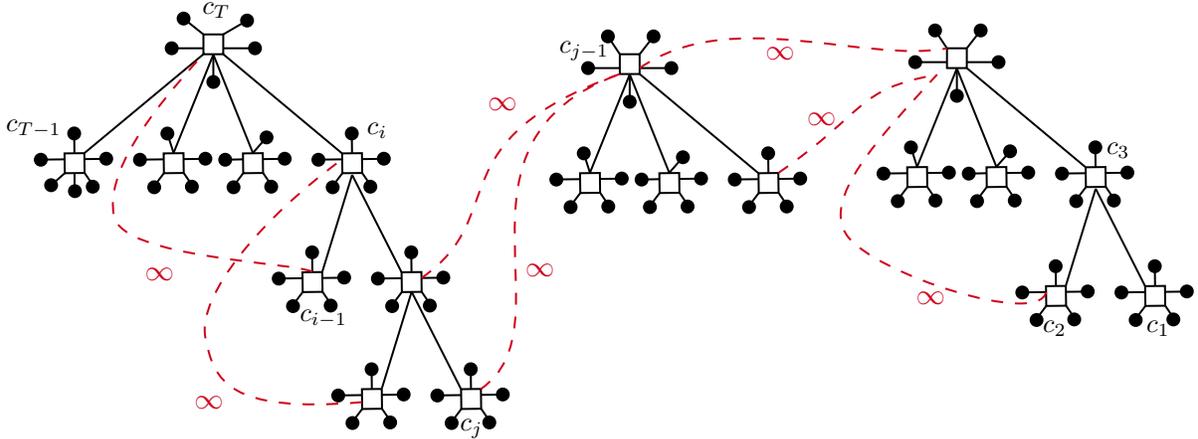

    \centering
    \include{img/estimation_forest}
    \vspace{-10mm}
    \caption{The structure of an estimation-forest constructed by \Cref{alg:nlogn-build-estimation-forest} in \Cref{sec:adaptive-algorithm}. White squares represent cluster centers, while black circles represent the other items of $[n]$. A new level in the forest is created when two nodes are compared and the estimate of their ratio is ``$\infty$''. If this happens twice consecutively (for the parent node and the children with smallest estimated weight), then a new tree is created. In the figure, we have $i < T-1$ and $j < i-1$.}
    \label{fig:estimation-forest}
\end{figure}

Observe that a cluster graph is not yet an estimation-forest. Indeed, there might be items in different clusters (but close in the ordering) whose ratio is constant. To obtain an estimation-forest, we add extra edges between some pairs of centers. We do so in an iterative way, starting from the center of the last cluster and moving backwards. \emph{A priori}, these multiplicative estimates can potentially be costly to obtain, since the ratio between the weights of distinct cluster centers could be arbitrarily large. In order to maintain a low query complexity, we employ a careful thresholding strategy. This ensures that we only require an accurate estimate of the ratio when this is not too large to make a significant difference in the MNL winning distributions. For instance, if the ratio between two items is greater than $\Omega({n\over \varepsilon})$, then it is safe to act as if the second item's weight is infinitely larger than the first, as this approximation only causes a $d_1$-error of magnitude $O(\varepsilon)$. When we find two clusters that are incomparable, we restart the iteration process from the last cluster that was comparable. It can be shown that this leads to an $(O(1),\epsilon)$-estimation-forest (see \Cref{thm:nlogn-build-estimation-forest}). We show in \Cref{fig:estimation-forest} a forest that can be produced by our algorithm.

In summary, our algorithm has \emph{three} phases. In the first phase, we construct a $(\Theta(1), \Theta(1), \Theta(\epsilon))$-cluster graph. In the second phase, we extend the cluster graph to a $(\Theta(1), \Theta(\epsilon))$-estimation-forest. Finally, in the third phase, we use the forest to recover an estimate of the MNL weights. A representation of the steps in our algorithm is in Figure~\ref{fig:overview-diagram}. The first two phases require at most $O(\frac{n \log n}{\epsilon^3})$ queries, while the last one does not make any further queries, yielding our main result:

\begin{restatable}[]{theorem}{ThmAdaptiveNlogN}\label{thm:adaptive-nlogn}
Choose any $\epsilon\in (0,1)$ and $\delta=n^{-c}$ for a constant $c>0$. There exists an adaptive randomized algorithm that, with probability at least $1- \delta$, makes $O\left(\frac{n \log n}{\epsilon^3}\right)$ $\maxsample$ queries and solves the MNL Learning Problem for $\dmul$-error on $[n]$ with accuracy parameter $\epsilon$. Moreover, the algorithm only queries pairs and runs in time proportional to the number of queries.
\end{restatable}

\begin{figure}
    \centering
    \tikzset{every picture/.style={line width=0.75pt}} %

\begin{tikzpicture}[x=0.75pt,y=0.75pt,yscale=-1,xscale=1]
\draw   (108,38) -- (179,38) -- (179,95) -- (108,95) -- cycle ;
\draw   (268,38) -- (348,38) -- (348,95) -- (268,95) -- cycle ;
\draw   (438,38) -- (524,38) -- (524,95) -- (438,95) -- cycle ;
\draw    (23,70) -- (106,70) ;
\draw [shift={(108,70)}, rotate = 180] [color={rgb, 255:red, 0; green, 0; blue, 0 }  ][line width=0.75]    (10.93,-3.29) .. controls (6.95,-1.4) and (3.31,-0.3) .. (0,0) .. controls (3.31,0.3) and (6.95,1.4) .. (10.93,3.29)   ;
\draw    (180,71) -- (263,71) ;
\draw [shift={(265,71)}, rotate = 180] [color={rgb, 255:red, 0; green, 0; blue, 0 }  ][line width=0.75]    (10.93,-3.29) .. controls (6.95,-1.4) and (3.31,-0.3) .. (0,0) .. controls (3.31,0.3) and (6.95,1.4) .. (10.93,3.29)   ;
\draw    (350,71) -- (433,71) ;
\draw [shift={(435,71)}, rotate = 180] [color={rgb, 255:red, 0; green, 0; blue, 0 }  ][line width=0.75]    (10.93,-3.29) .. controls (6.95,-1.4) and (3.31,-0.3) .. (0,0) .. controls (3.31,0.3) and (6.95,1.4) .. (10.93,3.29)   ;

\draw (120,50) node [anchor=north west][inner sep=0.75pt]  [font=\small] [align=left] {\begin{minipage}[lt]{31.79pt}\setlength\topsep{0pt}
\begin{center}
Cluster\\Graph
\end{center}

\end{minipage}};
\draw (275,50) node [anchor=north west][inner sep=0.75pt]   [align=left] {\begin{minipage}[lt]{48.11pt}\setlength\topsep{0pt}
\begin{center}
{\small Estimation }\\{\small Forest}
\end{center}

\end{minipage}};
\draw (452,50) node [anchor=north west][inner sep=0.75pt]  [font=\small] [align=left] {\begin{minipage}[lt]{38.42pt}\setlength\topsep{0pt}
\begin{center}
 \ MNL 
\end{center}
Estimate
\end{minipage}};
\draw (27,50) node [anchor=north west][inner sep=0.75pt]  [font=\small] [align=left] {Algorithm~\ref{alg:cluster-sort}};
\draw (184,51) node [anchor=north west][inner sep=0.75pt]  [font=\small] [align=left] {Algorithm~\ref{alg:nlogn-build-estimation-forest}};
\draw (354,51) node [anchor=north west][inner sep=0.75pt]  [font=\small] [align=left] {Algorithm~\ref{alg:generate-weights} };
\end{tikzpicture}
\vspace{-6mm}
    \caption{The structure of our algorithm to learn MNLs adaptively. The non-adaptive algorithm follows the same overall structure, but the first two steps are replaced by $\AlgQuicksortClustering$ (described in \Cref{prop:quicksort-cluster-graph}) and by \Cref{alg:build-forest-non-adaptive} respectively.}
    \label{fig:overview-diagram}
\end{figure}
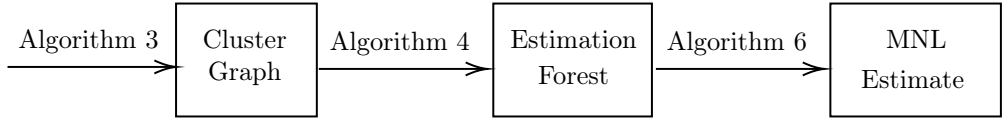

\subsection{Learning MNLs Adaptively in $d_1$-Error}\label{sec:overview-adaptive-d1}

In order to obtain a tight dependence in $\epsilon$ for $d_1$-error guarantees, we relax the notion of estimation-forest to the following:

\begin{definition}[$\epsilon$-weak-estimation-forest]\label{def:weak-estimation-forest}
Let $\epsilon\in(0,1)$. An $\epsilon$-\emph{weak-estimation-forest} for an MNL supported on $[n]$ with weights $\{w_1, \dots, w_n\}$ is a tuple
$
    \cF = (F, r, (C_1, \dots, C_T), (c_1, \dots, c_T)),
$
where $F = ([n],E)$ is an undirected forest, $r$ is a directed weighting on $F $, and $(C_1, \dots, C_T), (c_1, \dots , c_T)$ is an ordered clustering over $[n]$. For any $u\in[n]$ the following holds:
\begin{enumerate}[nosep]
    \item let $\mathcal{C}$ be the connected component of $u$, then for any $A \subseteq \{s\in \mathcal{C} \mid \gamma(s) \leq \gamma(u)\}$ it holds that:
    \[
        \sum_{s\in A} \left| \frac{w_s}{w_u} - r(P(s,u)) \right| \leq \epsilon \cdot \left(1 + \sum_{s\in A} \frac{w_s}{w_u}\right)
    \]
    \item for any $v\in[n]$ in a different connected component than $u$ and such that $\gamma(v)\leq \gamma(u)$, we have: 
    \[
        \sum_{\substack{s\in [n]\\ \cluster(s) \leq \cluster(v)}} \frac{w_s}{w_u} \leq \epsilon.
    \]
    Also, for any $u'$ (resp. $v'$) in the same connected component of $u$ (resp. $v$), it holds that $\cluster(u') > \cluster(v')$.
\end{enumerate}    
\end{definition}

We will show that an $\epsilon$-weak-estimation-forest is enough to obtain a $d_1$-error of $O(\epsilon)$ (\Cref{thm:weak-estimation-forest-reconstruction}), and that such weaker notion of estimation forest can be constructed more efficiently (\Cref{thm:build-weak-estimation-forest}). Therefore, we obtain the following:

\begin{restatable}[]{theorem}{ThmAdaptiveNlogNOptEps}\label{thm:adaptive-nlogn-epsilon-tight}
Choose any $\epsilon\in (0,1)$ and $\delta=n^{-c}$ for a constant $c>0$. There exists an adaptive randomized algorithm that, with probability at least $1- \delta$, makes $O\left(\frac{n \log n}{\epsilon^2}\right)$ $\maxsample$ queries and solves the MNL Learning Problem for $d_1$-error on $[n]$ with accuracy parameter $\epsilon$. Moreover, the algorithm only queries pairs and runs in time proportional to the number of queries.
\end{restatable}

\subsection{Learning MNLs Non-Adaptively}\label{sec:overview-non-adaptive}
We next present an algorithm to learn MNLs non-adaptively, i.e.,\ by making a single batch of queries. In order to do this we leverage the following reduction.

\begin{lemma}\label{lem:reduction-from-non-adaptive-to-adaptive}
    Given an adaptive algorithm for learning MNLs with the $\maxsample$ oracle that queries any pair of items at most $m$ times, one can construct a non-adaptive algorithm for the same problem that makes at most $m \binom{n}{2} = O(mn^2)$ queries.
\end{lemma}
\begin{proof}
    The non-adaptive algorithm queries each pair $m$ times and then simulates the adaptive algorithm by replacing each $\maxsample$ oracle call with a revealed response from the set of non-adaptive queries. 
\end{proof}
For constant $\epsilon$, the number of $\maxsample$ queries made to any pair $\{u,v\} \subseteq [n]$ of items by the adaptive algorithms described above could be as high as $\tilde{O}(n)$; this would naively yield an  $\tilde{O}({n^3})$-algorithm. Instead, we design an algorithm with query complexity $\tilde{O}(n^2)$. To accomplish this, we modify the adaptive algorithm to obtain a new (adaptive) algorithm that has a worse overall query complexity than the algorithms of \Cref{thm:adaptive-nlogn} and \Cref{thm:adaptive-nlogn-epsilon-tight}, but allows us to uniformly bound the number of queries made to each pair of items. In particular, in Section~\ref{sec:non-adaptive-algorithm} we show the following result.

\begin{restatable}[]{theorem}{ThmAdaptiveBalanced}\label{thm:adaptive-balanced}
Choose any $\epsilon,\delta \in (0,1)$. There exists an adaptive randomized algorithm that, with probability at least $1- \delta$, queries each pair at most $O\left(\frac{\log(n/\epsilon)\cdot \log(n/\delta)}{\epsilon^3}\right)$ times and solves the MNL Learning Problem for $\dmul$-error on $[n]$ with accuracy parameter $\epsilon$. 
\end{restatable}
To obtain this result, we use three new technical ingredients. First, we make use of a different algorithm to approximately order the items of the MNL (\Cref{sec:appendix-non-adaptive}). This algorithm, which is a straight-forward adaptation of the classical Quicksort algorithm, makes more queries than the previous one overall, but guarantees a uniform upper bound on the number of queries on each pair of items.
Second, we introduce a new algorithm to construct the estimation-forest. This algorithm only needs to make $O(|C_j| \log^2 n)$ comparisons between any pair $\{c_i,c_j\}$ of cluster centers (with $i > j$) whenever the ratio $\nicefrac{w_{c_i}}{w_{c_j}}$ is estimated.
Finally, we introduce a subroutine (Algorithm~\ref{alg:get-ratio-estimator}) that allows one to amortize the cost of estimating the ratio $\nicefrac{w_{c_i}}{w_{c_j}}$ among all the pairs of the form $\{c_i,s\}$, where $s$ belongs to the cluster $C_j$. This allows one to distribute the $O(|C_j| \log^2 n)$ cost nearly equally among all items in $C_j$, and hence to guarantee each pair is queried at most $O(\log^2n)$ times.

Combining \Cref{thm:adaptive-balanced} and \Cref{lem:reduction-from-non-adaptive-to-adaptive} yields: 
\begin{corollary}\label{cor:non-adaptive-algo}
Choose any $\epsilon,\delta \in (0,1)$. There exist a non-adaptive algorithm that, with probability at least $1-\delta$, makes at most ${O}\left({n^2\cdot \log(n/\epsilon)\cdot \log(n/\delta)\over \varepsilon^3}\right)$ queries and solves the MNL Learning Problem for $\dmul$-error on $[n]$ with accuracy parameter $\epsilon$.
\end{corollary}

\subsection{Lower Bounds}\label{sec:overview-lower-bounds}
We prove lower bounds that show that our adaptive algorithms are asymptotically optimal in both $n$ and $\epsilon$, and that our non-adaptive algorithm has nearly-optimal (at most a $\log n$ factor away from optimal) dependence on $n$. %

For learning MNLs with  adaptive queries to $\maxsample$, in \Cref{sec:adaptive-lower-bound} we show the following.
\begin{restatable}
{theorem}{AdaptiveLowerBound}\label{thm:adaptive-lower-bound}
    Any (possibly randomized and adaptive) algorithm that, given in input $\varepsilon,\delta\in (0,\frac{1}{32})$ and access to a $\maxsample$ oracle for any MNL $M$, outputs an MNL $\hat{M}$ satisfying:
    \[
        \Pr[d_{\infty}(M, \hat{M})\leq \varepsilon] \geq 1-\delta,
    \]
    must make $\Omega({n\over \varepsilon^2} \log {n\over \delta})$ queries in the worst case.
\end{restatable}

Note that this lower bound holds even for $d_\infty$, while our matching algorithm works for $d_1$. For the case of $\dmul$-error we can show a stronger dependency on $\epsilon$ which matches our upper bound for learning in $\dmul$-error:
\begin{restatable}{theorem}{AdaptiveLowerBoundDmul}\label{thm:adaptive-lower-bound-dmul}
    Any (possibly randomized and adaptive) algorithm that, given in input $\varepsilon,\delta\in (0,\frac{1}{128})$ and access to a $\maxsample$ oracle for any MNL $M$, outputs an MNL $\hat{M}$ satisfying:
    \[
        \Pr[\dmul(M, \hat{M})\leq \varepsilon] \geq 1-\delta,
    \]
    must make $\Omega({n\over \varepsilon^3} \log {\frac{n}{\delta}})$ queries in the worst case.
\end{restatable}
    
For the non-adaptive case, in Section~\ref{sec:non-adaptive-lower-bound} we show the following.
\begin{restatable}
{theorem}{NonAdaptiveLowerBound}\label{thm:non-adaptive-lower-bound}
Any (possibly randomized) non-adaptive algorithm that, given in input $\varepsilon\in (0,\frac{1}{32})$ and access to a $\maxsample$ oracle for any MNL $M$,  outputs an MNL $\hat{M}$ satisfying:
    \[
        \Pr[d_{\infty}(M, \hat{M})\leq \varepsilon] \geq {9\over 10},
    \]
    must make $\Omega({n^2 \over \varepsilon^2} \log {n})$ queries in the worst case.
\end{restatable}

The lower bounds we provide are based on reductions from the problem of approximately identifying the $n\over 2$ coins with the largest probability of \emph{heads} in a set of $n$ biased coins.

\section{Algorithmic Primitives}\label{sec:algorithmic-primitives}
    We will make use of the following two key subroutines, $\AlgCompare$ and $\AlgEstimateRatio$, and we will frequently refer to their guarantees provided below.

    \begin{algorithm}[h]
    \caption{\AlgCompare$(i,j,c,\varepsilon,\delta)$}
    \begin{algorithmic}[1]
    \State \textbf{Input:} Two items $i$ and $j$ of $[n]$, parameters $c, \varepsilon, \delta \in (0,1)$, and access to a $\maxsample$ oracle for an MNL $M$ supported on $[n]$.
    \State \textbf{Output:} Estimates $\hat{p}_{i}$ and $\hat{p}_{j}$ of $M_{\{i,j\}}(i)$ and $M_{\{i,j\}}(j)$ respectively.
    \State Make $m = {20\over c \cdot \varepsilon^2}\ln\left({6\over \delta}\right)$ queries to $\maxsample(\{i,j\})$ and let $m_i$ and $m_j$ be the number of queries that return $i$ and $j$ respectively.
    \State Let $\hat{p}_i = {m_i\over m}$ and $\hat{p}_j = {m_j\over m}$
    \If{$\hat{p}_i <c/2$}
        \State \textbf{return} $(0, \hat{p}_j)$
    \EndIf
    \If{$\hat{p}_j <c/2$}
        \State \textbf{return} $(\hat{p}_i,0)$
    \EndIf
    \State \textbf{return} $(\hat{p}_i, \hat{p}_j)$
    \end{algorithmic}
\end{algorithm}

In particular, a simple consequence of standard tail bounds is the following guarantee, which we prove for completeness in \Cref{sec:appendix-algo-primitives}.

\begin{restatable}[\AlgCompare{} guarantees]{lemma}{CompareGuarantees}\label{lem:compare-guarantees}
For any $c,\epsilon, \delta \in(0,1)$, \AlgCompare$(i,j,c,\varepsilon,\delta)$ makes $O\left({1\over c\varepsilon^2}\log{1\over \delta}\right)$ queries and outputs a pair $(\hat{p}_i,\hat{p}_j)$ that, with probability at least $1-\delta$ satisfies, for $k\in\{i,j\}$:

\begin{enumerate}[nosep]
    \item If $M_{\{i,j\}}(k) \leq c/4$, then $\hat{p}_k = 0$,
    \item If $M_{\{i,j\}}(k) \geq c$, then $\hat{p}_k \neq 0$,
    \item If $\hat{p}_k \neq 0$ then $(1-\epsilon) M_{\{i,j\}}(k) \leq \hat{p}_k \leq (1+\epsilon) M_{\{i,j\}}(k)$.
\end{enumerate}
\end{restatable}

This in turn implies the following lemma, which shows guarantees on the behavior of $\AlgEstimateRatio$. This is also proved in \Cref{sec:appendix-algo-primitives}.

\begin{algorithm}
\caption{\AlgEstimateRatio$(i,j,\alpha,\varepsilon,\delta)$}
\label{alg:estimage-ratio}
\begin{algorithmic}[1]
\State \textbf{Input:} Two items $i,j\in [n]$, parameters $\alpha, \varepsilon, \delta \in (0,1)$, and access to a $\maxsample$ oracle for an MNL $M$ supported on $[n]$ with weights $\{w_1, \dots, w_n\}$.
\State \textbf{Output:} An estimate $r(i,j)$ of the ratio of the weights $w_i/w_j$ in the MNL.
\State Let $c=\alpha/(\alpha+1)$.
\State $(\hat{p}_i,\hat{p}_j) \myassign $ \AlgCompare$(i,j,c, \nicefrac{\varepsilon}{3}, \delta)$
\If{$\hat{p}_i=0$}
    \State \textbf{return }$r(i,j) \myassign 0$
\EndIf
\If{$\hat{p}_j=0$}
    \State\textbf{return }$r(i,j) \myassign \infty$ 
\EndIf
\State \textbf{return }$r(i,j) \myassign {\hat{p}_i \over \hat{p}_j}$
\end{algorithmic}
\end{algorithm}

\begin{restatable}[\AlgEstimateRatio{} Guarantees]{lemma}{EstimateRatioGuarantees}\label{lem:weight-pair-ratio-estimate}
    Given two items $i$ and $j$ of $[n]$, and parameters $\alpha$, $\varepsilon$, and $\delta$ in $(0,\frac12]$, the algorithm $\AlgEstimateRatio(i,j,\alpha,\varepsilon, \delta)$ makes $O({1\over \alpha \varepsilon^2 } \log {1\over \delta})$ queries and produces an estimate $r(i,j)$ of the ratio $\frac{w_i}{w_j}$ that, with probability $1-\delta$, satisfies the following guarantees:
    \begin{enumerate}
        \item If $\frac{w_i}{w_j} \le \frac{\alpha}{3\alpha+4}$, then $r(i,j) = 0$.
        \item If $\frac{w_i}{w_j} \ge \frac{3\alpha+4}{\alpha}$, then $r(i,j) = \infty$.
        \item If $\frac{w_i}{w_j} \le \frac{1}{\alpha}$, then $r(i,j) \neq \infty$, and if $\frac{w_i}{w_j} \ge \alpha$ then $r(i,j) \neq 0$.
        \item Whenever $r(i,j) \not\in\{0, \infty\}$:
        \[
            r(i,j) \in (1\pm\varepsilon) {w_i\over w_j}\hspace{3mm}\text{ and }\hspace{3mm}
             \frac{1}{r(i,j)} \in (1\pm \varepsilon) {w_j\over w_i}.
        \]
    \end{enumerate}
\end{restatable}

\section{Adaptive Algorithm for $\dmul$ Error}\label{sec:adaptive-algorithm}
In this section we describe our adaptive algorithm for $\dmul$-error. The algorithm comprises three phases: (i) construct a cluster graph to approximately sort and group together items of similar weights (\Cref{sec:cluster-graph-from-sorting}), (ii) extend the cluster graph to an estimation-forest (\Cref{sec:contructing-the-forest-adaptive}), and (iii) extract an MNL from the estimation-forest (\Cref{sec:MNL-from-forest-adaptive}). 

\subsection{Constructing a Cluster Graph with $O\left({n\log n\over \varepsilon^2 }\right)$ Queries}\label{sec:cluster-graph-from-sorting}
In this section, we describe and analyze an algorithm to construct a cluster graph as defined in \Cref{def:cluster-graph}. This makes up the first phase of our adaptive algorithm to learn MNLs.

The first ingredient to obtain this result is an adaptive algorithm developed by \cite{fjopr18} to sort items in approximately increasing order of weight by querying a noisy pairwise comparison oracle. We consider the following definition.

\begin{restatable}[$\orderingerror$-ordering]{definition}{EpsilonOrdering}
    An \emph{$\orderingerror$-ordering} for an MNL $M$ supported on $[n]$ with weights $\{w_1, \dots , w_n\}$ is an ordering $(s_1, \dots , s_n)$ of the items of $[n]$ such that, for any pair $i,j$ with $i<j$: $(1-\orderingerror) w_{s_i} \le w_{s_j}.$
\end{restatable}

The following is a consequence of \citep[Theorem 9]{fjopr18}, where we boosted the success probability and made the runtime explicit (proof in \Cref{sec:computing-epsilon-orderings} provided for completeness).

\begin{restatable}{theorem}{EpsilonOrderingAlgorithm}\label{cor:falahatgar-for-MNLs}%
    Let $\orderingerror, \delta \in (0,1)$. 
    There is an algorithm that given access to a $\maxsample$ oracle for an MNL $M$ supported on $[n]$, with probability at least $1-\delta$, makes $O\left({n \log (n/\delta)\over \orderingerror^2} \cdot  \left(1+\frac{\log(1/\delta)}{\log n}\right)\right)$ queries and returns an $\orderingerror$-ordering of the items of $M$. Moreover, all the queries made by the algorithm are to slates of size two and the algorithm runs in time proportional to the number of queries.
\end{restatable}

Below, we introduce the main algorithm of this section: $\AlgClusterSort$ (Algorithm~\ref{alg:cluster-sort}). At a high level, the algorithm first computes an $O(1)$-ordering of the items as described above, and then it partitions them into clusters that are adjacent in this ordering. The center of a cluster is always chosen to be the first item in the cluster to appear in the $O(1)$-ordering. At each iteration the algorithm tries to add the $\ell$th item $s_\ell$ in the ordering to the cluster centered at some item $c_i$. If the ratio $\nicefrac{w_{s_\ell}}{w_{c_i}}$ is estimated to be too large, the algorithm simply starts a new cluster.

We begin by analyzing the algorithm's running time and query complexity.

\begin{algorithm}
\caption{$\AlgClusterSort(\alpha, \varepsilon,\delta)$}\label{alg:cluster-sort}
\begin{algorithmic}[1]
\State \textbf{Input:} Access to a $\maxsample$ oracle for an MNL $M$ supported on $[n]$ with weights $\{w_1, \dots, w_n\}$, parameters $\varepsilon\in (0, 1/7)$, $\alpha, \delta \in (0,1)$.

\State \textbf{Output:} A $(\frac{2}{\alpha}, \frac{1}{\alpha}, \varepsilon)$-cluster graph $\mathcal{G}=(F, r,(C_1, \dots, C_T), (c_1, \dots, c_T))$.

\State $\tau\myassign {3(1+\varepsilon)\over 2\alpha }$.
\State $\ClusterList \myassign \varnothing, \Centers \myassign \varnothing, E\myassign \varnothing$
\State Construct an $1\over3$-ordering $S =(s_1, \dots, s_n)$ for $M$ using the algorithm of \Cref{cor:falahatgar-for-MNLs} with error probability $\delta/2$.
\State $i\myassign 1$, $j\myassign 1$, $\ell \myassign 2$
\State $c_i \myassign s_1$
\While{$\ell \leq n$}
    \State $r(s_\ell,c_i) \myassign \Call{\AlgEstimateRatio}{s_{\ell},c_{i},\frac{2\alpha}{3},\varepsilon, {\delta \over 2n}}$
    
    \If{$r(s_\ell,c_i) > \tau$} 
        \State $C_i \myassign \{s_j, \dots, s_{\ell-1}\}$
        \State $\ClusterList \myassign \ClusterList\circ (C_i)$, $\Centers \myassign \Centers \circ (c_i)$
        \State Add to $E$ edges $\{s_a, c_i\}$ for $a \in \{j+1, \dots, \ell-1\}$ with weight $r(s_a,c_i)$
        \State $i \myassign i+1$, $c_i \myassign s_\ell$, $j\myassign \ell$
    \EndIf
    \State $\ell \myassign \ell+1 $
\EndWhile
\State $C_i \myassign \{s_j, \dots, s_{n}\}$
\State $\ClusterList \myassign \ClusterList \circ (C_i)$, $\Centers \myassign \Centers \circ (c_i)$
\State Add to $E$ edges $\{s_a, s_j\}$ for $a \in \{j+1, \dots, n\}$ with weight $r(s_a,s_j)$
\State \textbf{return } $(F=([n],E),r, \ClusterList, \Centers)$
\end{algorithmic}
\end{algorithm}

\begin{proposition}[Complexity of $\AlgClusterSort$]\label{lem:complexity-of-clustersort}
   With probability at least $1-\delta$, \AlgClusterSort$(\alpha, \epsilon,\delta)$ makes $O\left(n \cdot \log({n\over \delta}) \cdot \left(\frac{1}{\alpha\epsilon^2} + \frac{\log(1/\delta)}{\log n}\right) \right)$ queries to the $\maxsample$ oracle.%
\end{proposition}
\begin{proof}
     The first part of the algorithm (the $1\over 3$-ordering) makes $O\left({n\log (n/\delta)} \cdot \left(1 + \frac{\log (1/\delta)}{\log n}\right)\right)$ queries by \Cref{cor:falahatgar-for-MNLs}. The while loop is executed $O(n)$ times. At each execution we call \AlgEstimateRatio{} with parameters $2\alpha/3, \varepsilon$, and  $\delta/n$ and hence the number of $\maxsample$ queries per call is $O\left( {1\over \alpha \varepsilon^2} \log {n\over \delta} \right)$ by \Cref{lem:weight-pair-ratio-estimate}. Hence, the result follows. %
\end{proof}

We then show that the algorithm correctly computes the cluster graph.
\begin{theorem}[Guarantees for $\AlgClusterSort$]\label{lem:cluster-sort-produces-cluster-graph}
    Let $\alpha,\delta\in(0,1)$, and $\epsilon\in(0,1/7)$. Let $\mathcal{G}$ be the output of $\AlgClusterSort(\alpha, \epsilon,\delta)$. Then, with probability at least $1-\delta$, $\mathcal{G}$ is a $({2\over \alpha}, {1\over \alpha}, \varepsilon)$-cluster graph.
\end{theorem}
\begin{proof}

Let $\mathcal{G}=(F=([n],E),r,(C_1, \dots, C_T), (c_1, \dots, c_T))$ be the output of $\AlgClusterSort(\alpha, \epsilon,\delta)$.
We note that the ${1\over 3}$-ordering procedure succeeds with probability at least $1-{\delta\over 2}$, and that each call to \AlgEstimateRatio{} succeeds with probability at least $1 - {\delta \over 2n}$. By a union bound they all succeed with probability at least $1-\delta$.  For the rest of the proof, we will assume this holds.

We start by proving the first property of being a cluster graph (\Cref{def:cluster-graph}). When the cluster $C_i$ is created, its center $c_i$ is chosen to be the item in $C_i$ that comes first in the $1 \over 3$-ordering. In particular, for any $s_\ell \in C_i$ we have that: 
\begin{equation}\label{eq:cluster-sort-lb-ratio-same-cluster}
    {w_{s_\ell}\over w_{c_i}} \geq {2\over 3} \ge \frac{\alpha}{2}.
\end{equation}
On the other hand, for each item $s_\ell \in C_i$, the algorithm computes the estimate $r(s_{\ell},c_i)$ of ${w_{s_\ell} \over w_{c_i}}$ via $\AlgEstimateRatio$ and finds that $r(s_{\ell},c_i) \leq \tau$ (otherwise $\ell$ would have been placed in a different cluster). Since
$r(s_{\ell},c_i) \leq \tau$, there are two possibilities: either $r(s_{\ell},c_i)=0$ or $r(s_{\ell},c_i) \in (0,\tau]$. We consider the cases separately (and show that the first case cannot happen).
\begin{description}
    \item[Case 1.] Suppose $r(s_{\ell}, c_i)= 0$, then by the guarantees of $\AlgEstimateRatio$ it must have been the case that ${w_{s_\ell} \over w_{c_i}} < {2\alpha\over 3}$. But the $1\over 3$-ordering guarantees imply ${w_{s_\ell} \over w_{c_i}} \geq {2\over 3}$, giving a contradiction.
    
    \item[Case 2.] On the other hand, if $r(s_{\ell}, c_i)\in(0,\tau]$ then $\AlgEstimateRatio$ must have returned an accurate estimate for $w_{s_{\ell}} \over w_{c_i}$ and in particular we have:
    \begin{align*}
        {w_{s_\ell}\over w_{c_i}} \le \frac{r(s_\ell, c_i)}{1-\epsilon} \leq \frac{\tau}{1-\epsilon} \leq \frac{2}{\alpha},
    \end{align*}
    since $\varepsilon\leq {1\over 7}$, concluding the proof of the first property.
\end{description}

We now prove the second property of \Cref{def:cluster-graph}. Fix a choice of $i\in [T-1]$. Let $\ell^*>i$ be the smallest index such that $r(s_{\ell^*}, c_i) > \tau$. Note that, since there is at least one cluster following $C_i$, this choice $\ell^*$ must exist and we must have $\ell^*\le n$. By construction, we have:
\[
    C_{i+1} \cup \cdots \cup C_{T}= \{s_{\ell^*}, \dots , s_n\}.
\]

For all $\ell \geq \ell^*$, we have, by the definition of $1\over 3$-ordering:
\begin{equation}\label{eq:ordering-equation-1}
    {w_{s_\ell} \over w_{s_{\ell^*}}} \geq 1-{1\over 3} = {2\over 3}.
\end{equation}
We now have two possibilities: either $r(s_{\ell^*}, c_i) = \infty$ or $r(s_{\ell^*}, c_i) \in (0,\infty)$. Note that since $r(s_{\ell^*}, c_i) > \tau$, we have that $r(s_{\ell^*}, c_i)$ is non-zero. We consider the two cases separately.

\begin{description}
    \item[Case 1:] If $r(s_{\ell^*}, c_i) = \infty$ then due to the $\AlgEstimateRatio$ guarantees, we must have:
    \begin{equation}\label{eq:ratio-is-larger-than-alpha2}
        \frac{w_{s_{\ell^*}}}{w_{c_i}} > {3\over 2\alpha},
    \end{equation}
    and hence:
    \[
        \frac{w_{s_{\ell}}}{w_{c_i}} = \frac{w_{s_{\ell}}}{w_{s_{\ell^*}}}\cdot \frac{w_{s_{\ell^*}}}{w_{c_i}} \overset{\eqref{eq:ordering-equation-1},\eqref{eq:ratio-is-larger-than-alpha2}}{>}  {1\over \alpha}.
    \]
    
    \item[Case 2:] If  $r(s_{\ell^*}, c_i) \in (0,\infty)$ then the guarantees of $\AlgEstimateRatio$ imply that $r(s_{\ell^*}, c_i)$ must be a good approximation to $\frac{w_{s_{\ell^*}}}{w_{c_i}}$. We then have:
    \begin{align*}
    {w_{s_\ell} \over w_{c_i}} &= {w_{s_\ell} \over w_{s_{\ell^*}}}\cdot \frac{w_{s_{\ell^*}}}{ w_{c_i}}  \overset{\eqref{eq:ordering-equation-1}}{\geq}  {2\over 3} \cdot {w_{s_{\ell^*}} \over w_{c_i}} \geq {2\over 3(1+\varepsilon)} \cdot r(s_{\ell^*},c_i) > {2\over 3(1+\varepsilon)} \tau= {1\over \alpha}.
\end{align*}
\end{description}
This yields the second property in \Cref{def:cluster-graph}.

Finally, we prove the third property. For this, we simply argue that the estimates $r(s_{\ell},c_i)$ produced by the algorithm are never $0$ or $\infty$ for any item $s_\ell$ that is part of cluster $C_i$; the result will then follow from the guarantees of $\AlgEstimateRatio$. Note that if an item $s_{\ell}$ is placed in the same cluster as $c_i$ then $r(s_{\ell},c_i) \leq \tau < \infty$ and hence $r(s_{\ell},c_i) \neq \infty$. On the other hand, as we previously argued, by \Cref{eq:cluster-sort-lb-ratio-same-cluster} and properties of \AlgEstimateRatio{} it also holds $r(s_\ell,c_i)\neq 0$.
\end{proof}

\subsection{Constructing an Estimation-Forest with $O\left(\frac{n\log n}{\epsilon^3}\right)$ Queries}\label{sec:contructing-the-forest-adaptive}

We now present \Cref{alg:nlogn-build-estimation-forest} to construct an $(O(1), \epsilon)$-estimation-forest.  The algorithm starts by obtaining a $(\frac{2}{\alpha}, \frac{1}{\alpha},\Theta(\epsilon))$-cluster graph with ordered clusters $(C_1,\dots, C_T)$ and their corresponding centers $(c_1, \dots, c_T)$. Then, starting from $c_T$ the algorithm attempts to estimate the ratios $w_{c_i}/w_{c_j}$ for pairs $\{c_i,c_j\}$ of cluster centers in order to construct an estimation-forest from the cluster graph.  However, if this  ratio is very large, estimating it accurately will require too many queries. Luckily, in this case, we can pretend as if one center is infinitely heavier (in terms of its MNL weight) than the other; this will contribute to only a small error in the estimated $\maxsample$ distributions.  Intuitively, if $c_T$ wins with probability at least $1-\epsilon$ in the slate $\{c_T\} \cup C_{j} \cup \dots \cup C_1$, then it is not worth estimating the ratio between $w_{c_T}$ and $w_{c_j}$. Instead, we can just conclude that the ratio is very large and continue.

In order to get a low query complexity, and guarantee an accurate MNL estimate, it is crucial to design a good thresholding condition to establish when the weight of a cluster center is very large compared to another one for their ratio to be estimated accurately. A threshold too large would force the algorithm to estimate very large ratios, hence have high query complexity.  On the other hand, a threshold too small would cause a large error in the estimated $\maxsample$ distributions.

For each center $c_j$, we define a potential $Z_j$ that keeps into account both the sizes of the clusters $C_j,\dots, C_1$ and also their distance in the ordering: $Z_j = \sum_{i=1}^j\alpha^{j-i} |C_i|$. We will show that we can deem $c_T$ too large compared to $c_j$ if $\frac{w_{c_j}}{w_{c_T}} \leq \beta_j$ where $\beta_j = \Theta\left(\frac{\epsilon}{Z_j}\right)$.  
Even if $c_T$ is too large compared to $c_j$, it might still be that $c_{j+1}$ is comparable with $c_j$; therefore, we continue the process with $c_{j+1}$ instead of $c_T$. If $c_{j+1}$ is also deemed too large, then we start building a new tree starting from $c_j$. 

\begin{algorithm}[h]
\caption{\AlgBuildEstimationForest$(\alpha, \varepsilon,\delta)$}
\label{alg:nlogn-build-estimation-forest}
\begin{algorithmic}[1]
\State \textbf{Input:} Parameters $\alpha, \varepsilon, \delta \in (0,1)$, and access to a $\maxsample$ oracle for an MNL $M$ supported on $[n]$ with weights $\{w_1, \dots, w_n\}$.
\State \textbf{Output:} A $(5,\epsilon)$-estimation-forest $\mathcal{F}=(F, r, (C_1, \dots, C_T), (c_1, \dots, c_T))$.
\State $\epsilon_1 \myassign \frac{\epsilon}{10}$
\State $(F=([n],E), r, (C_1, \dots, C_T), (c_1, \dots, c_T)) \myassign \Call{\AlgClusterSort}{ \alpha, \epsilon_1, \frac{\delta}{3}}$ 
\State $Z_0 \myassign 0$
\For{$i=1, \dots, T$}
    \State $Z_i \myassign \alpha \cdot Z_{i-1} + |C_{i}|$
    \State $\beta_i \myassign \frac{\alpha^2 \cdot \epsilon}{8 \cdot Z_i}$
\EndFor

\State $i\myassign T$
\State $j\myassign T-1$
\While{$j > 0$}
    \State $r(c_{i},c_{j}) \myassign \max\left\{\Call{\AlgEstimateRatio}{c_{i}, c_{j}, \beta_{j}, \epsilon_1, \frac{\delta}{6n}}, \frac{1}{\alpha^{i-j}}\right\}$  
    \If{$r(c_{i},c_{j}) \neq \infty$}
        \State $E\myassign E\cup \left\{\{c_i, c_j\}\right\}$
        \State $j\myassign j-1$
    \ElsIf{$i = j+1$}
        \State $i\myassign j$
        \State $j\myassign j -1$
    \Else
        \State $i\myassign j+1$
    \EndIf
\EndWhile

\State \textbf{return} $(F=([n], E), r, (C_1, \dots, C_T), (c_1, \dots, c_T))$ 
\end{algorithmic}
\end{algorithm}

We now analyze \Cref{alg:nlogn-build-estimation-forest}.  (With a slight abuse of notation, the directed weighting $r$ produced by the algorithm is defined also for some pairs that are not in the set of edges.)  Our goal for this section is to prove the following result.

\begin{theorem}\label{thm:nlogn-build-estimation-forest}
Let $\epsilon, \alpha, \delta\in(0,1)$, then, with probability at least $1-\delta$, $\AlgBuildEstimationForest(\alpha, \epsilon, \delta)$ makes $O\left(n \log (\frac{n}{\delta}) \cdot \left(\frac{1}{(1-\alpha)\alpha^2\epsilon^3} + \frac{\log(1/\delta)}{\log n}\right)\right)$ $\maxsample$ queries and returns a $(5,\epsilon)$-estimation-forest.
\end{theorem}

We start by bounding the query complexity.

\begin{lemma}\label{lem:nlon-query-bound}
Let $\epsilon, \alpha, \delta\in(0,1)$, then, with probability at least $1-\delta/3$, $\AlgBuildEstimationForest(\alpha, \epsilon, \delta)$ makes $O\left(n \log (\frac{n}{\delta}) \cdot \left(\frac{1}{(1-\alpha)\alpha^2\epsilon^3} + \frac{\log(1/\delta)}{\log n}\right)\right)$ $\maxsample$ queries. 
\end{lemma}
\begin{proof}
By \Cref{lem:complexity-of-clustersort}, $\AlgClusterSort(\alpha,\epsilon_1, \delta/3)$ makes $O\left(n \log (\frac{n}{\delta}) \cdot \left(\frac{1}{\alpha\epsilon^2} + \frac{\log(1/\delta)}{\log n}\right)\right)$ queries with probability at least $1-\delta/3$. Observe that for each $j\in[T]$ there are at most two $i$'s for which we make a call to $\AlgEstimateRatio(c_i,c_j, \beta_j, \epsilon_1, \frac{\delta}{4n})$.
By \Cref{lem:weight-pair-ratio-estimate}, each call to $\AlgEstimateRatio$ costs:
\[
    O\left(\frac{1}{\beta_j \cdot \epsilon^2} \cdot \log\left(\frac{n}{\delta}\right)\right) = O\left(\frac{Z_j}{\alpha^2\epsilon^3}\cdot \log\left(\frac{n}{\delta}\right)\right).
\]
Moreover,  
\begin{align}\label{eq:Z-j-upper-bound-construction}
    \sum_{i=1}^T Z_i & = \sum_{i=1}^T \sum_{j=1}^i |C_j| \alpha^{i-j} \leq \sum_{i=1}^T |C_i| \left(\sum_{j=0}^T \alpha^j\right) \leq \sum_{i=1}^T |C_i| \left(\sum_{j=0}^\infty \alpha^j\right) = \sum_{i=1}^T \frac{|C_i|}{1-\alpha} = \frac{n}{1-\alpha},
\end{align}
where the first inequality follows by the fact that each $C_i$ gets summed up at most $T$ times and each time it is multiplied by a different value in $\{\alpha^0, \alpha^1, \dots, \alpha^T\}$. Thus, the total number of queries made by the algorithm after $\AlgClusterSort$ is 
\begin{align*}
    \sum_{i=1}^T 2\cdot O\left(\frac{Z_i \cdot \log \frac{n}{\delta}}{\alpha^2\epsilon^3}\right) & = O\left(\frac{\log \frac{n}{\delta}}{\alpha^2 \epsilon^3} \cdot \sum_{i=1}^T Z_i \right) \leq O\left(\frac{n\log \frac{n}{\delta}}{(1-\alpha)\alpha^2 \epsilon^3}\right).
\qedhere
\end{align*}
\end{proof}
We now move on to proving that the algorithm returns a $(5,\epsilon)$-estimation-forest. To do this, we need to show that the forest respects the four properties of \Cref{def:estimation-forest}. We will prove this in a series of lemmas. First, we show that items with close cluster indices end up in the same connected component. Recall that given an ordered partition $C_1, \dots, C_T$ of $[n]$, $\gamma(i)$ for $i \in [n]$ is the unique value $j\in [T]$ such that $i \in C_j$. We have the following result, which entails the fourth and part of the third property of \Cref{def:estimation-forest}.

\begin{lemma}\label{lem:nlogn-graph-structure}
For $\epsilon, \alpha, \delta\in(0,1)$, let $\mathcal{F}=(F=([n], E), r, (C_1, \dots, C_T), (c_1, \dots, c_T))$ be the output of $\AlgBuildEstimationForest(\alpha, \epsilon, \delta)$. Then, if $u,v\in[n]$ are in different connected components and $\cluster(u) > \cluster(v)$, we have that, for any $u'$ (resp. $v'$) in the same connected component of $u$ (resp. $v$), it holds that $\cluster(u') > \cluster(v')$. Moreover, if $\cluster(u)=\cluster(v)$, then $d(u,v)\leq 2$,
\end{lemma}
\begin{proof}
Let $\mathcal{C}$ be a connected component of $F$. Let $\psi$ (resp.\ $\phi$) be the maximum (resp.\ minimum) index such that $c_{\psi}\in \mathcal{C}$ (resp.\ $c_{\phi}\in \mathcal{C}$). By construction of $F$, we have $\mathcal{C} = \bigcup_{i=\phi}^\psi C_i$. This implies the first part of the lemma. Moreover, since the edges output by \AlgClusterSort{} form a star on each cluster $C_i$, for any pair of vertices $u$ and $v$ such that $\cluster(u)=\cluster(v)$, we have $d(u,v)\leq 2$. 
\end{proof}
We now prove that short paths produce good estimates of the weight ratios; this establishes the first requirement in \Cref{def:estimation-forest}.

\begin{lemma}\label{lem:nlogn-close-vertices}
For $\epsilon, \alpha, \delta\in(0,1)$, set $\epsilon_1=\frac{\epsilon}{10}$ and let $\mathcal{F}=(F=([n], E), r, (C_1, \dots, C_T), (c_1, \dots, c_T))$ be the output of $\AlgBuildEstimationForest(\alpha, \epsilon, \delta)$, and suppose that each call to $\AlgEstimateRatio$ as well as the call to $\AlgClusterSort$ is successful. Then, for any integer $t\geq 1$ and any $u,v\in [n]$ such that $d(u,v)\leq t$, 
\[
(1-\epsilon_1)^t \cdot \frac{w_u}{w_v} \leq r(P(u,v)) \leq (1+\epsilon_1)^t \cdot \frac{w_u}{w_v}. 
\]
In particular, if $d(u,v)\leq 5$, $r(P(u,v)) \in (1\pm \epsilon) \cdot \frac{w_u}{w_v}$.
\end{lemma}
\begin{proof}
By \Cref{lem:cluster-sort-produces-cluster-graph}, $\AlgClusterSort(\alpha, \epsilon_1, \delta/2)$ returns a $(\frac{2}{\alpha}, \frac{1}{\alpha}, \epsilon_1)$-cluster graph. Consider $\{c_i, c_j\}\in E$ with $i>j$. By \Cref{def:cluster-graph}, it holds that $w_{c_i} \geq w_{c_j}$. Moreover, since $\{c_i,c_j\}\in E$, $r(c_i,c_j)\neq \infty$, hence by the guarantees of $\AlgEstimateRatio$ (\Cref{lem:weight-pair-ratio-estimate}) we must have that the value $\rho:=\AlgEstimateRatio(c_i,c_j,\beta_j,\varepsilon_1, {\delta \over 4n}) \neq 0$ %
and specifically:
\begin{equation}\label{eq:rho-good-estimate-close-vertices}
    \rho \in (1\pm \epsilon_1) \frac{w_{c_i}}{w_{c_j}} \quad\text{and}\quad  {1\over \rho} \in (1\pm \epsilon_1)\frac{w_{c_j}}{w_{c_i}}.
\end{equation} Moreover, by the guarantees of a $(\frac{2}{\alpha}, \frac{1}{\alpha}, \epsilon_1)$-cluster graph, 
\begin{equation}\label{eq:ratio-guarantees-compose}
    \frac{w_{c_i}}{w_{c_j}} = {w_{c_i}\over w_{c_{i-1}}}\cdot {w_{c_{i-1}}\over w_{c_{i-2}}}\cdots {w_{c_{j+1}}\over w_{c_{j}}} \geq \frac{1}{\alpha^{i-j}}.
\end{equation}
Thus,
\[
r(c_i,c_j) = \max\left\{\rho, \frac{1}{\alpha^{i-j}}\right\} \overset{\eqref{eq:rho-good-estimate-close-vertices},\eqref{eq:ratio-guarantees-compose}}{\leq} \max\left\{(1+\epsilon_1)\frac{w_{c_i}}{w_{c_j}}, \frac{w_{c_i}}{w_{c_j}}\right\} \leq (1+\epsilon_1)\frac{w_{c_i}}{w_{c_j}},
\]
and clearly $r(c_i,c_j) \geq \rho \geq  (1-\epsilon_1)\frac{w_{c_i}}{w_{c_j}}$. Similarly, 

\[
r(c_j,c_i) = \frac{1}{\max\{\rho, \frac{1}{\alpha^{i-j}}\}} = \min\left\{\frac{1}{\rho}, \alpha^{i-j}\right\} \overset{\eqref{eq:rho-good-estimate-close-vertices},\eqref{eq:ratio-guarantees-compose}}{\geq} (1-\epsilon_1)\frac{w_{c_j}}{w_{c_i}},
\]
and clearly $r(c_j,c_i) \leq {1\over \rho} \leq (1+\epsilon_1)\frac{w_{c_j}}{w_{c_i}}$. Note that for each edge $\{u,v\}\in E$, where $\{u,v\}\not\subseteq\{c_1, \dots , c_T\}$, the value of $r(u,v)$ was computed during the construction of the cluster graph. Therefore $r(u,v) \in (1\pm \epsilon_1) \frac{w_u}{w_v}$ and $r(v,u) \in (1\pm \epsilon_1) \frac{w_v}{w_u}$ by the definition of cluster graph. Hence these guarantees hold for every pair $\{u,v\}\in E$. 

Consider now any $u,v \in [n]$ such that $d(u,v)\leq t$. Let $P(u,v)=a_1, \dots, a_{t+1}$. We have,
\begin{align*}
    r(P(u,v)) &= \prod_{i=1}^{t} r(a_i,a_{i+1}) \leq (1+\epsilon_1)^{t} \cdot \prod_{i=1}^{t} \frac{w_{a_i}}{w_{a_{i+1}}} = (1+\epsilon_1)^{t} \cdot \frac{w_{a_1}}{w_{a_{t+1}}} = (1+\epsilon_1)^{t} \cdot \frac{w_{u}}{w_{v}}.
\end{align*}
Note that for $t\leq 5$, $2t\cdot \epsilon_1 \in (0,1)$, thus, by using that $(1+a)^b \leq 1+2ab$ for $a\in[0,1], b\geq0, 2ab\in(0,1)$, we obtain: 
\begin{align*}
    r(P(u,v))&\leq (1+2\cdot t\cdot \epsilon_1) \frac{w_{u}}{w_{v}} \leq (1+10\cdot \epsilon_1) \frac{w_{u}}{w_{v}} \leq (1+\epsilon) \frac{w_{u}}{w_{v}}.
\end{align*}
Similarly,
\begin{align*}
    r(P(u,v)) &= \prod_{i=1}^{t} r(a_i,a_{i+1}) \geq (1-\epsilon_1)^{t} \cdot \frac{w_{u}}{w_{v}} \geq (1-2\cdot t \cdot \epsilon_1) \frac{w_{u}}{w_{v}} \geq (1-\epsilon) \frac{w_{u}}{w_{v}},
\end{align*}
where the last two inequalities hold for $t\leq 5$. In particular, we used that $(1-a)^b \geq 1-2ab$, for all 
$a\in[0,1]$, $b\geq1$, $2ab\in(0,1)$.
\end{proof}

\begin{figure}
    \centering
    \tikzset{every picture/.style={line width=0.75pt}} %

\begin{tikzpicture}[x=0.75pt,y=0.75pt,yscale=-1,xscale=1]
\draw   (72,42.94) -- (82,42.94) -- (82,52.94) -- (72,52.94) -- cycle ;
\draw   (170,153) -- (180,153) -- (180,163) -- (170,163) -- cycle ;
\draw   (45,123) -- (55,123) -- (55,133) -- (45,133) -- cycle ;
\draw   (206,221) -- (216,221) -- (216,231) -- (206,231) -- cycle ;
\draw  [fill={rgb, 255:red, 0; green, 0; blue, 0 }  ,fill opacity=1 ] (40.82,143.9) .. controls (40.84,145.56) and (39.51,146.92) .. (37.86,146.94) .. controls (36.2,146.96) and (34.84,145.63) .. (34.82,143.98) .. controls (34.8,142.32) and (36.13,140.96) .. (37.78,140.94) .. controls (39.44,140.92) and (40.8,142.25) .. (40.82,143.9) -- cycle ;
\draw  [fill={rgb, 255:red, 0; green, 0; blue, 0 }  ,fill opacity=1 ] (231.82,243.9) .. controls (231.84,245.56) and (230.51,246.92) .. (228.86,246.94) .. controls (227.2,246.96) and (225.84,245.63) .. (225.82,243.98) .. controls (225.8,242.32) and (227.13,240.96) .. (228.78,240.94) .. controls (230.44,240.92) and (231.8,242.25) .. (231.82,243.9) -- cycle ;
\draw   (192,183) -- (202,183) -- (202,193) -- (192,193) -- cycle ;
\draw   (101,123) -- (111,123) -- (111,133) -- (101,133) -- cycle ;
\draw   (72,83.94) -- (82,83.94) -- (82,93.94) -- (72,93.94) -- cycle ;
\draw    (37.78,142.94) -- (45,133) ;
\draw    (216,231) -- (228.82,243.94) ;
\draw    (197,193.33) -- (211,220.33) ;
\draw    (180,163) -- (192,183) ;
\draw [color={rgb, 255:red, 208; green, 2; blue, 27 }  ,draw opacity=1 ] [dash pattern={on 4.5pt off 4.5pt}]  (164,227.33) .. controls (144,230.33) and (138,221.33) .. (135,202.33) .. controls (132,183.33) and (151,171.33) .. (170,163) ;
\draw    (111,133) .. controls (113.11,131.96) and (114.69,132.5) .. (115.74,134.61) .. controls (116.78,136.72) and (118.36,137.26) .. (120.47,136.21) .. controls (122.58,135.17) and (124.16,135.71) .. (125.21,137.82) .. controls (126.25,139.93) and (127.83,140.47) .. (129.94,139.42) .. controls (132.05,138.38) and (133.63,138.92) .. (134.68,141.03) .. controls (135.72,143.14) and (137.3,143.68) .. (139.41,142.63) .. controls (141.52,141.59) and (143.1,142.13) .. (144.15,144.24) .. controls (145.19,146.35) and (146.77,146.89) .. (148.88,145.84) .. controls (150.99,144.8) and (152.57,145.34) .. (153.62,147.45) .. controls (154.66,149.56) and (156.24,150.1) .. (158.35,149.05) .. controls (160.46,148.01) and (162.04,148.55) .. (163.09,150.66) .. controls (164.13,152.77) and (165.71,153.31) .. (167.82,152.26) -- (170,153) -- (170,153) ;
\draw    (77,53.33) .. controls (78.67,55) and (78.67,56.66) .. (77,58.33) .. controls (75.33,60) and (75.33,61.66) .. (77,63.33) .. controls (78.67,65) and (78.67,66.66) .. (77,68.33) .. controls (75.33,70) and (75.33,71.66) .. (77,73.33) .. controls (78.67,75) and (78.67,76.66) .. (77,78.33) .. controls (75.33,80) and (75.33,81.66) .. (77,83.33) -- (77,83.33) ;
\draw    (50,122.33) -- (77,94.33) ;
\draw    (77,94.33) -- (101,123) ;
\draw   (166,222) -- (176,222) -- (176,232) -- (166,232) -- cycle ;
\draw    (192,193) -- (176,222) ;

\draw (67,27.94) node [anchor=north west][inner sep=0.75pt]  [font=\small]  {$c_{R}$};
\draw (172,140) node [anchor=north west][inner sep=0.75pt]  [font=\small]  {$c_{z}$};
\draw (20,108) node [anchor=north west][inner sep=0.75pt]  [font=\small]  {$c_{\gamma( u)}$};
\draw (219,214) node [anchor=north west][inner sep=0.75pt]  [font=\small]  {$c_{\gamma ( v)}$};
\draw (26.82,146.98) node [anchor=north west][inner sep=0.75pt]  [font=\small]  {$u$};
\draw (232.07,245) node [anchor=north west][inner sep=0.75pt]  [font=\small]  {$v$};
\draw (102,110) node [anchor=north west][inner sep=0.75pt]  [font=\small]  {$c_{x}$};
\draw (112,182.79) node [anchor=north west][inner sep=0.75pt]  [color={rgb, 255:red, 208; green, 2; blue, 27 }  ,opacity=1 ]  {$\infty $};
\draw (161,208) node [anchor=north west][inner sep=0.75pt]  [font=\small]  {$c_{k}$};

\end{tikzpicture}
    \vspace{-10mm}
    \caption{Highlight of some vertices of the estimation-forest computed by \Cref{alg:nlogn-build-estimation-forest} that are used in the proofs of \Cref{lem:nlogn-distant-vertices-real} and \Cref{lem:nlogn-distant-vertices-estimates}. White squares are cluster centers and black dots are items that are not cluster centers. Curly edges represent paths of length $\geq 0$. The dashed edge is not present in the tree but indicates that \Cref{alg:nlogn-build-estimation-forest} observed $r({c_z}, {c_k})=\infty$. Vertex $c_x$ is referenced only in the proof of \Cref{lem:nlogn-distant-vertices-estimates}. Note that in the figure we assume that $u$ and $v$ are not centers, but it might also be $u=c_{\cluster(u)}$ or $v=c_{\cluster(v)}$. Moreover, we assume that $c_{\cluster(u)}$ and $c_x$ are siblings but it might also be $c_{\cluster(u)}=c_x$ (similarly for $c_k$ and $c_{\cluster(v)}$). In particular, it holds that $R\geq \cluster(u) \geq x \geq z > k \geq \cluster(v)$, and also $\cluster(u) > z$.}
    \label{fig:proof-tree}
\end{figure}
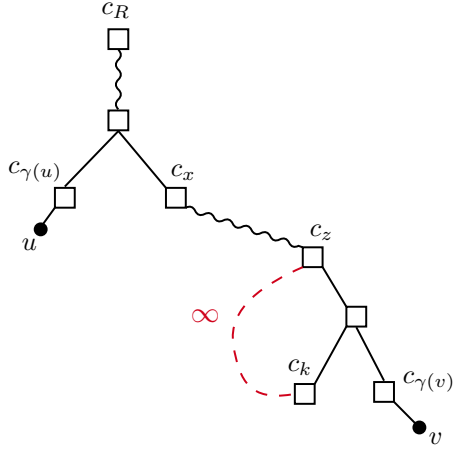

We now prove that items far away in the forest have negligible ratios, concluding the proof of the third point and part of the second point of \Cref{def:estimation-forest}.

\begin{lemma}\label{lem:nlogn-distant-vertices-real}
For $\epsilon, \alpha, \delta\in(0,1)$, let $\mathcal{F}=(F=([n], E), r, (C_1, \dots, C_T), (c_1, \dots, c_T))$ be the output of $\AlgBuildEstimationForest(\alpha, \epsilon, \delta)$, and suppose that each call to $\AlgEstimateRatio$ as well as the call to $\AlgClusterSort$ is successful. Then, for any $u,v\in [n]$ such that $d(u,v)>5$ and $\cluster(u)>\cluster(v)$, it holds that $\sum_{s\in H_v} \frac{w_s}{w_u} \leq \epsilon$, where $H_v = \{s\in [n] \mid \cluster(s) \leq \cluster(v)\}$.
\end{lemma}
\begin{proof}
Since $d(u,v)\geq 6$, it must be the case that $d(c_{\cluster(u)}, c_{\cluster(v)}) \geq 4$. Consider first the case where $u$ and $v$ are in the same tree. Let $R$ be the largest number in $[T]$ such that $c_{R}$ is in the same connected component as $u$ and $v$. We root the tree so that $c_{R}$ is its root (see \Cref{fig:proof-tree} for reference). Let $c_z$ be the ancestor of $c_{\cluster(v)}$ at distance 2 from $c_{\cluster(v)}$ in $F$. Note that $c_z$ exists and is on a level of the tree no smaller than the level of $c_{\cluster(u)}$ since $d(c_{\cluster(u)}, c_{\cluster(v)}) \geq 4$. Therefore, by construction of the tree, it must also hold $\cluster(u) > z$ and by the definition of $(\frac{2}{\alpha}, \frac{1}{\alpha},\epsilon)$-cluster graph, 
\begin{equation}\label{eq:nlogn-close-vertices-ccu-cz-relation}
    w_{c_{\cluster(u)}} \geq w_{c_z}.
\end{equation}
Observe that, during the construction of this tree, $c_z$ must have been compared with a sibling of $c_{\cluster(v)}$ (or with $c_{\cluster(v)}$ itself) and the result of the estimation must have been $\infty$, which caused the creation of a new level in the tree. Formally, there must exists $c_k$, with $z>k\geq \cluster(v)$, such that $r(c_z,c_k) =\infty$. In particular, the result of $\AlgEstimateRatio(c_z, c_k, \beta_k, \epsilon_1, \frac{\delta}{4n})$ must have returned $\infty$. Thus, by \Cref{lem:weight-pair-ratio-estimate}, 
\begin{equation}\label{eq:nlogn-close-vertices-cz-ck-relation}
    w_{c_z} \geq \frac{w_{c_k}}{\beta_k}.
\end{equation}
Let $H_{c_k} = \bigcup_{\ell=1}^{k} C_\ell$. Note that $H_v = \bigcup_{\ell=1}^{\cluster(v)} C_\ell \subseteq H_{c_k}$. Consider any $s\in H_{c_k}$. By the definition of $(\frac{2}{\alpha}, \frac{1}{\alpha},\epsilon)$-cluster graph, we have
\begin{align}\label{eq:nlogn-close-wu-ws}
    w_{u} &\geq \frac{\alpha}{2} \cdot w_{c_{\cluster(u)}} \overset{\eqref{eq:nlogn-close-vertices-ccu-cz-relation}}{\geq} \frac{\alpha}{2} \cdot w_{c_z} \overset{\eqref{eq:nlogn-close-vertices-cz-ck-relation}}{\geq} \frac{\alpha}{2 \cdot \beta_k} \cdot w_{c_k} \geq \frac{\alpha}{2\beta_k} \cdot \frac{1}{\alpha^{k-\cluster(s)}} \cdot w_{c_{\cluster(s)}} \nonumber\\
    & \geq \frac{\alpha^2}{4\cdot \beta_k \cdot \alpha^{k-\cluster(s)}} \cdot w_{s} \geq \frac{Z_k}{\epsilon \cdot \alpha^{k-\cluster(s)}} \cdot w_s, 
\end{align}
where the last inequality is by the definition of $\beta_k$. Thus,
\begin{align*}
    \sum_{s\in H_{c_k}} \frac{w_{s}}{w_u} &\overset{\eqref{eq:nlogn-close-wu-ws}}{\leq} \sum_{s\in H_{c_k}} \frac{\alpha^{k-\cluster(s)} \cdot \epsilon}{Z_{k}} = \sum_{\ell=1}^k |C_\ell| \cdot \frac{\alpha^{k-\ell} \cdot \epsilon}{Z_k} = \frac{\epsilon}{Z_k} \sum_{\ell=1}^k |C_\ell| \cdot\alpha^{k-\ell} = \frac{\epsilon}{Z_k} \cdot Z_k = \epsilon.
\end{align*}
Therefore, by $H_v\subseteq H_{c_k}$, $\sum_{s\in H_v} \frac{w_s}{w_u} \leq \epsilon$. This concludes the proof for the case where $u$ and $v$ are in the same tree.

Consider now the case where $u$ and $v$ are in different connected components of $F$. Let $z$ be the smallest integer such that $c_z$ is in the same connected component as $u$. Then, since $\cluster(u)\geq z$, we have $w_{c_{\cluster(u)}} \geq w_{c_z}$. Note also that it must be $z > \cluster(v)$. Moreover, by construction, $r(c_z,c_{z-1}) =\infty$, otherwise $c_{z-1}$ would be in the same connected component as $u$. Thus, by the guarantees on $\AlgEstimateRatio$ in \Cref{lem:weight-pair-ratio-estimate}, $w_{c_z} \geq \frac{1}{\beta_{z-1}} \cdot w_{c_{z-1}}$. Let $H_{c_{z-1}} = \bigcup_{\ell=1}^{z-1} C_\ell$, and note that $H_v \subseteq H_{c_{z-1}}$ given that $z-1 \geq \cluster(v)$. Let $s\in H_{c_{z-1}}$.  Similarly to the computation of \Cref{eq:nlogn-close-wu-ws}:
\begin{align*}    
    w_u &\geq \frac{\alpha}{2} \cdot w_{c_{\cluster(u)}} \geq \frac{\alpha}{2} \cdot w_{c_{z}} \geq \frac{\alpha}{2\cdot \beta_{z-1}} \cdot w_{c_{z-1}} \geq \frac{\alpha}{2\cdot \beta_{z-1}\cdot \alpha^{z-1-\cluster(s)}} \cdot w_{c_{\cluster(s)}} \\
    &\geq \frac{\alpha^2}{4\cdot \beta_{z-1} \cdot \alpha^{z-1-\cluster(s)}} \cdot w_{s} \geq \frac{Z_{z-1}}{\epsilon\cdot \alpha^{z-1-\cluster(s)}} \cdot w_s.
\end{align*}
Therefore, similar to previous calculations,
\[
\sum_{s\in H_v} \frac{w_s}{w_u} \leq \sum_{s\in H_{c_{z-1}}} \frac{w_s}{w_u} \leq \frac{\epsilon}{Z_{z-1}} \sum_{\ell=1}^{z-1} |C_\ell| \cdot \alpha^{z-1-\ell} = \epsilon,
\]
where the first inequality follows by $H_v\subseteq H_{c_{z-1}}$. 
\end{proof}
We are only left with showing that the estimates are well-behaved also along long paths. Before proving this, we show an auxiliary property that applies specifically to paths going from a cluster center to its descendants. 

\begin{lemma}\label{lem:nlogn-ancestors-property}
For $\epsilon, \alpha, \delta\in(0,1)$, let $\mathcal{F}=(F=([n], E), r, (C_1, \dots, C_T), (c_1, \dots, c_T))$ be the output of $\AlgBuildEstimationForest(\alpha, \epsilon, \delta)$, and suppose that each call to $\AlgEstimateRatio$ as well as the call to $\AlgClusterSort$ is successful. Consider any connected component $\mathcal{C}$ in forest $F$ and let $i$ be the largest index such that $c_i \in \mathcal{C}$. If $c_x$ is an ancestor of $c_y$ in the tree $\mathcal{C}$ rooted at $c_i$, then $r(P(c_x, c_y)) \geq \frac{1}{\alpha^{x-y}}$.
\end{lemma}
\begin{proof}
By construction, we must have $x\geq y$. Let $P(c_x, c_y):=c_x = c_{i_1}, \dots, c_{i_k} = c_y$ be the unique path from $c_x$ to $c_y$. Note that, by construction, for each $j\in[k-1]$, $r(c_{i_j}, c_{i_{j+1}}) \geq \frac{1}{\alpha^{i_j - i_{j+1}}}$. Thus,
\[
    r(P(c_x, c_y)) = \prod_{j=1}^{k-1} r(c_{i_j}, c_{i_{j+1}}) \geq \prod_{j=1}^{k-1} \frac{1}{\alpha^{i_j - i_{j+1}}} = \frac{1}{\alpha^{i_1 - i_k}} = \frac{1}{\alpha^{x-y}}. \qedhere
\]
\end{proof}

We are now ready to show that the estimates are well-behaved on long paths. This concludes the proof of the second point of \Cref{def:estimation-forest} and hence all four properties of \Cref{def:estimation-forest} are proved. 

\begin{lemma}\label{lem:nlogn-distant-vertices-estimates}
For $\epsilon, \alpha, \delta\in(0,1)$, let $\mathcal{F}=(F=([n], E), r, (C_1, \dots, C_T), (c_1, \dots, c_T))$ be the output of $\AlgBuildEstimationForest(\alpha, \epsilon, \delta)$, and suppose that each call to $\AlgEstimateRatio$ as well as the call to $\AlgClusterSort$ is successful. Suppose that $u,v\in[n]$ are in the same connected component $\mathcal{C}$, and $\cluster(u)>\cluster(v)$ and $d(u,v)>5$. Then, $\sum_{s\in K_v} r(P(s,u)) \leq \epsilon$, where $K_v = \{s\in \mathcal{C} \mid \cluster(s)\leq \cluster(v)\}$.
\end{lemma}
\begin{proof}
Let $R\in[T]$ be the maximum index such that $c_R \in \mathcal{C}$, and consider the tree $\mathcal{C}$ rooted at $c_R$ (see \Cref{fig:proof-tree} for reference). Since $d(u,v)\geq 6$, it must be $d(c_{\cluster(u)}, c_{\cluster(v)}) \geq 4$. Similarly to \Cref{lem:nlogn-distant-vertices-real}, let $c_z$ be the ancestor of $c_{\cluster(v)}$ at distance 2 from $c_{\cluster(v)}$ in $F$. Note that $c_z$ exists and $\cluster(u) \geq z$ since $d(c_{\cluster(u)},c_z)\geq 2$. Let $c_x$ be the sibling of $c_{\cluster(u)}$ with minimum cluster index (possibly, $c_x=c_{\cluster(u)}$ or it might also be $c_x=c_z$; but it surely holds $x\leq \cluster(u)$). Note that $d(u,c_x)\leq 3$ and $c_x$ is an ancestor of $c_z$ (hence, $x\geq z$). Thus, by \Cref{lem:nlogn-close-vertices} and \Cref{lem:nlogn-ancestors-property}, 
\begin{align}\label{eq:nlogn-distant-vertices-est-P-u-cz}
    r(P(u,c_z)) &= r(P(u,c_x)) \cdot r(P(c_x, c_z)) \geq (1-\epsilon_1)^3 \frac{w_{u}}{w_{c_x}} \cdot \frac{1}{\alpha^{x-z}} \nonumber\\
    & \geq (1-\epsilon_1)^3 \frac{w_{u}}{w_{c_{\cluster(u)}}} \cdot \frac{w_{c_{\cluster(u)}}}{w_{c_x}} \geq (1-\epsilon_1)^3 \cdot \frac{\alpha}{2} \cdot \frac{1}{\alpha^{\cluster(u)-x}}\geq (1-\epsilon_1)^3 \cdot \frac{\alpha}{2}.
\end{align}

Again similarly to \Cref{lem:nlogn-distant-vertices-real}, there must exists $c_k$, with $z>k\geq \cluster(v)$, such that $r(c_z,c_k) =\infty$. Note that $d(c_z, c_k)=2$ and, in particular, $c_k$ and $c_{\cluster(v)}$ are siblings (or $c_k=c_{\cluster(v)}$). Let $K_{c_k} = \mathcal{C} \cap \left(\bigcup_{\ell=1}^{k} C_\ell \right)$. Note that $K_v = \mathcal{C} \cap \left(\bigcup_{\ell=1}^{\cluster(v)} C_\ell \right) \subseteq K_{c_k}$. Consider any $s\in K_{c_k}$, by the same argument as in \Cref{lem:nlogn-distant-vertices-real}, 
\begin{equation}\label{eq:nlogn-distant-vertices-est-cz-s}
    w_{c_z} \geq \frac{\alpha}{2\cdot \beta_k \cdot \alpha^{k-\cluster(s)}} \cdot w_s.
\end{equation} 
Consider now $A=\{s\in K_{c_k} \mid c_{\cluster(s)} \text{ is a sibling of } c_k\}$. Note that any vertex in $A$ is at distance either 2 or 3 from $c_z$. Thus, for any $a\in A$, by \Cref{lem:nlogn-close-vertices}, 
\begin{equation}\label{eq:nlogn-distant-est-P-cz-a}
    r(P(c_z, a)) \geq (1-\epsilon_1)^3 \cdot \frac{w_{c_z}}{w_a} \overset{\eqref{eq:nlogn-distant-vertices-est-cz-s}}{\geq} \frac{(1-\epsilon_1)^3 \alpha}{2\cdot \beta_k \cdot \alpha^{k-\cluster(a)}}.
\end{equation}

Let $\psi$ be the smallest index such that $c_\psi \in A$. Let $B= K_{c_k} \setminus A$. Note that for any $b \in B$, $c_\psi$ is an ancestor of $b$. Then for any $b\in B$, by \Cref{lem:nlogn-ancestors-property}, 
\begin{equation}\label{eq:nlogn-distant-est-P-cpsi-b}
    r(P(c_\psi, b)) \geq \frac{1}{\alpha^{\psi - \cluster(b)}}.
\end{equation} 
Thus, for any $b\in B$,
\begin{align}\label{eq:nlogn-distant-est-P-cz-b}
    r(P(c_z, b)) &= r(P(c_z, c_\psi)) \cdot r(P(c_\psi, b)) \overset{\eqref{eq:nlogn-distant-est-P-cz-a},\eqref{eq:nlogn-distant-est-P-cpsi-b}}{\geq} \frac{(1-\epsilon_1)^3\alpha}{2\cdot \beta_k \cdot \alpha^{k-\psi}} \cdot \frac{1}{\alpha^{\psi - \cluster(b)}} = \frac{(1-\epsilon_1)^3\alpha}{2\cdot \beta_k \cdot \alpha^{k-\cluster(b)}}. 
\end{align}
Therefore, for any $s\in K_{c_k}$, 
\begin{align}\label{eq:nlogn-distant-est-P-us}
    r(P(u,s)) &= r(P(u,c_z)) \cdot r(P(c_z, s)) \overset{\eqref{eq:nlogn-distant-vertices-est-P-u-cz},\eqref{eq:nlogn-distant-est-P-cz-a},\eqref{eq:nlogn-distant-est-P-cz-b}}{\geq} \frac{(1-\epsilon_1)^3\alpha}{2} \cdot \frac{(1-\epsilon_1)^3\alpha}{2\cdot \beta_k \cdot \alpha^{k-\cluster(s)}} \nonumber\\
    &= \frac{(1-\epsilon_1)^6 \alpha^2}{4\cdot\beta_k \cdot \alpha^{k-\cluster(s)}} \geq \frac{\alpha^2}{8\cdot \beta_k \cdot \alpha^{k-\cluster(s)}} \geq \frac{Z_k}{\epsilon\cdot \alpha^{k-\cluster(s)}},
\end{align}
where we used that $\epsilon_1\leq \frac{1}{10}$ and $(1-x)^6 \geq 1/2$ for all $x\leq \frac{1}{10}$. Note that $r(P(s,u))=\frac{1}{r(P(u,s))}$. Therefore,
\begin{align*}
    \sum_{s\in K_v} r(P(s,u)) \overset{\eqref{eq:nlogn-distant-est-P-us}}{\leq} \sum_{s\in K_{c_k}} \frac{\epsilon\cdot  \alpha^{k - \cluster(s)}}{Z_k} \leq \frac{\epsilon}{Z_k} \sum_{\ell=1}^k |C_\ell| \cdot \alpha^{k-\ell} = \frac{\epsilon}{Z_k} \cdot Z_k = \epsilon, 
\end{align*}
where we used that $K_v \subseteq K_{c_k} \subseteq \bigcup_{\ell=1}^k C_\ell$.
\end{proof}

We now have all the ingredients to prove that \Cref{alg:nlogn-build-estimation-forest} produces a $(5,\epsilon)$-estimation-forest.

\begin{proof}[Proof of \Cref{thm:nlogn-build-estimation-forest}]
By \Cref{lem:cluster-sort-produces-cluster-graph}, $\AlgClusterSort(\alpha, \epsilon_1, \delta/3)$ correctly returns a $(\frac{2}{\alpha}, \frac{1}{\alpha}, \epsilon_1)$-cluster graph with probability at least $1- \frac{\delta}{3}$. Moreover, the next part of the algorithm makes at most $2n$ calls to $\AlgEstimateRatio$, and each call is successful with probability at least $1-\frac{\delta}{6n}$. Therefore, all the $\AlgEstimateRatio$ calls are successful with probability at least $1-\frac{\delta}{3}$ and therefore each call to $\AlgEstimateRatio$ as well as the call to $\AlgClusterSort$ is successful with probability at least $1 - \frac{2}{3}\cdot \delta$. If this event happens, then \Cref{lem:nlogn-graph-structure,lem:nlogn-close-vertices,lem:nlogn-distant-vertices-real,lem:nlogn-distant-vertices-estimates} ensure that the algorithm returns a $(5, \epsilon)$-estimation-forest. Finally, with probability at least $1-\frac{\delta}{3}$, the upper bound on the number of queries follows by \Cref{lem:nlon-query-bound}.
\end{proof}

\subsection{Learning the MNL from the Estimation-Forest}\label{sec:MNL-from-forest-adaptive}

In this section, we show how to use a $(t,\epsilon)$-estimation-forest to produce MNL weights that approximate the hidden MNL within a $\dmul$ error of $O(\epsilon)$.

Intuitively, an estimation-forest ensures that multiplying the estimates along a path gives a good estimate for the ratio of the weights of the two endpoints. These estimates are, in some sense, well-behaved even if the path is long. These properties suggest the following natural algorithm to generate the weights from a tree of the estimation-forest: assign an arbitrary weight to an initial vertex, and then assign all the other weights following the unique path from the initial one to all of the others. If we have multiple trees, by the properties of the estimation-forest it must be that any item in the tree with larger cluster indices wins against all the items of the other tree with very large probability (at least $1-\epsilon$). To mimic this property with our estimated weights, we rescale our estimates for the second tree by a sufficiently small value. The algorithm boils down to a depth-first search that we present in \Cref{alg:generate-weights}.

\begin{algorithm}
\caption{\AlgGenerateWeightsRec$(F, r, v, parent, \hat{w})$}
\label{alg:generate-weights-rec}
\begin{algorithmic}[1]
\State \textbf{Input:} A forest $F=([n], E)$ with a directed weighting $r$, the current vertex $v\in [n]$, the parent $parent\in[n]$ of $v$, and a vector of values $\hat{w}_1, \dots, \hat{w}_n$.
\State \textbf{Output:} For each $i$ in the subtree of $v$, it sets $\hat{w}_i$ to a positive weight, and it returns a set $W$ containing all the indices that have been modified.
\State $W\myassign \varnothing$
\For{$\{u,v\}\in E$ such that $u\neq parent$}
    \State $\hat{w}_u \myassign \hat{w}_v \cdot r(u,v)$ \label{line:gen-weight-follow-edge}
    \State $W\myassign W \cup \{u\} \cup \Call{\AlgGenerateWeightsRec}{F, r, u, v, \hat{w}}$
\EndFor
\State \textbf{return} $W$
\end{algorithmic}
\end{algorithm}

\begin{algorithm}
\caption{\AlgGenerateWeights$(\mathcal{F})$}
\label{alg:generate-weights}
\begin{algorithmic}[1]
\State \textbf{Input:} A $(t,\epsilon)$-estimation-forest $\mathcal{F}=(F=([n], E), r, (C_1, \dots, C_T), (c_1, \dots, c_T))$ for an MNL $M$.
\State \textbf{Output:} A new MNL $\hat{M}$ such that $d_1(M,\hat{M})\leq 9 \cdot \epsilon$.

\State Initialize an array $\hat{w}=(\hat{w}_1, \dots, \hat{w}_n)$, setting each entry to $\bot$
\State $w_{\text{min}} \myassign 1$
\For{$c=c_T, c_{T-1}\dots, c_1$} \Comment{iterate over the centers from the one of largest weight}
    \If{$\hat{w}_{c} = \bot$}
        \State $\hat{w}_{c} \myassign 1$ \Comment{initially assign weights w.r.t.\ 1}
        \State $W \myassign \{c\} \cup \Call{\AlgGenerateWeightsRec}{F, r, c, -1, \hat{w}}$
        \State $\Upsilon \myassign \max\limits_{j\in W}\{\hat{w}_j\}$
        \If{$c\neq c_T$}
            \For{$j\in W$} \Comment{rescale weights to account for the trees with larger weights}
                \State $\hat{w}_j \myassign \hat{w}_j \cdot \frac{\epsilon}{\Upsilon\cdot n} \cdot w_{\text{min}}$ \label{line:rescale-weights}
            \EndFor
        \EndIf
        \State $w_{\text{min}} \myassign \min\left\{w_{\text{min}}, \min\limits_{j\in W} \{\hat{w}_j\}\right\}$
    \EndIf
\EndFor
\State \textbf{return} MNL $\hat{M}$ induced by weights $\hat{w}_1, \dots, \hat{w}_n$
\end{algorithmic}
\end{algorithm}

We now show that the ratios of the estimated weights are well-behaved. 

\begin{lemma}\label{lem:estimated-weights-guarantees}
Let $\epsilon\in(0,\frac13)$, $\alpha\in(0,1)$, and let $\mathcal{F}=(F=([n], E), r, (C_1, \dots, C_T), (c_1, \dots, c_T))$ be a $(t,\epsilon)$-estimation-forest for MNL $M$ and let $\hat{M}$ be the MNL returned by $\AlgGenerateWeights(\mathcal{F})$. The following holds:
\begin{enumerate}
    \item for any $u,v\in [n]$ such that $d(u,v)\leq t$, $\frac{\hat{w}_u}{\hat{w}_v} \in (1\pm \epsilon) \cdot \frac{w_u}{w_v}$,
    \item for any $u,v\in [n]$ such that $d(u,v)>t$ and $\cluster(u)>\cluster(v)$, it holds that $\sum_{s\in H_v} \frac{\hat{w}_s}{\hat{w}_u} \leq 2\epsilon$, where $H_v = \{s\in [n] \mid \cluster(s) \leq \cluster(v)\}$.
\end{enumerate}
\end{lemma}
\begin{proof}
Let us start with the first point. Since $d(u,v)\leq t$, $u$ and $v$ are in the same connected component $\mathcal{C}$. Let $i\in [T]$ be the maximum index such that $c_i\in \mathcal{C}$. By construction, for any $x\in \mathcal{C}$, $\hat{w}_x = r(P(x,c_i)) \cdot \hat{w}_{c_i}$. Thus, we have, 
\begin{equation}\label{eq:estimated-weights-equal-r-path}
\frac{\hat{w}_u}{\hat{w}_v} = \frac{r(P(u,c_i)) \cdot \hat{w}_{c_i}}{r(P(v,c_i)) \cdot \hat{w}_{c_i}} = r(P(u,c_i)) \cdot r(P(c_i,v)) = r(P(u,v)),
\end{equation}
where the last equality holds since $\mathcal{C}$ is a tree and there is a unique path between every pair of vertices. Therefore, since $\mathcal{F}$ is a $(t,\epsilon)$-estimation-forest, $\frac{\hat{w}_u}{\hat{w}_v} = r(P(u,v)) \in (1\pm\epsilon) \cdot \frac{w_u}{w_v}$. 

We now prove the second point. Fix any $u,v\in [n]$ such that $d(u,v)>t$ and $\cluster(u)>\cluster(v)$. Let $\mathcal{C}$ be the connected component of $u$. Let us partition $H_v = \{s\in [n] \mid \cluster(s) \leq \cluster(v)\}$ into $A=\{s\in \mathcal{C} \mid \cluster(s)\leq \cluster(v)\}$ and $B=H_v \setminus A$. Observe that either $A$ is empty, or $v\in A$, by \Cref{def:estimation-forest}. Consider any $s\in A$. Since $s$ and $u$ are in the same connected component, by \Cref{eq:estimated-weights-equal-r-path}, $\frac{\hat{w}_s}{\hat{w}_u} = r(P(s,u))$. Then, by \Cref{def:estimation-forest},
\[
    \sum_{s\in A} \frac{\hat{w}_s}{\hat{w}_u} \overset{\eqref{eq:estimated-weights-equal-r-path}}{=} \sum_{s\in A} r(P(s, u)) \leq \epsilon.
\]
Consider now $s\in B$ and let $j$ be the largest index such that $c_j$ is in the same connected component as $s$. By construction of $\hat{M}$, it must be the case that 
\begin{equation}\label{eq:estimated-weights-cj-eps-n}
\hat{w}_{c_j} = w_{\text{min}}\cdot \frac{\epsilon}{\Upsilon \cdot n},
\end{equation}
for some $w_{\text{min}} \leq \hat{w}_u$ given that $\cluster(u) > j$ and $u$ and $c_j$ are in different trees. Moreover, by the definition of $\Upsilon$, $r(P(s, c_j)) \leq \Upsilon$. Observe that $\hat{w}_s = r(P(s,c_j)) \cdot \hat{w}_{c_j}$, and thus, 
\[
\hat{w}_s = r(P(s,c_j))\cdot \hat{w}_{c_j} \leq \Upsilon \cdot \hat{w}_{c_j}\overset{\eqref{eq:estimated-weights-cj-eps-n}}{\leq} \frac{\epsilon}{n} \cdot \hat{w}_u.
\]
Since this holds for each $s\in B$,
\[
    \sum_{s\in B} \frac{\hat{w}_s}{\hat{w}_u} \leq \sum_{s\in B} \frac{\epsilon}{n} = \epsilon\cdot \frac{|B|}{n} \leq \epsilon.
\]
Thus, $\sum_{s\in H_v} \frac{\hat{w}_s}{\hat{w}_u} \leq 2\epsilon$.
\end{proof}

We can finally prove that $\hat{M}$ has the desired guarantees. 

\begin{theorem}\label{thm:estimation-forest-produces-mnl}
Let $\epsilon\in (0,\frac{1}{9})$, $t\geq 2$, and let $\mathcal{F}=(F=([n], E), r, (C_1, \dots, C_T), (c_1, \dots, c_T))$ be a $(t,\epsilon)$-estimation-forest for MNL $M$ and let $\hat{M}$ be the MNL returned by $\AlgGenerateWeights(\mathcal{F})$. Then, $\dmul(M, \hat{M}) \leq 6\epsilon$.
\end{theorem}
\begin{proof}
Consider any slate $S \subseteq [n]$. Let $m\in S$ be such that for any other $s\in S$, $\cluster(s) \leq \cluster(m)$. Let $S_1 = \{s\in S \mid d(s,m)\leq t\}$ and $S_2 = S \setminus S_1$. Let $m_2\in S_2$ be such that for any $s\in S_2$, $\cluster(s) \leq \cluster(m_2)$. Let $H_{m_2} = \{s\in [n] \mid \cluster(s) \leq \cluster(m_2)\}$. 

By definition of $S_2$, $d(m, m_2)>t$. Note that, by \Cref{def:estimation-forest}, this implies $\cluster(m_2)\neq\cluster(m)$. By definition of $m_2$ and $m$, this in turn implies $\cluster(m_2)<\cluster(m)$.

Note also that $S_2 \subseteq H_{m_2}$ by the choice of $m_2$. Given these observations, by \Cref{def:estimation-forest}, we have,
\begin{align}\label{eq:proof-estimated-weights-real-S2}
    \sum_{s\in S_2} M_S(s) \leq \sum_{s\in S_2} \frac{w_s}{w_m} \leq \sum_{s\in H_{m_2}} \frac{w_s}{w_m} \leq \epsilon.
\end{align}
Similarly, by \Cref{lem:estimated-weights-guarantees},
\begin{align}\label{eq:proof-estimated-weights-est-S2}
    \sum_{s\in S_2} \hat{M}_S(s) \leq \sum_{s\in S_2} \frac{\hat{w}_s}{\hat{w}_m} \leq \sum_{s\in H_{m_2}} \frac{\hat{w}_s}{\hat{w}_m} \leq 2\epsilon.
\end{align}

Let us now focus on $S_1$. Note that, for any $u,v\in S_1$, $\frac{\hat{w}_u}{\hat{w}_v} \in (1\pm 3\epsilon) \frac{w_u}{w_v}$. Specifically, by \Cref{lem:estimated-weights-guarantees}, $\frac{\hat{w}_u}{\hat{w}_m} \in (1\pm \epsilon) \frac{w_u}{w_m}$ and $\frac{\hat{w}_m}{\hat{w}_v} \in (1\pm \epsilon) \frac{w_m}{w_v}$, thus,
\begin{equation}\label{eq:proof-estimated-weights-mult-approx}
     (1-2\epsilon) \cdot \frac{w_u}{w_v}\leq (1-\epsilon)^2 \cdot \frac{w_u}{w_v} \leq \frac{\hat{w}_u}{\hat{w}_m} \cdot \frac{\hat{w}_m}{\hat{w}_v} \leq (1+\epsilon)^2 \cdot \frac{{w}_u}{{w}_v} \leq (1+3\epsilon) \cdot \frac{w_u}{w_v},
\end{equation}
where we used that $(1+\epsilon)^2 \leq 1+3\epsilon$, for $\epsilon\in(0,1)$, and $(1-\epsilon)^2 \geq 1-2\epsilon$ for $\epsilon> 0$.

Note that by \Cref{eq:proof-estimated-weights-real-S2} and \Cref{eq:proof-estimated-weights-est-S2}, we also know that:
\begin{equation}\label{eq:proof-estimated-weights-sum-S2}
    \sum_{s\in S_2} w_s \leq \epsilon \cdot w_m, \quad \text{and}\quad  \sum_{s\in S_2} \hat{w}_s \leq 2\epsilon \cdot \hat{w}_m.
\end{equation} 
Therefore, for any $v\in S_1$, we have,
\begin{align}\label{eq:proof-estimated-weights-multapprox-1}
    \hat{M}_S(v) &= \frac{1}{\sum_{s\in S} \frac{\hat{w}_s}{\hat{w}_v}} \leq \frac{1}{\sum_{s\in S_1} \frac{\hat{w}_s}{\hat{w}_v}} \overset{\eqref{eq:proof-estimated-weights-mult-approx}}{\leq} \frac{1}{(1-2\epsilon)\sum_{s\in S_1} \frac{w_s}{w_v}} = \frac{1}{(1-2\epsilon)\left(\sum_{s\in S} \frac{w_s}{w_v} - \sum_{s\in S_2} \frac{w_s}{w_v}\right)} \nonumber\\
    & \overset{\eqref{eq:proof-estimated-weights-sum-S2}}{\leq} \frac{1}{(1-2\epsilon)\left(\sum_{s\in S} \frac{w_s}{w_v} - \epsilon \frac{w_m}{w_v}\right)} \leq \frac{1}{(1-2\epsilon)\left(\sum_{s\in S} \frac{w_s}{w_v} - \epsilon \sum_{s\in S}\frac{w_s}{w_v}\right)} \nonumber\\
    & = \frac{1}{(1-2\epsilon)(1-\epsilon) \sum_{s\in S} \frac{w_s}{w_v}} \leq \frac{1}{(1-3\epsilon) \sum_{s\in S} \frac{w_s}{w_v}} \nonumber\\
    & \leq (1+6\epsilon) \frac{1}{\sum_{s\in S} \frac{w_s}{w_v}} = (1+6\epsilon) \cdot  M_S(v), \nonumber
\end{align}
where we used that $(1-2\epsilon)(1-\epsilon) \geq 1-3\epsilon$ for $\epsilon>0$ and $\frac{1}{1-a} \leq 1+2a$ for $a\in(0,\frac12)$. Similarly,
\begin{align*}
    \hat{M}_S(v) &= \frac{1}{\sum_{s\in S} \frac{\hat{w}_s}{\hat{w}_v}} = \frac{1}{\sum_{s\in S_1} \frac{\hat{w}_s}{\hat{w}_v} + \sum_{s\in S_2} \frac{\hat{w}_s}{\hat{w}_v}} \overset{\eqref{eq:proof-estimated-weights-sum-S2}}{\geq} \frac{1}{\sum_{s\in S_1} \frac{\hat{w}_s}{\hat{w}_v} + 2\epsilon \cdot \frac{\hat{w}_m}{\hat{w}_v}} \nonumber\\
    & \geq \frac{1}{\sum_{s\in S_1} \frac{\hat{w}_s}{\hat{w}_v} + 2\epsilon \cdot \sum_{s\in S_1}\frac{\hat{w}_s}{\hat{w}_v}} = \frac{1}{(1+2\epsilon)\sum_{s\in S_1} \frac{\hat{w}_s}{\hat{w}_v}} \overset{\eqref{eq:proof-estimated-weights-mult-approx}}{\geq} \frac{1}{(1+2\epsilon)(1+3\epsilon)\sum_{s\in S_1} \frac{{w}_s}{{w}_v}} \nonumber\\
    & \geq \frac{1}{(1+6\epsilon)\sum_{s\in S_1} \frac{{w}_s}{{w}_v}} \geq (1-6\epsilon) \cdot M_{S_1}(v) \geq (1-6\epsilon) \cdot M_S(v), \nonumber
\end{align*}
where we used that $(1+2\epsilon)(1+3\epsilon) \leq 1+6\epsilon$ for $\epsilon \in(0,\frac16)$ and $\frac{1}{1+a} \geq 1-a$ for $a\geq 0$. Thus, $\hat{M}_S(v) \in (1\pm 6\epsilon) M_S(v)$ for any $v\in S_1$. 

Since this holds for any slate, we have that $\dmul(M, \hat{M})\leq 6\epsilon$.
\end{proof}

Putting everything together we obtain our main result as a corollary. The proof simply consists in first building an estimation-forest and then extracting the weights from it. We also show how to deal with the fact that weights could be large to store. Note that \Cref{thm:adaptive-nlogn} is a special case of the following result when $\delta=n^{-c}$ for some constant $c$. 

\begin{theorem}\label{thm:nlogn-algo-with-runtime}
There exists an adaptive randomized algorithm that, takes as input $\varepsilon \in(0,1)$, $\delta \in (0,1)$, and access to $\maxsample$ oracle for an MNL $M$ supported on $[n]$ for $n \in \mathbb{N}$, and with probability at least $1- \delta$, makes $O\left(n \log (\frac{n}{\delta})\cdot \left(\frac{1}{\epsilon^3} + \frac{\log(1/\delta)}{\log n}\right)\right)$ $\maxsample$ queries and solves the MNL learning problem for $\dmul$-error with accuracy parameter $\epsilon$. The algorithm runs in time proportional to the query complexity. 
\end{theorem}

\begin{proof}
Let $\epsilon'=\frac{\epsilon}{20}$. Obtain a $(5, \frac{\epsilon'}{6})$-estimation-forest $\mathcal{F}$ by calling $\AlgBuildEstimationForest(\frac{1}{2}, \frac{\epsilon'}{6}, \delta)$. Note that the algorithm has the desired query complexity. Then, obtain the MNL $\hat{M}$ by running algorithm  $\AlgGenerateWeights(\mathcal{F})$. By \Cref{thm:estimation-forest-produces-mnl}, we have $\dmul(M, \hat{M})\leq \epsilon'$. 

We now focus on the computational complexity. Note that \Cref{alg:nlogn-build-estimation-forest} has a running time equal to its query complexity. Also, \Cref{alg:generate-weights} performs $O(n)$ multiplications and makes no further queries. Let us assume without loss of generality that the output weights $\hat{w}_1, \dots, \hat{w}_n$ are sorted so that $\hat{w}_1 \geq \dots \geq \hat{w}_n$. Note that: 
\begin{equation}\label{eq:bound-maximum-size-est-weights}
\hat{w}_i \leq \hat{w}_{i+1}\cdot \frac{300\cdot n}{\epsilon'}.
\end{equation}
This is because each estimated weight is separated by its adjacent weight in the order by either: (i) a factor of $\frac{n}{\epsilon'}$ as in Line \ref{line:rescale-weights} of \Cref{alg:generate-weights}, or (ii) a factor of $r(u,v)$ as in Line \ref{line:gen-weight-follow-edge} of \Cref{alg:generate-weights-rec}. In the second case, by the properties of \AlgEstimateRatio{} (\Cref{lem:weight-pair-ratio-estimate}) and by inspecting the pseudocode of \Cref{alg:nlogn-build-estimation-forest}, each value $r(u,v)$ computed by $\AlgBuildEstimationForest(\frac{1}{2}, \frac{\epsilon'}{9}, \delta)$ is at most $(1+\frac{\epsilon'}{10}) \cdot \frac{288\cdot n}{\epsilon'} \leq \frac{300\cdot n}{\epsilon'}$. Note that storing all the values $\{\hat{w}_i\}_{i\in[n]}$ directly would require $O(n^2\log\frac{n}{\epsilon})$ bits. 

In order to address this issue, instead of storing the values $\{\hat{w}_i\}_{i\in[n]}$, we store their natural logarithm approximately and compute directly on these values. In particular, for any number $x$ used in the algorithm, we maintain a value $\lambda(x)$ that approximates $\ln(x)$. For all the values $x = r(u,v)$, which was represented as a fraction and computed using standard (exact) arithmetic by the previous subroutines, we compute a value $\lambda(x)$ such that:
\[
    \lambda(x) \in \ln x \pm {\varepsilon' \over n^2}.
\]
We do the same for the value $x=\frac{\epsilon'}{n}$. In general, for any number $x$, this value can be represented with $O(\ln \ln x  + \ln {n \over \varepsilon})$ bits. Whenever the algorithm needs to multiply two numbers $x$ and $y$ to obtain some $z= x\cdot y$  (e.g., Line \ref{line:gen-weight-follow-edge} of $\AlgGenerateWeightsRec$ or Line \ref{line:rescale-weights} of $\AlgGenerateWeights$), we instead compute $\lambda(z) = \lambda(x) + \lambda(y)$. Similarly, for $z=\max\{x,y\}$, we compute $\lambda(z)=\max\{\lambda(x),\lambda(y)\}$.

Note that every value $\frac{1}{\Upsilon}$ is determined by the product of at most $n$ numbers, and therefore $\lambda(\frac{1}{\Upsilon})$ is correct within an additive error of $\frac{\epsilon'}{n}$. Now, each weight is computed as the product of at most $2n$ numbers (considering the values of $r(u,v)$, $\frac{1}{\Upsilon}$, and $\frac{\epsilon'}{n}$), and for each of these numbers $x$ in the product, the value $\lambda(x)$ is correct within an additive error of at most $\frac{\epsilon'}{n}$. Therefore, $\lambda(\hat{w}_i)$ is correct within a $2\epsilon'$ additive error.  
This implies that:
\begin{equation}\label{eq:approximate-log-guarantees}
    (1-4\varepsilon') \cdot{\hat{w}_u} \leq e^{-{2\varepsilon'}} \cdot{\hat{w}_u} \leq e^{\ln(\hat{w}_u) - {2\varepsilon'}} \leq e^{\lambda(\hat{w}_u)} \leq e^{\ln(\hat{w}_u) + {2\varepsilon'}} \leq e^{2\varepsilon'} \cdot{\hat{w}_u} \leq (1+4\varepsilon') \cdot{\hat{w}_u}
\end{equation}
The algorithm then outputs the approximate logarithms of the weights $\{\lambda(\hat{w}_{i})\}_{i\in[n]}$. If one were to use the values $\{e^{\lambda(\hat{w}_{i})}\}_{i\in[n]}$ as proxies for the weights $\{\hat{w}_{i}\}_{i\in[n]}$ these would be correct to within multiplicative error $4\varepsilon'$ (by \eqref{eq:approximate-log-guarantees}). In particular, we have $\dmul(M, \tilde{M})\leq 18\epsilon'$, where $\tilde{M}$ is the MNL supported on $\{e^{\lambda(\hat{w_i})}\}_{i\in[n]}$. 

With these changes, all arithmetic operations performed need to be executed on numbers of at most $O(\log \,{n\over \varepsilon} )$ bits (by \eqref{eq:bound-maximum-size-est-weights}), and thus each of them can be executed in time $O(\log \,{n\over \varepsilon} )$. Therefore, the runtime of $O(n \log \frac{n}{\epsilon})$ of $\AlgGenerateWeights(\mathcal{F})$ is no larger than the query complexity. 
\end{proof}

\paragraph{Supporting items of weight zero.} In the context of MNLs, weights are assumed to be strictly positive. However, in the conditional sampling literature it is common to allow items of weight zero, and if a slate consists only of items of weight zero its distribution is uniform \citep{c20}. Our algorithm can be run as is and it works even in this case. Specifically, it turns out that any algorithm that can learn MNLs can also learn MNLs where items of zero weight are allowed. This is because these latter models arise as limits of MNLs and any algorithm that learns MNLs must necessarily learn these limiting models too, as we prove in \Cref{sec:extension-to-pseudo-mnl}.

\section{Adaptive $O(\frac{n \log n}{\epsilon^2})$ Algorithm for $d_1$}\label{sec:adaptive-optimal-eps-algorithm}

In this section we discuss how to improve the dependence in $\epsilon$ when approximating the $d_1$-error instead of $\dmul$. The high level idea is still to build an estimation forest where edges are labeled with the ratio of the endpoints, but we relax the guarantees required by this structure. Specifically, our goal is to construct an $\epsilon$-weak-estimation-forest as defined in \Cref{def:weak-estimation-forest}.

\subsection{Constructing an $\epsilon$-Weak-Estimation-Forest}

We now show how to build an $\epsilon$-weak-estimation-forest in $O(n \log n / \epsilon^2)$ time. We will start with the clustering phase of \Cref{alg:cluster-sort}, and then we will build a path over the cluster centers (from the heaviest one to the lightest one), breaking the path when the edge does not require accurate estimate. We will use a potential similar to the one used in \Cref{alg:nlogn-build-estimation-forest} to determine when to create a new connected component.

We will make use of the following subroutine to estimate weight ratios so that one side of the inequality is without error. The guarantees of this subroutine follow by Bernstein's inequality. 

\begin{algorithm}[h]
\caption{\AlgEstimateWeakRatio$(i,j,Z,\epsilon,\delta)$}
\label{alg:estimate-weak-ratio}
\begin{algorithmic}[1]
\State \textbf{Input:} Two items $i,j\in[n]$ with
$w_i\geq w_j$, parameters $Z\geq 1$,
$\epsilon,\delta\in(0,\frac{1}{2})$, and access to a
$\maxsample$ oracle for an MNL $M$ supported on $[n]$ with weights $\{w_1, \dots, w_n\}$.
\State \textbf{Output:} An estimate
$r(i,j)\in(0,\infty]$ of $w_i/w_j$.

\State $L\myassign\ln\frac{2}{\delta}$
\State $m\myassign
\left\lceil\frac{512\cdot Z \cdot L}{\epsilon^2}\right\rceil$
\State $X\myassign 0$

\For{$t=1,\dots,m$}
    \If{$\maxsample(\{i,j\})=j$}
        \State $X\myassign X+1$
    \EndIf
\EndFor

\If{$X<64 \cdot L$}
    \State \textbf{return} $\infty$
\EndIf

\State $\hat{p} \myassign
\frac{X-4\sqrt{X \cdot L}}{m}$
\State \textbf{return} $\frac{1-\hat{p}}{\hat{p}}$
\end{algorithmic}
\end{algorithm}

\begin{restatable}[\AlgEstimateWeakRatio{} Guarantees]{lemma}{EstimateWeakRatioGuarantees}
\label{lem:weak-ratio-estimation}
Let $i,j\in[n]$ satisfy $w_i\geq w_j$.
Let $Z\geq1$, $\epsilon,\delta\in(0,1)$. With probability at least $1-\delta$,
$\AlgEstimateWeakRatio(i,j,Z,\epsilon,\delta)$ returns a value $r(i,j)$ satisfying the following guarantees:
\begin{enumerate}[nosep]
    \item If $r(i,j)\neq\infty$, then:
    \[
        \left(1-\epsilon \cdot \sqrt{\frac{w_i}{Z \cdot w_j}} \right)
        \frac{w_j}{w_i} 
        \leq\frac{1}{r(i,j)}
        \leq\frac{w_j}{w_i}.
    \]
    \item If $r(i,j)=\infty$, then: 
    \[
        \frac{w_j}{w_i} \leq \frac{\epsilon^2}{2\cdot Z}.
    \]
\end{enumerate}
Moreover, the algorithm makes $O\left(\frac{Z \cdot \log (1/\delta)}{\epsilon^2}\right)$ pairwise queries.
\end{restatable}

We will make use of the following elementary inequality:

\begin{restatable}{lemma}{HelperInequalityEpsUB}\label{lem:build-weak-forest-inequality-utility}
    For any $x_1, \dots, x_n\in[0,1]$ and $z\in\mathbb{R}$, we have:
    \[
        \left|1 - z\cdot \prod_{i=1}^n x_i\right| \leq |1-z| + \sum_{i=1}^n (1-x_i). 
    \]
\end{restatable}

We now analyze \Cref{alg:build-weak-estimation-forest} and show that it successfully builds a weak-estimation-forest.

\begin{algorithm}[h]
\caption{\AlgBuildWeakForest$(\alpha,\varepsilon,\delta)$}
\label{alg:build-weak-estimation-forest}
\begin{algorithmic}[1]
\State \textbf{Input:} Parameters
$\alpha,\varepsilon,\delta\in(0,1)$, and access
to a $\maxsample$ oracle for an MNL $M$
supported on $[n]$ with weights
$\{w_1,\dots,w_n\}$.
\State \textbf{Output:} An
$\varepsilon$-weak-estimation-forest
$\mathcal{F}=(F,r,(C_1,\dots,C_T),(c_1,\dots,c_T))$.

\State $K_\alpha\myassign
16\left(1+\frac{2}{\alpha (1-\sqrt{\alpha})}\right)$
\State $\eta\myassign\frac{\varepsilon}{K_\alpha}$

\State $(F=([n],E),r,(C_1,\dots,C_T),
(c_1,\dots,c_T)) \myassign \Call{\AlgClusterSort}{\alpha,\eta,\frac{\delta}{3}}$

\State $Z_0\myassign0$
\For{$j=1,\dots,T$}
    \State $Z_j\myassign\alpha Z_{j-1}+|C_j|$
\EndFor

\For{$j=1,\dots,T-1$}
    \State $r(c_{j+1},c_j) \myassign
    \Call{\AlgEstimateWeakRatio}
    {c_{j+1},c_j,Z_j,\eta,\frac{\delta}{3n}}$

    \If{$r(c_{j+1},c_j)\neq\infty$}
        \State $E\myassign
        E\cup\{\{c_j,c_{j+1}\}\}$
    \EndIf
\EndFor

\State \textbf{return}
$(F=([n],E),r,(C_1,\dots,C_T),(c_1,\dots,c_T))$
\end{algorithmic}
\end{algorithm}

\begin{theorem}[Building a weak-estimation-forest]
\label{thm:build-weak-estimation-forest}
Let $\epsilon, \alpha, \delta\in(0,1)$ with constant $\alpha$, then, with probability at least $1-\delta$, $\AlgBuildWeakForest(\alpha, \epsilon, \delta)$ makes $O\left(n \log (\frac{n}{\delta}) \cdot \left(\frac{1}{\epsilon^2} + \frac{\log(1/\delta)}{\log n}\right)\right)$ $\maxsample$ queries and returns an $\epsilon$-weak-estimation-forest.
\end{theorem}
\begin{proof}
Note that since we make $O(n)$ calls to $\AlgEstimateWeakRatio$, all calls to this subroutine as well as the call to $\AlgClusterSort$ are successful with probability larger than $1-\delta$. Note that if $\AlgClusterSort$ is successful, then $w_{c_{j+1}} \geq \frac{1}{\alpha} w_{c_j} > w_{c_j}$, therefore the calls to $\AlgEstimateWeakRatio$ are well defined. We condition on these events for the rest of the proof. Note also that the query complexity follows by the complexity of $\AlgClusterSort$ (\Cref{lem:complexity-of-clustersort}) as well as by \Cref{lem:weak-ratio-estimation}:
\[
    \sum_{j=1}^{T-1} O\left(\frac{Z_j}{\eta^2} \cdot \log\left(\frac{n}{\delta}\right) \right) = \frac{\log(n/\delta)}{\epsilon^2} \cdot O\left(\sum_{j=1}^{T-1} Z_j\right) \overset{\eqref{eq:Z-j-upper-bound-construction}}{=} O\left(\frac{n \log(n/\delta)}{\epsilon^2}\right).
\]
We now focus on proving that the properties of \Cref{def:weak-estimation-forest} are satisfied. 

We start by proving the second property. Let $u\in [n]$ and consider any other vertex $v\in [n]$ in a different component than $u$ and such that $\gamma(u) \geq \gamma(v)$. Note that by construction, the connected component of the resulting graph consists of adjacent clusters, therefore for any $u'$ in the connected component of $u$ and $v'$ in the connected component of $v$, we have $\gamma(u') > \gamma(v')$. Let $j$ be the smallest index such that $c_j$ is in the same connected component of $u$. By \Cref{lem:weak-ratio-estimation}, since $\AlgEstimateWeakRatio$ returned $\infty$, we have:
\[
    w_{c_j} \geq \frac{2\cdot Z_{j-1}}{\eta^2} \cdot w_{c_{j-1}}.
\]
Note that $\gamma(v) \leq j-1$. Consider any $s\in[n]$ such that $\gamma(s) \leq \gamma(v)$. We have:
\begin{align*}
    w_u &\overset{\Cref{def:cluster-graph}}{\geq} \frac{\alpha}{2} \cdot w_{c_j} \geq \frac{\alpha \cdot Z_{j-1}}{\eta^2} \cdot w_{c_{j-1}} \overset{\Cref{def:cluster-graph}}{\geq} \frac{\alpha^2 \cdot Z_{j-1} \cdot \alpha^{\gamma(s) - j+1}}{2\cdot \eta^2} \cdot w_s \geq \frac{Z_{j-1} \cdot \alpha^{\gamma(s) - j+1}}{\epsilon^2} \cdot w_s.
\end{align*}
Therefore:
\[
    \sum_{\substack{s\in [n]\\ \gamma(s)\leq \gamma(v)}} \frac{w_s}{w_u} \leq \sum_{\substack{s\in [n]\\ \gamma(s)\leq \gamma(v)}} \frac{\epsilon^2 \cdot \alpha^{j-1-\gamma(s)}}{Z_{j-1}} \leq \epsilon^2 \cdot \sum_{\ell=1}^{j-1} \frac{|C_\ell| \cdot \alpha^{j-1-\ell}}{Z_{j-1}} = \epsilon^2 \leq \epsilon,
\]
which concludes the proof of the second property of \Cref{def:weak-estimation-forest}.

We now focus on the first property. Fix any $u\in [n]$ and let $\mathcal{C}$ be its connected component. Consider any edge between two cluster centers in $\mathcal{C}$, say $c_j$ and $c_{j+1}$. Since $r(c_{j+1}, c_j)\neq \infty$, by \Cref{lem:weak-ratio-estimation} we have:
\begin{align}\label{eq:weak-construction-ratio-bound}  
    1- \eta\cdot \sqrt{\frac{w_{c_{j+1}}}{Z_j\cdot w_{c_j}}} \leq \frac{r(c_{j}, c_{j+1})}{w_{c_j}/ w_{c_{j+1}}} \leq 1
\end{align}
Consider any $s\in \mathcal{C}$ with $\gamma(s)\leq \gamma(u)$. We have:
\begin{align*}
    &r(P(s,u)) = r(s,c_{\gamma(s)}) \cdot r(c_{\gamma(u)}, u) \cdot \prod_{j=\gamma(s)}^{\gamma(u)-1} r(c_{j}, c_{j+1})\\
    \iff & \frac{r(P(s,u))}{w_s/w_u} =   \frac{r(s,c_{\gamma(s)})}{w_s / w_{c_{\gamma(s)}}} \cdot \frac{r(c_{\gamma(u)}, u)}{ w_{c_{\gamma(u)}} / w_u} \cdot \prod_{j=\gamma(s)}^{\gamma(u)-1} \frac{r(c_{j}, c_{j+1})}{w_{c_{j}}/w_{c_{j+1}}}.
\end{align*}
Define for convenience $\zeta_{s,u} := \frac{r(s,c_{\gamma(s)})}{w_s / w_{c_{\gamma(s)}}} \cdot \frac{r(c_{\gamma(u)}, u)}{ w_{c_{\gamma(u)}} / w_u}$. Note that $\zeta_{s,u} \in [(1-\eta)^2, (1+\eta)^2]$ since the edges $\{s,c_{\gamma(s)}\}$ and $\{u,c_{\gamma(u)}\}$ belong to the cluster graph. We then have:
\begin{align*}
    \left|\frac{w_s}{w_u} - r(P(s,u))\right| &= \frac{w_s}{w_u} \cdot \left|1 - \frac{r(P(s,u))}{w_s/w_u}\right| =\frac{w_s}{w_u} \cdot \left|1 - \zeta_{s,u} \cdot \prod_{j=\gamma(s)}^{\gamma(u)-1} \frac{r(c_{j}, c_{j+1})}{w_{c_{j}}/w_{c_{j+1}}}\right| \\
    &\overset{\Cref{lem:build-weak-forest-inequality-utility}}{\leq} \frac{w_s}{w_u} \left( |1-\zeta_{s,u}| + \sum_{j=\gamma(s)}^{\gamma(u)-1} \left(1- \frac{r(c_j, c_{j+1})}{w_{c_j} / w_{c_{j+1}}}\right) \right)\\
    &\overset{\eqref{eq:weak-construction-ratio-bound}}{\leq} \frac{w_s}{w_u} \left( 3\eta + \eta\cdot \sum_{j=\gamma(s)}^{\gamma(u)-1} \sqrt{\frac{w_{c_{j+1}}}{Z_j \cdot w_{c_j}}} \right),
\end{align*}
where the invocation of \Cref{lem:build-weak-forest-inequality-utility} is justified by \eqref{eq:weak-construction-ratio-bound}. 

Consider now any
$A\subseteq\{s\in\mathcal{C}\mid
\gamma(s)\leq\gamma(u)\}$.
Without loss of generality, we assume $u\notin A$, since its
contribution to the sum is zero. By the previous derivation, we have:
\begin{align}\label{eq:weak-construction-final-bound}
    \sum_{s\in A}
    \left|\frac{w_s}{w_u}-r(P(s,u))\right| &\leq
    \frac{3\eta}{w_u} \cdot \sum_{s\in A} w_s
    +\frac{\eta}{w_u}
    \sum_{s\in A}
    \sum_{j=\gamma(s)}^{\gamma(u)-1}
    w_s\cdot \sqrt{\frac{w_{c_{j+1}}}{Z_j\cdot w_{c_j}}} \notag  \\
    &\leq
    \frac{3\eta}{w_u} \cdot \sum_{s\in A} w_s
    +\frac{\eta}{w_u}
    \sum_{j=1}^{\gamma(u)-1} \left(
    \sqrt{\frac{w_{c_{j+1}}}{Z_j\cdot w_{c_j}}} \cdot
    \sum_{\substack{s\in A\\\gamma(s)\leq j}}w_s \right).
\end{align}
Note that, by the properties of a cluster graph, for every $j<\gamma(u)$ we have:
\begin{align*}
    \sum_{\substack{s\in A\\ \gamma(s)\leq j}}w_s 
    \leq
    \sum_{\ell=1}^{j}
    |C_\ell|\cdot\frac{2}{\alpha}\cdot w_{c_\ell} \leq
    \frac{2 \cdot w_{c_j}}{\alpha} \cdot 
    \sum_{\ell=1}^{j}|C_\ell|
    \alpha^{j-\ell} =\frac{2}{\alpha}\cdot Z_j\cdot w_{c_j},
\end{align*}
from which we get:
\[
    \sum_{\substack{s\in A\\ \gamma(s)\leq j}}w_s  = \sqrt{\sum_{\substack{s\in A\\ \gamma(s)\leq j}}w_s  \cdot \sum_{\substack{s\in A\\ \gamma(s)\leq j}}w_s } \leq \sqrt{\frac{2}{\alpha}\cdot Z_j \cdot w_{c_j} \cdot \sum_{s\in A} w_s}.
\]
Plugging this into \eqref{eq:weak-construction-final-bound}, we obtain:
\begin{align*}
    \sum_{s\in A}
    \left|\frac{w_s}{w_u}-r(P(s,u))\right|
    &\leq
    \frac{3\eta}{w_u} \cdot \sum_{s\in A} w_s
    +\frac{\eta}{w_u}
    \sum_{j=1}^{\gamma(u)-1} \left(
    \sqrt{\frac{2\cdot w_{c_{j+1}}}{\alpha} \cdot \sum_{s\in A} w_s} \right)\\
    &\overset{\Cref{def:cluster-graph}}{\leq} \frac{3\eta}{w_u} \cdot \sum_{s\in A} w_s
    +\frac{\eta}{w_u}
    \sum_{j=1}^{\gamma(u)-1} \left(
    \sqrt{\frac{4\cdot w_{u}}{\alpha^2} \cdot \alpha^{\gamma(u)-j-1} \cdot \sum_{s\in A} w_s} \right)\\
    &= \frac{3\eta}{w_u} \cdot \sum_{s\in A} w_s + \frac{2\eta}{\alpha} \cdot \left( \sqrt{\sum_{s\in A} \frac{w_s}{w_u}}  \cdot \sum_{j=1}^{\gamma(u)-1} \alpha^{(\gamma(u) - j -1)/2} \right) \\
    & \leq  \frac{3\eta}{w_u} \cdot \sum_{s\in A} w_s + \frac{2\eta}{\alpha} \left( \frac{1 + \sum_{s\in A} w_s / w_u }{2} \cdot \sum_{j=0}^\infty (\sqrt{\alpha})^j\right)\\
    &\leq  \frac{3\eta}{w_u} \cdot \sum_{s\in A} w_s + \frac{\eta}{\alpha (1-\sqrt{\alpha})} \cdot \left(1 + \sum_{s\in A} \frac{w_s}{w_u} \right)\\
    & \leq \eta \cdot \left(1 + \sum_{s\in A} \frac{w_s}{w_u} \right) \cdot \left(3 + \frac{1}{\alpha (1-\sqrt{\alpha})}\right) \\
    & \leq \epsilon \cdot \left(1 + \sum_{s\in A} \frac{w_s}{w_u} \right) 
\end{align*}
where the third inequality uses that $\sqrt{x} \leq \frac{x+1}{2}$ for $x\geq 0$ and the last inequality follows by the definition of $\eta$.

This concludes the proof of the first
property of \Cref{def:weak-estimation-forest},
and therefore the theorem.
\end{proof}

\subsection{Learning the MNL from a Weak-Estimation-Forest}

We now show that an $\epsilon$-weak-estimation-forest is enough to obtain a $d_1$-error of $O(\epsilon)$ which concludes the description of our algorithm. The algorithm to generate weights from the forest is exactly the same as \Cref{alg:generate-weights}, only the analysis changes. 

\begin{theorem}[Reconstruction from a weak-estimation-forest]
\label{thm:weak-estimation-forest-reconstruction}Let $\epsilon\in (0,\frac{1}{9})$, and let $\mathcal{F}=(F=([n], E), r, (C_1, \dots, C_T), (c_1, \dots, c_T))$ be an $\epsilon$-weak-estimation-forest for MNL $M$ and let $\hat{M}$ be the MNL returned by $\AlgGenerateWeights(\mathcal{F})$. Then, $d_1(M, \hat{M}) \leq 6\epsilon$.
\end{theorem}
\begin{proof}
Consider any slate $S$. Let $u\in S$ be an item such that for any other $s\in S$, $\gamma(u) \geq \gamma(s)$. Let $U\subseteq S$ be the items in $S$ that are in the same connected component as $u$, and define $L=S \setminus U$. 

Note that for any $i\in U$ we have $M_S(i)=(1-\sum_{j\in L} M_S(j)) \cdot M_U(i)$ (and similarly for $\hat{M}_S(i)$). By this observation and by triangle inequality, we have:
\begin{equation}\label{eq:rec-decomposition-eps2}
    \norm{M_S - \hat{M}_S}_1 \leq \norm{M_U - \hat{M}_U}_1 + 2 \cdot \sum_{j\in L} (M_S(j) + \hat{M}_S(j))
\end{equation}
We analyze the two terms in the sum separately. First, observe that if $L= \varnothing$, the second term is zero. Suppose instead that $L\neq \varnothing$ and let $\ell\in L$ be an item such that for any other $\ell_2\in L$ it holds that $\gamma(\ell)\geq \gamma(\ell_2)$. By the second property of \Cref{def:weak-estimation-forest} we have:
\begin{equation}\label{eq:bound-L-eps2-rec}
    \sum_{j\in L} M_S(j) \leq \sum_{\substack{s\in S\\ \gamma(s)\leq \gamma(\ell)}} M_S(s) \leq \sum_{\substack{s\in [n]\\ \gamma(s)\leq \gamma(\ell)}} \frac{w_s}{w_u} \overset{\Cref{def:weak-estimation-forest}}{\leq} \epsilon.
\end{equation}
By construction of $\hat{w}_1, \dots, \hat{w}_n$ in \Cref{alg:generate-weights}, we have: $\hat{w}_s / \hat{w}_u \leq \frac{\epsilon}{n}$ for any $s\in L$, given that $\gamma(u) \geq \gamma(s)$ and they are in different connected components. This implies:
\begin{equation}\label{eq:bound-est-L-eps2-rec}
    \sum_{j\in L} \hat{M}_S(j) \leq \sum_{j\in L}\frac{\hat{w}_j}{\hat{w}_u} \leq \epsilon.
\end{equation}
Define a third MNL, $M'$, induced by the weights: 
\[
    w'_i := \frac{w_u}{\hat{w}_u} \cdot \hat{w}_i = w_u \cdot r(P(i, u)),
\]
for $i\in U$. Since $w'_i\propto \hat{w}_i$, we have $M'_U(j)=\hat{M}_U(j)$ for each $j\in U$. By the first property of \Cref{def:weak-estimation-forest} applied to $A=U \setminus\{u\}$, we have:
\begin{align*}
    &\sum_{s\in U\setminus\{u\}} \left| \frac{w_s}{w_u} - r(P(s,u)) \right| \leq \epsilon \cdot \left(1 + \sum_{s\in U\setminus\{u\}} \frac{w_s}{w_u}\right)\\
    \iff& \sum_{s \in U\setminus\{u\}} |w_s - w'_s| \leq \epsilon \cdot \sum_{s\in U} w_s\\
    \iff& \frac{\sum_{s \in U} |w_s - w'_s|}{\sum_{s\in U} w_s} \leq \epsilon
\end{align*}
We now use the following inequality. For any $a, b \in \mathbb{R}^n_{>0}$ we have:
\begin{align*}    
    \sum_{i\in[n]} \left| \frac{a_i}{\sum_{j\in[n]} a_j} - \frac{b_i}{\sum_{j\in[n]} b_j} \right| &\leq \sum_{i\in[n]} \left| \frac{a_i}{\sum_{j\in[n]} a_j} - \frac{b_i}{\sum_{j\in[n]} a_j} \right| + \sum_{i\in[n]} \left| \frac{b_i}{\sum_{j\in[n]} a_j} - \frac{b_i}{\sum_{j\in[n]} b_j} \right| \\
    &= \frac{\sum_{i\in[n]} |a_i - b_i|}{\sum_{j\in[n]} a_j} + \frac{\left| \sum_{j\in [n]} (b_j - a_j)\right|}{(\sum_{j\in [n]} a_j) \cdot (\sum_{j\in [n]} b_j)} \cdot \sum_{i\in [n]} b_i
    \\&\leq 2 \cdot \frac{\sum_{i\in[n]} |a_i - b_i|}{\sum_{j\in[n]} a_j}.
\end{align*}
Therefore, we have:
\begin{equation}\label{eq:bound-U-eps2-rec}
    \norm{M_U - \hat{M}_U}_1 = \norm{M_U - M'_U}_1 = \sum_{i\in U} \left| \frac{w_i}{\sum_{j\in U} w_j} - \frac{w'_i}{\sum_{j\in U} w'_j} \right| \leq 2 \cdot \frac{\sum_{i\in U} |w_i - w'_i|}{\sum_{j\in U} w_j} \leq 2\epsilon.
\end{equation}
Putting together \eqref{eq:rec-decomposition-eps2}, \eqref{eq:bound-L-eps2-rec}, \eqref{eq:bound-est-L-eps2-rec}, and \eqref{eq:bound-U-eps2-rec}, we obtain: $\norm{M_S - \hat{M}_S}_1 \leq 6\epsilon$ for any slate $S$, concluding the proof.
\end{proof}

It is easy to see that \Cref{thm:adaptive-nlogn-epsilon-tight} follows immediately from \Cref{thm:weak-estimation-forest-reconstruction} and \Cref{thm:build-weak-estimation-forest}.

\section{The Non-Adaptive Algorithm}\label{sec:non-adaptive-algorithm}
In this section we show that we can learn an MNL within a $\dmul$-error of $\epsilon$ by making at most $O\left(\frac{n^2 \log(n/\epsilon)\log (n/\delta)}{\epsilon^3}\right)$ non-adaptive queries. Specifically, by \Cref{lem:reduction-from-non-adaptive-to-adaptive}, it is sufficient to prove \Cref{thm:adaptive-balanced}, as this will imply the wanted result (\Cref{cor:non-adaptive-algo}). We recall the statement of \Cref{thm:adaptive-balanced}.

\ThmAdaptiveBalanced*

Recall that given a $(t,\varepsilon)$-estimation-forest, one can find the weights of an estimate MNL $\hat{M}$ with $\dmul(M,\hat{M})\leq O(\varepsilon)$ without making any more $\maxsample$ queries by using \AlgGenerateWeights{} (\Cref{alg:generate-weights}). Therefore, our goal is to design an adaptive algorithm that queries each pair of items at most $O(\log^2 n)$ times (with potentially some dependency on the error parameters $\delta$ and $\epsilon$) and then produces a $(t,\epsilon)$-estimation-forest. Observe that, unfortunately, \Cref{alg:nlogn-build-estimation-forest} does not have this property for two reasons.

First, the algorithm used to compute the $\orderingerror$-ordering  (described in \Cref{cor:falahatgar-for-MNLs}) can compare some items $\Omega(\log^3n)$ times. This can easily be fixed by creating the cluster graph via a variant of the classic Quicksort algorithm where each comparison is repeated sufficiently many times. In doing so, we obtain a smaller upper bound on the number of queries on each pair, in exchange for a higher overall worst-case query complexity. In fact, this algorithm will make $O(n\log^2 n)$ queries in total (for constant $\delta$ and $\epsilon$), but each pair will be queried at most $O(\log n)$ times. Formally, we have the following result, that we prove in \Cref{sec:appendix-non-adaptive}.

\begin{restatable}{proposition}{QuicksortClusterGraph}%
\label{prop:quicksort-cluster-graph}
There exists an algorithm $\AlgQuicksortClustering(\alpha, \varepsilon,\delta)$ that, given parameters $\alpha,\epsilon,\delta\in(0,1)$ and access to a $\maxsample$ oracle for an MNL $M$ supported on $[n]$, queries each pair of items at most $O\left(\frac{\log(n/\delta)}{\alpha \epsilon^2}\right)$ times and that, with probability at least $1-\delta$, returns a $(\frac{7}{\alpha}, \frac{1}{\alpha}, \epsilon)$-cluster graph.
\end{restatable}

The second reason is that \Cref{alg:nlogn-build-estimation-forest} might make $\Omega(n)$ queries to some pairs after the construction of the cluster graph. Indeed, consider an instance with $T=3$ clusters with $|C_1|=n-2$ and $|C_2|=|C_3|=1$. For constant $\alpha$ and $\epsilon$, we have $Z_2 = \Theta(n)$. In this case, \Cref{alg:nlogn-build-estimation-forest} queries the pair $\{c_3, c_2\}$ for $\Theta(Z_2)=\Theta(n)$ times and therefore does not have the property we seek to obtain an efficient non-adaptive algorithm. 

We fix this issue by introducing two technical ingredients. First, we modify our algorithm that constructs the estimation-forest so that it requires at most $O(|C_j| \cdot\log^2 (n))$ queries between any pair of cluster centers $\{c_i, c_j\}$. Second, instead of estimating the ratio $\nicefrac{w_{c_i}}{w_{c_j}}$ using only queries to the slate $\{c_i,c_j\}$, we make use of a new subroutine that constructs an estimate of $\nicefrac{w_{c_i}}{w_{c_j}}$ by querying each pair $\{c_i, e\}$, with $e\in C_j$, a balanced number of times. By dividing the $O(|C_j| \cdot\log^2 (n))$ queries equally among the $|C_j|$ items of cluster $C_j$, we obtain an algorithm that queries each pair at most $O(\log^2(n))$ times.

We first show how to spread the queries over the cluster in \Cref{sec:non-adaptive-spreading-queries}, and then we show an algorithm to build the estimation-forest by querying each pair at most $O(\log^2n)$ times in \Cref{sec:non-adaptive-balanced-algorithm}.

\subsection{Spreading the Queries Among the Cluster Items}\label{sec:non-adaptive-spreading-queries}

    \begin{algorithm}
    \caption{\AlgGetGeometric$(u, v)$}
    \begin{algorithmic}[1]
    \State \textbf{Input:} Two items $u$ and $v$ in $[n]$ and access to a $\maxsample$ oracle for an MNL $M$ supported on $[n]$ with weights $\{w_1, \dots, w_n\}$.
    \State \textbf{Output:} A natural number representing the number of samples taken from $\maxsample(\{u,v\})$ until $u$ is the winner (last one not included).
    
    \State $i \myassign{} 0$
    \While{True}
        \State winner $\myassign \maxsample(\{u,v\})$
        \If{winner $= u$}
            \State \textbf{return } $i$
        \EndIf
        \State $i \myassign i+1$
    \EndWhile
    \end{algorithmic}
    \end{algorithm}

    \begin{algorithm}
    \caption{\AlgBalancedEstimateRatio$(\mathcal{G},i,j,A_1, A_2, \varepsilon,\alpha, \delta)$}\label{alg:get-ratio-estimator}
    \begin{algorithmic}[1]
    \State \textbf{Input:} An $(A_1, A_2, \varepsilon)$-cluster graph $\mathcal{G}=(F, r, (C_1, \dots, C_T), (c_1,\dots, c_T))$, two natural numbers $i$ and $j$ representing the index of two clusters $C_i$ and $C_j$ in $\mathcal{G}$, parameters $A_1, A_2, \varepsilon,\alpha,$ and $\delta$.
    \State \textbf{Output:} An estimate $r(c_i,c_j)$ of the ratio $\nicefrac{w_{c_i}}{w_{c_j}}$.
    \State $B_1 = \max \left\{{2\varepsilon \over 1-\varepsilon -{3\over 4}}, {6\over (1-\varepsilon)}, {24\varepsilon \over23-4\varepsilon}\right\}$\Comment{Note that $B_1=O(1)$ for $\epsilon\in(0,\frac{1}{5})$.}
    \State $N(\alpha, \varepsilon) = { B_1^2\over \alpha \varepsilon^2}$
    \State $M\myassign{} \left\lceil{8 \log(2/\delta)}\right\rceil$
    \State $N \myassign{} \left\lceil 2\cdot A_1 \cdot \left(1+{A_1\over A_2}\right) N(\alpha, \varepsilon)\right\rceil$
    \State $\xi \myassign \left\lceil{MN \over |C_j|}\right\rceil$ 
    \For{$s \in C_j$}
        \For{$\ell=1, \dots, \xi$}
           \State $X_{\ell,s} \myassign {r(c_j,s)}\cdot \AlgGetGeometric(c_i, s)$
        \EndFor
    \EndFor
    \State Divide the first $MN$ values of $\{X_{\ell,s}\}_{\ell,s}$ into $M$ groups of size $N$:
    $\{X^{(1)}_1, \dots , X^{(1)}_{N}\}$, $\{X^{(2)}_1, \dots , X^{(2)}_{N}\}, \dots , \{X^{(M)}_1, \dots , X^{(M)}_{N}\}$.
    
    \For{$\ell=1, \dots , M$}
        \State $Y_\ell = {1\over N} \sum_{q=1}^{N}  X^{(\ell)}_{q}$
    \EndFor
    \State $Y=$ median$(\{Y_\ell\})$ \Comment{Computes the median of the values $Y_1, \dots , Y_{M}$}
    \If{$Y \leq {3\over 4} \cdot {\alpha}$}
        \State \textbf{return }$r(c_i, c_j) =\infty$
    \EndIf
    \State \textbf{return} $r(c_i, c_j) = 1/Y$
    \end{algorithmic}
    \end{algorithm}
In this section, we show a subroutine $\AlgBalancedEstimateRatio$ that produces an estimate $r(c_i,c_j)$ of $w_{c_i}\over w_{c_j}$ by spreading the queries among all items of the cluster $C_j$, instead of simply querying the pair $(c_i,c_j)$ repeatedly.

The key observation behind this algorithm, is that one can obtain an unbiased estimator of the ratio $\frac{w_v}{w_u}$ by counting the number of $\maxsample$ queries to $\{u,v\}$ needed before the oracle returns $u$ as the winner (\Cref{lem:geometric-guarantees}). Moreover, since the cluster graph produced in the first phase of the algorithm contains estimates of the ratios $\frac{w_u}{w_c}$ where $c$ is the center of $C_{\gamma(u)}$, one can compose these estimates with estimates of ratios of the form $\frac{w_{v}}{w_u}$ to obtain estimators for $\frac{w_v}{w_c}$. In the end, the algorithm simply uses a median-of-means estimator to aggregate the result. This produces an accurate estimate by standard concentration results.

We first show the following.
    \begin{lemma}\label{lem:geometric-guarantees}
        Let $Y$ be the output of \AlgGetGeometric$(u,v)$, then:
        \begin{equation*}
            \E{}{Y} = {w_v \over w_u} \hspace{1cm}\text{and} \hspace{1cm} \Var[Y] = {{w_{v} \over w_u+w_v} \over \left({w_u \over w_u+ w_v}\right)^2} = {w_v \over w_u}\left(1 + {w_v \over w_u}\right).
        \end{equation*}
    \end{lemma}
    \begin{proof}
        The statement follows directly from the fact that $Y+1 \sim \geom\left(\frac{w_u}{w_u+w_v}\right)$. 
    \end{proof}

    We will also make use of the following concentration bound:
    \begin{lemma}\label{lem:majority-boost}
        Let $X_1, \ldots, X_N$ be independent r.v's, where $X_i \sim \Bern(\mu_i)$ and for each $i$, $\mu_i \geq {3/4}$. Then, for $N \geq 8\ln{1\over \delta}$,
        \[
            \Pr\left[\sum_{i=1}^N X_i \leq \frac{N}{2} \right] \leq \delta.
        \]
    \end{lemma}
    \begin{proof}
        Note that $\E{}{\sum_{i=1}^N X_i} \geq \frac{3N}{4}$, thus, if $\sum_{i=1}^N X_i \leq \frac{N}{2}$ we also have $\E{}{\sum_{i=1}^N X_i} - \sum_{i=1}^N X_i \geq \frac{N}{4}$. By Chernoff-Hoeffding inequality (see, e.g., \citep[Theorem 1.1]{dp09}):
        \begin{align*}
            \Pr\left[\sum_{i=1}^N X_i \leq {N\over 2}\right] &\leq \Pr\left[\E{}{\sum_{i=1}^N X_i} - \sum_{i=1}^N X_i  \geq {N\over 4}\right] \leq \exp\left(- 2 \cdot {N^2\over 16} \cdot {1\over N}\right) \leq \delta.
            \qedhere
        \end{align*}
    \end{proof}
    We then show the following.

    \begin{lemma}\label{lem:correctness-of-ratio-estimator}
        Let $i,j \in [k]$ be such that $i > j$, and let $\mathcal{G} =(F, r, (C_1, \dots, C_T), (c_1, \dots, c_T))$ be an $(A_1,A_2,\varepsilon)$-cluster graph for some $\varepsilon \in (0,1/5)$. Then, with probability at least $1-\delta$, the algorithm \AlgBalancedEstimateRatio$(\mathcal{G},i,j,A_1, A_2, \varepsilon,\alpha, \delta)$ outputs $r(c_i,c_j) \in (0,\infty]$ such that: %
        \begin{enumerate}
            \item If ${w_{c_i}\over w_{c_j}} \leq {1\over \alpha}$ then $r(c_i,c_j) \neq \infty$,
            \item If ${w_{c_i}\over w_{c_j}} \geq {9 \over \alpha}$ then $r(c_i,c_j) = \infty$,
            \item If $r(c_i,c_j)  \neq \infty$ then:
            \[
                r(c_i,c_j)  \in (1\pm 10\varepsilon){w_{c_i} \over w_{c_j}} \hspace{4mm}\text{ and }\hspace{4mm} r(c_j,c_i) = \frac{1}{r(c_i,c_j)} \in (1\pm 10\varepsilon){w_{c_j} \over w_{c_i}}.
            \]
        \end{enumerate}
    \end{lemma}
    \begin{proof}

    Since $\mathcal{G}$ is a $(A_1,A_2,\varepsilon)$-cluster graph, we have, for every $s \in C_j$: 
    \begin{equation}\label{eq:ratio-je-is-almost-correct}
        r(c_j,s) \in(1\pm\varepsilon) {w_{c_j}\over w_s},
    \end{equation}
    and:
    \[
        {w_{c_i} \over w_{s}} = {w_{c_i} \over w_{c_j}} \cdot {w_{c_j} \over w_{s}} \geq A_2 \cdot {1\over A_1}=  {A_2\over A_1},
    \]
    and hence:
    \begin{equation}\label{eq:bound-wpi-we}
        {w_{c_i}\over w_s + w_{c_i}} = {1\over {w_s\over w_{c_i}} + 1} \geq {1\over {A_1\over A_2} + 1} = {A_2\over A_1+A_2}.
    \end{equation}
    
    By Lemma~\ref{lem:geometric-guarantees} and Equation \eqref{eq:ratio-je-is-almost-correct}, for any choice of $\ell\in [\xi]$ and any choice of $s \in C_j$:
    \begin{equation*}
        \E{}{X_{\ell,s}} = r(c_j,s)\cdot  \E{}{\AlgGetGeometric(c_i,s)} = r(c_j,s)\cdot {w_{s}\over w_{c_i}} \in (1\pm \varepsilon) {w_{c_j}\over w_{c_i}},
    \end{equation*}
    
    and:
    \begin{align*}
        \Var[X_{\ell,s}] &= \Var [r(c_j,s)\cdot{\text{\AlgGetGeometric}(c_i,s)}]\\
        &= r(c_j,s)^2 \cdot\Var [{\text{\AlgGetGeometric}(c_i,s)}]\\
        &= r(c_j,s)^2 \cdot {w_s\over w_{c_i}}\cdot \left(1+{w_s \over w_{c_i}}\right) \\
        &\leq (1+\varepsilon)^2\cdot \left({w_{c_j}\over w_s}\right)^2 \cdot {w_s\over w_{c_i}}\cdot \left(1+{w_s \over w_{c_i}}\right) \\
        &= (1+\varepsilon)^2 \cdot {w_{c_j}\over w_s} \cdot \left(1+{w_s \over w_{c_i}}\right)\cdot  {w_{c_j}\over w_{c_i}}  \\
        &\leq (1+\varepsilon)^2  \cdot A_1 \cdot \left(1+{A_1 \over A_2}\right)\cdot  {w_{c_j}\over w_{c_i}}\\
         &\leq 2 \cdot A_1 \cdot \left(1+{A_1 \over A_2}\right)\cdot  {w_{c_j}\over w_{c_i}}.
    \end{align*}

    In particular consider $Y_\ell$ obtained as the average of $\{X_1^{(\ell)}, \dots, X_N^{(\ell)}\}$, then:
    
    \begin{equation}\label{eq:approximately-correct-mean}
        \E{}{Y_{\ell}} \in (1\pm \varepsilon) \cdot {w_{c_j}\over w_{c_i}}
    \end{equation}

    and by independence:
    
    \[
        \Var[Y_{\ell}] \leq {1\over N} \max_{\substack{z\in[\xi]\\s\in C_j}} \Var[X_{z,s}]\leq {1\over N}\cdot 2\cdot A_1 \cdot \left(1+{A_1 \over A_2}\right) \cdot {w_{c_j}\over w_{c_i}} \leq {\alpha \varepsilon^2 \over B_1^2} \cdot {w_{c_j}\over w_{c_i}},
    \]
    giving:
    \[
        2\cdot \sigma(Y_\ell) = 2\cdot \sqrt{\Var[Y_\ell]}\leq 2{\varepsilon\over B_1}\sqrt{\alpha \cdot {w_{c_j}\over w_{c_i}}}.
    \]

    \noindent
    We are now ready to prove the three properties of the statement. We divide the proof depending on the value of $\frac{w_{c_i}}{w_{c_j}}$. First, if ${w_{c_i} \over w_{c_j}}\geq {9 \over \alpha} $, then the algorithm can fail only if it returns a value different from $\infty$. Under the assumption that ${w_{c_j} \over w_{c_i}}\leq { \alpha\over 9}$, we have:
    \[
        \E{}{Y_\ell}\leq (1+\varepsilon) \frac{w_{c_j}}{w_{c_i}} \leq {1+\varepsilon \over 9}\cdot \alpha 
    \]
        
    \[
        2\sigma(Y_{\ell}) \leq  2{\varepsilon\over B_1}\sqrt{\alpha \cdot {w_{c_j}\over w_{c_i}}} \leq 2{\varepsilon\over B_1}\sqrt{\frac{\alpha^2}{9}} \leq {2\over 3}\cdot {\varepsilon\over B_1}\cdot \alpha  \leq \left({3\over 4} - {1+\varepsilon \over 9}\right)\alpha
    \]

    By Chebyshev's Inequality:
    \begin{align*}
        \Pr\left[ Y_{\ell} \geq {3\over 4}\cdot \alpha \right] &\leq 
        \Pr\left[ Y_{\ell} - \E{}{Y_{\ell}}\geq {3\over 4}\cdot \alpha - \E{}{Y_{\ell}} \right]\\
        &\leq \Pr\left[ Y_{\ell} - \E{}{Y_{\ell}}\geq {3\over 4}\cdot \alpha - {(1+\varepsilon_1) \over 9}\cdot \alpha  \right]\\
        &\leq \Pr\left[ Y_{\ell} - \E{}{Y_{\ell}}\geq \left({3\over 4}  - {(1+\varepsilon_1) \over 9}\right)\cdot \alpha  \right]\\
        &\leq \Pr\left[ Y_{\ell} - \E{}{Y_{\ell}}\geq 2\cdot \sigma(Y_\ell) \right]\\
        &\leq {1\over 4}.
    \end{align*}
    The algorithm only returns $\infty$ if the value $Y$, which the median of $M$ $Y_{\ell}$'s, is smaller than ${3\over 4}\cdot {\alpha}$. This happens if and only if most of the $Y_{\ell}$'s are smaller than ${3\over 4}\cdot \alpha $. By Lemma~\ref{lem:majority-boost} this happens with probability at least $1- \delta$, as needed.

    Suppose now that $\frac{1}{\alpha}< {w_{c_i} \over w_{c_j}} \leq {9\over \alpha}$. In this case the algorithm fails if it simultaneously returns a value different from $\infty$ and such value is not a good estimate for the ratios. We now show that if ${w_{c_i} \over w_{c_j}} \leq {9\over\alpha}$ then $Y$ and $\frac{1}{Y}$ are good estimates with probability at least $1-\frac{\delta}{2}$. Under the assumption that ${w_{c_j} \over w_{c_i}} \geq {\alpha \over 9}$, we have:
    \[
        2\cdot \sigma(Y_{\ell}) \leq 2{\varepsilon\over B_1}\sqrt{\alpha \cdot {w_{c_j}\over w_{c_i}}} \leq 6\cdot{\varepsilon\over B_1} \cdot{w_{c_j}\over w_{c_i}}\leq \varepsilon\cdot (1-\varepsilon) \cdot{w_{c_j} \over w_{c_i}},
    \]
    and hence, again by Chebyshev's Inequality:
    
    \begin{align*}
        \Pr\left[ Y_{\ell} \not\in  (1\pm 3\varepsilon) \cdot {w_{c_j}\over w_{{c_i}}} \right]&\leq\Pr\left[ Y_{\ell} \not\in  (1\pm\varepsilon) \E{}{Y_{\ell}} \right]\\
        &= \Pr\left[ |Y_{\ell} - \E{}{Y_{\ell}}| \geq \varepsilon \E{}{Y_{\ell}} \right]\\
        &\leq \Pr\left[ |Y_{\ell} - \E{}{Y_{\ell}}| \geq \varepsilon\cdot (1-\varepsilon)\cdot{w_{c_j}\over w_{c_i}} \right]\\
        &\leq \Pr\left[ |Y_{\ell} - \E{}{Y_{\ell}}| \geq 2\sigma(Y_{\ell}) \right]\\
        &\leq {1\over 4}.
    \end{align*}
    where the first inequality follows from the fact that $(1\pm 3\varepsilon){w_{c_j}\over w_{c_i}} \supseteq (1\pm \varepsilon) \E{}{Y_\ell}$ by \eqref{eq:approximately-correct-mean}.     Hence with probability at least $3/4$ each $Y_{\ell}$ lies in the range $(1\pm 3\varepsilon){w_{c_j}\over w_{c_i}}$. Therefore, by the same argument as above, the median $Y$ lies in the interval $(1\pm 3\varepsilon){w_{c_j}\over w_{c_i}}$ with probability at least $1- {\delta\over 2}$. If this holds then we also have:
    \[
        (1-10\varepsilon)\cdot {w_{c_i} \over w_{c_j}}\leq {1\over 1+3\varepsilon}\cdot {w_{c_i} \over w_{c_j}}\leq {1\over Y} \leq {1\over 1-3\varepsilon}\cdot {w_{c_i} \over w_{c_j}}\leq (1+ 10\varepsilon) \cdot {w_{c_i} \over w_{c_j}}.
    \]
    Finally, suppose that ${w_{c_i} \over w_{c_j}} \leq \frac{1}{\alpha}$. In this case the algorithm can fail if it returns $\infty$ or if it returns an inaccurate estimate. As shown above, the latter happens with probability at most $\frac{\delta}{2}$. Moreover, the algorithm will return a number different from $\infty$ with probability at least $1-{\delta \over 2}$. Indeed, under the assumption ${w_{c_j} \over w_{c_i}} \geq {\alpha}$ we have:
    \[
        2\cdot \sigma(Y_{\ell}) \leq 2\frac{\epsilon}{B_1} \sqrt{  \alpha{w_{c_j} \over w_{c_i}}}\leq 2\frac{\epsilon}{B_1}\sqrt{ \left({w_{c_j} \over w_{c_i}}\right)^2}= 2  {\varepsilon\over B_1} \cdot {w_{c_j} \over w_{c_i}} \leq \left(1-\varepsilon -{3\over 4}\right){w_{c_j} \over w_{c_i}}.
    \]
    Applying Chebyshev's inequality, we have:
    \begin{align*}
       \Pr\left[Y_{\ell} \leq {3\over 4}\cdot \alpha \right] & = \Pr\left[\E{}{Y_\ell}  - Y_{\ell} \geq \E{}{Y_\ell} - {3\over 4}\cdot \alpha \right]\\
       & \leq  \Pr\left[|\E{}{Y_\ell}  - Y_{\ell}| \geq \E{}{Y_\ell} - {3\over 4}\cdot \alpha \right]\\
       & \leq  \Pr\left[|\E{}{Y_\ell}  - Y_{\ell}| \geq (1-\varepsilon)\cdot {w_{c_j} \over w_{c_i}} - {3\over 4}\cdot \alpha \right]\\
       & \leq  \Pr\left[|\E{}{Y_\ell}  - Y_{\ell}| \geq \left(1-\varepsilon - {3\over 4}\right)\cdot {w_{c_j} \over w_{c_i}}\right]\\
       &\leq \Pr\left[|\E{}{Y_\ell} - Y_\ell| \geq 2\sigma(Y_\ell)\right]\leq {1\over 4}.
    \end{align*}
    The algorithm returns $\infty$ only if most of the $Y_{\ell}$'s are smaller than ${3\over 4}\cdot \alpha$. By Lemma~\ref{lem:majority-boost} this happens with probability at most ${\delta \over 2}$, as needed. A union bound concludes the proof.
    \end{proof}

We now give a bound on the query complexity of $\AlgBalancedEstimateRatio$. We will use the following concentration bound for the sum of geometric random variables. 

\begin{lemma}\label{lem:new-geometric-sum-bound}
    Let $\lambda \in \R_{>0}$ and $X_1, \dots , X_n \sim \geom(p)$ be independent, identically distributed geometric random variables with parameter $p$, then:
    \[
        \Pr\left[ \sum_{i=1}^n X_i \geq {n\over p} + {\lambda\over p} +1\right]\le  \operatorname{exp} \left( -{2 p\lambda ^2\over n + \lambda + p }\right).
    \]
\end{lemma}
\begin{proof}
    Let $\nu \in  \R_{>0}$ such that $(\nu - 1)p > n$.
    The probability that $\sum_{i=1}^n X_i \ge \nu$ is the probability that the sum of $\lceil{\nu}\rceil-1$ independent Bernoulli random variables with parameter $p$ is less than $n$. Formally, let $Y\sim\bin(\lceil \nu \rceil-1, p)$, we have:
    \begin{align*}
        \Pr\left[ \sum_{i=1}^n X_i \ge \nu \right] &= \Pr\left[Y < n\right]=\Pr[Y - (\ceiling{\nu} - 1)p < n - (\ceiling{\nu} - 1)p]\\
        &\leq \exp\left(-\frac{2(n-(\ceiling{\nu}-1)p)^2}{\ceiling{\nu}-1}\right)\le \exp\left( -{2 (({\nu}-1) p- n)^2\over {\nu}}\right),
    \end{align*}
    where the first inequality follows by Chernoff-Hoeffding inequality (see, e.g., \cite[Theorem 1.1]{dp09}). Picking $\nu = {n\over p} + {\lambda\over p} + 1$ yields the result.
\end{proof}

\begin{lemma}\label{lem:complexity-get-ratio-estimator}
    For any choice of $\delta \in (0,1)$ and for $A_1=\Theta(1)$, $A_2=\Theta(1)$, a call to the algorithm \AlgBalancedEstimateRatio$(\mathcal{G},i,j,A_1, A_2, \varepsilon,\alpha, \delta)$, with probability at least $1-\delta$ queries each pair $\{c_i,s\}$, with $s\in C_j$, at most $O\left(\left(1+{1\over \alpha \varepsilon^2|C_j|}\right) \log {|C_j| \over \delta}\right)$ times.
\end{lemma}
\begin{proof}
    Note that \AlgBalancedEstimateRatio$(\mathcal{G}, i,j,A_1, A_2, \varepsilon_2, \delta)$ calls \AlgGetGeometric$(c_i,s)$ for each $s \in C_j$, $\xi =O(\nicefrac{MN}{|C_j|})$ times. Let $X_\ell$ be the number of $\maxsample$ queries made by the $\ell$th call to $\AlgGetGeometric(c_i,s)$. Note that $X_\ell\sim \geom(\frac{w_{c_i}}{w_s+w_{c_i}})$. By \Cref{eq:bound-wpi-we} we have:
    \[
        {w_{c_i}  \over w_{s}+w_{c_i}} \geq {A_2 \over A_1+A_2},
    \]
    and hence $\sum_{\ell=1}^\xi X_\ell$ is stochastically dominated by the sum of $\xi$ independent identically distributed geometric random variables with parameter $ p = {A_2 \over A_1 + A_2}$.

    By \Cref{lem:new-geometric-sum-bound}, we have that for $\lambda = \xi + {2\over p} \ln {10|C_j| \over \delta}$:
    \begin{align*}
        \Pr\left[ \sum_{\ell=1}^\xi X_\ell \geq {2\xi\over p} + {2 \over p^2 }\ln {10|C_j| \over \delta} + 1\right] &= \Pr\left[ \sum_{\ell=1}^\xi X_\ell \geq {\xi\over p} + {\lambda \over p} + 1\right] \leq \exp\left( - {2p \lambda^2 \over \xi + \lambda + p } \right) \\
        &\leq \exp\left( - {2p \lambda^2 \over 3\lambda } \right) =\exp\left( - {2p \lambda \over 3} \right)\\
        &\leq \exp\left( - \ln {10|C_j|\over \delta} \right) = {\delta \over 10 |C_j|}.
    \end{align*}
    Since ${2\xi\over p} + {2 \over p^2 }\ln {10|C_j| \over \delta} + 1 = O\left(\left(1+{1\over \alpha \varepsilon^2|C_j|}\right) \log {|C_j| \over \delta}\right)$ the lemma follows by the union bound.
\end{proof}

\subsection{Computing an MNL: Obtaining an Estimation-Forest}\label{sec:non-adaptive-balanced-algorithm}
We conclude this section by discussing how, given the subroutines described above, we can compute a $(t,\varepsilon)$-estimation-forest for an MNL $M$ by making at most $O\left(\frac{\log (n/\epsilon)\cdot \log(n/\delta)}{\epsilon^3}\right)$ queries per pair.   

We provide the pseudocode of our algorithm in \Cref{alg:build-forest-non-adaptive}. Before proceeding, we give some intuition for the algorithm. Recall that our $O(n\log n)$ adaptive algorithm (\Cref{alg:nlogn-build-estimation-forest}), starting from some center $c_i$, iteratively attempts to estimate the ratio of the weights of $c_i$ and those of the previous centers $c_{i-1}, c_{i-2}, \dots$ until it observes an ``$\infty$''. The first idea of this new algorithm is that we can limit ourselves to estimating the ratios between the weight of $c_i$ and those of $\{c_{i-1}, \dots, c_{i-\Lambda(n)}\}$ for $\Lambda(n)=\Theta(\log n)$---this is because each time that we jump to a cluster with smaller cluster index the weights of the items decrease by at least a constant. 
The second idea is to choose a different threshold for establishing when the ratio of two items is infinite. In \Cref{alg:nlogn-build-estimation-forest} the threshold was decided so that if $r(c_i,c_j)=\infty$ then $c_i$ alone wins with probability at least $1-\epsilon$ against all the items in $C_j \cup \dots \cup C_1$ put together. In this second algorithm, intuitively, we would like $r(c_i,c_j)=\infty$ if $c_i$ alone wins against all the items in $C_j$ with probability at least $1-\frac{\epsilon}{\Lambda(n)}$, as this is sufficient to obtain an $(t,O(\varepsilon))$-estimation-forest. 
This can be achieved by testing whether $\frac{w_{c_j}}{w_{c_i}} \leq \beta_j$ for $\beta_j=\Theta\left(\frac{\epsilon}{|C_j|\cdot \Lambda(n)}\right)$. This new threshold will lead to overall more queries but it will allow us to query each pair at most $O(\log^2n)$ many times. 

A negative byproduct of this weaker thresholding is that we cannot stop estimating the ratios between $c_i$ and the other centers as soon as we observe an infinity, but we should test $c_i$ against all the items in $\{c_{i-1}, \dots, c_{i-\Lambda(n)}\}$. In turn, this might create cases where we observe $r(c_i,c_j)=\infty$ but still obtain a good estimate for the ratio of $c_i$ and $c_\ell$ for $\ell>j$. To avoid this situation, we estimate the ratios between $c_i$ and, in order, $c_{i-\Lambda(n)}, c_{i-\Lambda(n)+1}, \dots, c_{i-1}$. We call $j_m$ the maximum index for which the estimate $r(c_i,c_{j_m})$ is not infinite---which is also the first one to occur. Now, for $\ell\in \{i-1, \dots, j_m+1\}$ we estimate the ratio between $c_i$ and $c_\ell$ by using $r(c_i, c_{j_m})$ and $r(c_\ell, c_{j_m})$. We can then leverage the properties of the cluster graph to ensure that, with high probability, for all these choices of $\ell$, $r(c_\ell, c_{j_m})$ will not be infinite.

We now analyze \Cref{alg:build-forest-non-adaptive}. Our goal is to prove the following result.
\begin{theorem}\label{thm:non-adaptive-estimation-forest-creation}
For any choice of three numbers $\alpha, \epsilon, \delta\in(0,1)$, with $\alpha=\Theta(1)$, with probability at least $1 - \delta$, $\AlgBuildForestNonAdaptive(\alpha, \epsilon,\delta)$ makes at most $O\left(\frac{\log(n/\epsilon) \cdot \log(n/\delta)}{\epsilon^3}\right)$ queries to each pair of items and returns a $(5,\epsilon)$-estimation-forest.
\end{theorem}

Observe that, thanks to \Cref{thm:estimation-forest-produces-mnl}, the above result immediately implies \Cref{thm:adaptive-balanced}. We will now prove \Cref{thm:non-adaptive-estimation-forest-creation} via a series of lemmas. We begin by analyzing the query complexity. 
\begin{algorithm} \caption{\AlgBuildForestNonAdaptive$(\alpha,\varepsilon,\delta)$}\label{alg:build-forest-non-adaptive}
\begin{algorithmic}[1]
\State \textbf{Input:} Parameters $\alpha, \varepsilon, \delta\in(0,1)$, and access to a $\maxsample$ oracle for an MNL $M$ supported on $[n]$ with weights $\{w_1, \dots, w_n\}$.
\State \textbf{Output:} with probability $\geq 1 - \delta$, a $(5,\varepsilon)$-estimation-forest 
\State $\epsilon_1 \myassign \nicefrac{\epsilon}{10}$
\State $\epsilon_2 \myassign \nicefrac{\epsilon_1}{30}$
\State $(F=([n], E), r, (C_1, \dots, C_T), (c_1, \dots, c_T)) \myassign \AlgQuicksortClustering(\alpha, \epsilon_2, \frac{\delta}{4})$ %
\State Let $\mathcal{G}$ be a copy of the cluster graph $(F=([n], E), r, (C_1, \dots, C_T), (c_1, \dots, c_T))$
\State $\Lambda(n)=\lceil \log_{1/\alpha}(\frac{49\cdot n}{\alpha\cdot \epsilon_1}) \rceil$
\For{$i\in [T]$}
    \State $\beta_i \myassign \frac{\alpha^2 \cdot \epsilon_1}{49\cdot |C_i| \cdot \Lambda(n)}$
\EndFor
\State $i\myassign T$
\While{$i > 1$}
    \State $j_m \myassign -1$
    \State $j \myassign \max\{1, i-\Lambda(n)\}$
    \While{$j < i$ and $j_m =-1$}
        \State $r(c_i, c_j) \myassign \max\left\{\AlgBalancedEstimateRatio(\mathcal{G}, i, j,\frac{7}{\alpha}, \frac{1}{\alpha}, \epsilon_2, \beta_j, \frac{\delta}{4n^2}), \frac{1}{\alpha^{i-j}}\right\}$\label{BBEF:line:First-Call-to-BER}
        \If{$r(c_i,c_j)<\infty$}
            \State $j_m \myassign j$
            \State $E \myassign E \cup \{(c_i, c_j)\}$ 
        \EndIf
        \State $j \myassign j +1$
    \EndWhile
    
    \If{$j_m=-1$}
        \State {\color{gray} // {If all the estimates of $r(c_i,c_j)$ are $\infty$, the current tree ends here}}
        \State $i\myassign i-1$
    \Else 
        \State {\color{gray}// Otherwise, the values of $r(c_i,c_j)$ for $j \in \{i-1, ... , j_m+1\}$ is estimated using estimates for the ratios $\nicefrac{w_{c_i}}{w_{c_{j_m}}}$ and $\nicefrac{w_{c_{j_m}}}{w_{c_{j}}}$}
        \For{$j = i-1, \dots, j_m+1$}
            \State $\rho \myassign \AlgBalancedEstimateRatio(\mathcal{G}, j, j_m, \frac{7}{\alpha}, \frac{1}{\alpha}, \epsilon_2, \frac{\beta_{j_m}}{9}, \frac{\delta}{4n^2})$\label{BBEF:line:Second-Call-to-BER}
            \If{$\rho\in\{0,\infty\}$}
                \State \textbf{return} failure
            \EndIf
            \State $r(c_i, c_j) \myassign \max\left\{\frac{r(c_i,c_{j_m})}{\rho}, \frac{1}{\alpha^{i-j}}\right\}$
            \State $E \myassign E \cup \{(c_i, c_j)\}$
        \EndFor
        \State $i\myassign j_m$
    \EndIf
\EndWhile
\State \textbf{return} $(F=([n], E), r, (C_1, \dots, C_T), (c_1, \dots, c_T))$
\end{algorithmic}
\end{algorithm}

\begin{lemma}\label{lem:non-adaptive-query-cost-algorithm}
Let $\alpha,\epsilon,\delta\in(0,1)$, with $\alpha=\Theta(1)$. Suppose that the call made to $\AlgQuicksortClustering$ correctly computes a cluster graph. Then, with probability at least $1-\frac{\delta}{2}$, we have that the algorithm $\AlgBuildForestNonAdaptive(\alpha,\epsilon, \delta)$ makes at most $O\left(\frac{\log(n/\epsilon) \cdot \log(n/\delta)}{\epsilon^3}\right)$ $\maxsample$ queries on each pair of items, and it makes $O\left(\frac{n \log^2(n/\epsilon) \log(n/\delta)}{\epsilon^3}\right)$ queries in total. 
\end{lemma}
\begin{proof}
$\AlgQuicksortClustering$ makes at most $O\big(\frac{\log(n/\delta)}{\epsilon^2}\big)$ queries for each pair as stated in \Cref{prop:quicksort-cluster-graph}.

After that, all the $\maxsample$ queries performed by $\AlgBuildForestNonAdaptive$ occur during a call to $\AlgBalancedEstimateRatio$. 

Note that, for each pair $(i,j)$, $\AlgBalancedEstimateRatio(\mathcal{G}, i,j,\dots)$ is called at most once. %
By \Cref{lem:complexity-get-ratio-estimator}, with probability $1-\frac{\delta}{4n^2}$, $\AlgBalancedEstimateRatio(\mathcal{G}, i, j,\dots)$ makes at most 
\[
    O\left(\left(1+\frac{1}{\beta_j \cdot |C_j| \cdot \epsilon^2}\right) \log\left(\frac{n^2\cdot |C_j|}{\delta}\right)\right) =     O\left(\frac{\Lambda(n)}{\epsilon^3} \log\left(\frac{n}{\delta}\right)\right) = O\left(\frac{\log (\frac{n }{\epsilon}) \cdot\log\left(\frac{n}{\delta}\right)}{\epsilon^3}\right)
\]
queries to each pair $(c_i, v)$, with $v\in C_j$.    

Given that $\AlgBalancedEstimateRatio$ gets called less than $n^2$ times, with probability at least $1 - \frac{\delta}{4}\geq 1 - \frac{\delta}{2}$, each pair $(c_i,v)$ is queried at most $O\left(\frac{\log(n/\epsilon) \log(n/\delta)}{\epsilon^3}\right)$ times after the clustering algorithm.

Note also that for each $j$ there are at most $\Lambda(n)$ values of $i$ for which the algorithm makes a call to $\AlgBalancedEstimateRatio(\mathcal{G}, i,j, \dots)$. Thus, the total number of queries is upper bounded by,
\[
    \sum_{j\in [T]} \Lambda(n) \cdot |C_j| \cdot  O\left(\frac{\log(n/\epsilon) \cdot \log(n/\delta)}{\epsilon^3}\right) \leq O\left(\frac{n \log^2(n/\epsilon)\log(n/\delta)}{\epsilon^3}\right). \qedhere
\]
\end{proof}

We now show that the algorithm computes a $(5,\epsilon)$-estimation-forest. This consists in proving four properties. We will do so in a series of lemmas. We start by showing that short paths provide good estimates of the weights ratio. 

\begin{lemma}\label{lem:non-adaptive-close-vertices}
Let $\epsilon, \alpha, \delta\in(0,1)$, and suppose that each call to $\AlgBalancedEstimateRatio$ as well as the call to $\AlgQuicksortClustering$ is successful. Then, $\AlgBuildForestNonAdaptive(\alpha, \epsilon, \delta)$ does not return failure. Moreover, let $\mathcal{F}=(F=([n], E), r, (C_1, \dots, C_T), (c_1, \dots, c_T))$ be the output of the algorithm. Then, for $\epsilon_1=\frac{\epsilon}{10}$ and for any integer $t\geq 1$ and any $u,v\in [n]$ such that $d(u,v)\leq t$, 
\[
(1-\epsilon_1)^t \cdot \frac{w_u}{w_v} \leq r(P(u,v)) \leq (1+\epsilon_1)^t \cdot \frac{w_u}{w_v}. 
\]
In particular, if $d(u,v)\leq 5$, $r(P(u,v)) \in (1\pm \epsilon) \cdot \frac{w_u}{w_v}$.
\end{lemma}
\begin{proof}
Just as in to the proof of \Cref{lem:nlogn-close-vertices}, it is sufficient to show that each edge $(u,v)$ in the forest is associated with a $(1\pm \epsilon_1)$ estimate of the ratio $\frac{w_u}{w_v}$. Since $\AlgQuicksortClustering$ returned a $(\frac{7}{\alpha}, \frac{1}{\alpha}, \epsilon_2)$-cluster graph with $\epsilon_2=\frac{\epsilon_1}{30}$, we have $r(c_i,v)\in (1\pm \epsilon_2) \frac{w_{c_i}}{w_v}$ and $r(v,c_i)\in (1\pm \epsilon_2) \frac{w_{v}}{w_{c_i}}$ for each edge $(c_i,v)$ for $i\in[T]$, $v\in C_i$. 

Consider now an edge $(c_i, c_j)$, for $i>j$. This is either added in the internal \textbf{while} loop that looks for $j_m$ or it is included in the internal \textbf{for} loop. 

We first consider the \textbf{while} loop case. Since $i>j$ and $r(c_i,c_j)\neq \infty$, by \Cref{lem:correctness-of-ratio-estimator}, the call to $\AlgBalancedEstimateRatio$ returned a value $\zeta$ such that $\zeta\in (1\pm 10\epsilon_2) \frac{w_{c_i}}{w_{c_j}}$ and $1/\zeta \in (1\pm 10\epsilon_2) \frac{w_{c_j}}{w_{c_i}}$. Moreover, by the properties of the $(\frac{7}{\alpha}, \frac{1}{\alpha}, \epsilon_2)$-cluster graph, $\frac{w_{c_i}}{w_{c_j}} \geq \frac{1}{\alpha^{i-j}}$. This implies that $r(c_i,c_j)=\max\{\zeta, 1/\alpha^{i-j}\} \in (1\pm 10\epsilon_2) \frac{w_{c_i}}{w_{c_j}}$ and $r(c_j,c_i)=1/r(c_i,c_j) = \min\{1/\zeta, \alpha^{i-j}\} \in (1\pm 10\epsilon_2)\frac{w_{c_j}}{w_{c_i}}$.

Consider now the \textbf{for} loop case. Note that we have $i> j > j_m$. Note that when the \textbf{for} loop is executed, the value of $r(c_i,c_{j_m})$ is never $\infty$. In particular the call made on Line~\ref{BBEF:line:First-Call-to-BER} to $\AlgBalancedEstimateRatio(\mathcal{G}, i, j_m,\frac{7}{\alpha}, \frac{1}{\alpha}, \epsilon_2, \beta_{j_m}, \frac{\delta}{4n^2})$ returned a value other than $\infty$, and hence, by \Cref{lem:correctness-of-ratio-estimator}, it must be true that $\frac{w_{c_i}}{w_{c_{j_m}}} \leq \frac{9}{\beta_{j_m}}$. Moreover, we have $w_{c_i} \geq w_{c_j}$, and thus, $\frac{w_{c_j}}{w_{c_{j_m}}} \leq \frac{9}{\beta_{j_m}}$. Then, by \Cref{lem:correctness-of-ratio-estimator}, the call made to $\AlgBalancedEstimateRatio(\mathcal{G}, j, j_m, \frac{7}{\alpha}, \frac1\alpha, \epsilon_2, \frac{\beta_{j_m}}{9}, \frac{\delta}{4n^2})$ on Line~\ref{BBEF:line:Second-Call-to-BER}, must return $\rho\neq \infty$, and therefore $\rho\in (1\pm 10\epsilon_2)\frac{w_{c_j}}{w_{c_{j_m}}}$. Therefore, the algorithm does not return ``failure'' and moreover, we have,
\[
    (1-20\epsilon_2) \frac{w_{c_i}}{w_{c_j}}\leq \frac{(1-10\epsilon_2)}{(1+10\epsilon_2)} \cdot \frac{w_{c_i}}{w_{c_j}}\leq \frac{r(c_i, c_{j_m})}{\rho} \leq \frac{(1+10\epsilon_2)}{(1-10\epsilon_2)} \cdot \frac{w_{c_i}}{w_{c_{j_m}}} \cdot \frac{w_{c_{j_m}}}{w_{c_j}} \leq (1+30\epsilon_2) \frac{w_{c_i}}{w_{c_j}},
\]
where we used that $\frac{1-a}{1+a} \geq 1-2a$ for $a\geq 0$ and $\frac{1+a}{1-a} \leq 1+3a$ for $a\in(0,1/3)$. Moreover, by the fact that $\frac{w_{c_i}}{w_{c_j}} \geq \frac{1}{\alpha^{i-j}}$ and since $\epsilon_2 \leq \frac{\epsilon_1}{30}$, we have $r(c_i, c_j)=\max\left\{\frac{r(c_i,c_{j_m})}{\rho}, \frac{1}{\alpha^{i-j}} \right\}\in (1\pm\epsilon_1) \frac{w_{c_i}}{w_{c_j}}$. One can similarly prove the result for $1/r(c_i,c_j)$. This directly concludes the proof. 
\end{proof}

We now show that if two vertices are far away in the forest, then the ratio of their weights is negligible. We will do so by analyzing the structure of the tree similarly to the argument of \Cref{lem:nlogn-distant-vertices-real}. \Cref{fig:proof-tree} can be used as a reference to visualize the disposition of the vertices: although the tree in the figure is generated by \Cref{alg:nlogn-build-estimation-forest}, the disposition of the vertices is similar in this proof.

\begin{lemma}\label{lem:non-adaptive-distant-vertices-real}
Let $\epsilon, \alpha, \delta\in(0,1)$, and suppose that each call to $\AlgBalancedEstimateRatio$ as well as the call to $\AlgQuicksortClustering$ is successful. Let $\mathcal{F}=(F=([n], E), r, (C_1, \dots, C_T), (c_1, \dots, c_T))$ be the output of $\AlgBuildForestNonAdaptive(\alpha, \epsilon, \delta)$. Then, for any $u,v\in [n]$ such that $d(u,v)>5$ and $\cluster(u)>\cluster(v)$, it holds that $\sum_{s\in H_v} \frac{w_s}{w_u} \leq \epsilon$, where $H_v = \{s\in [n] \mid \cluster(s) \leq \cluster(v)\}$.
\end{lemma}
\begin{proof}
We first consider the case where $u$ and $v$ are in the same connected component $\mathcal{C}$. Since $d(u,v)\geq 6$, $d(c_{\cluster(u)}, c_{\cluster(v)})\geq 4$. Note that $\mathcal{C}$ is a tree. Let $R\in[T]$ be the maximum index such that $c_R\in \mathcal{C}$ and suppose the tree is rooted at $c_R$. Let $c_z$ be the ancestor of $c_{\cluster(v)}$ at distance 2 from $c_{\cluster(v)}$. Note that $c_z$ exists and by construction $\cluster(u) \geq z > \cluster(v)$. Thus, $w_{\cluster(u)} \geq w_{c_z}$. Since $c_{\cluster(v)}$ is not a child of $c_z$, for any $s\in H_v$, $c_{\cluster(s)}$ is not a child of $c_z$. By construction this means exactly one of two things must hold: either $z - \cluster(s) > \Lambda(n)$ or we observed $\AlgBalancedEstimateRatio(\mathcal{G}, z,\cluster(s), \dots, {\beta_{\cluster(s)}},\dots)=\infty$ during the construction of the forest. Let us consider these two situations individually. Formally, let $A=\{s\in H_v \mid z-\cluster(s) > \Lambda(n)\}$ and $B=H_v \setminus A$. Consider any $b\in B$. We have, $1\leq z-\cluster(b) \leq \Lambda(n)$. Moreover,
\begin{align}\label{eq:non-adaptive-real-distant-b-cz}
    w_b \leq \frac{7}{\alpha} \cdot w_{c_{\cluster(b)}} \leq \frac{7}{\alpha} \cdot \beta_{\cluster(b)} \cdot w_{c_z} \leq \frac{\alpha\cdot \epsilon_1}{7\cdot |C_{\cluster(b)}| \cdot \Lambda(n)} \cdot w_{c_z},
\end{align}
where the first inequality is by definition of $(\frac{7}{\alpha}, \frac{1}{\alpha}, \epsilon_2)$-cluster graph, in the second inequality we use that, by \Cref{lem:correctness-of-ratio-estimator}, $\frac{w_{c_z}}{w_{c_{\cluster(b)}}} \geq \frac{1}{\beta_{\cluster(b)}}$, and the last inequality is by definition of $\beta_{\cluster(b)}$.

Let $K=\{\cluster(b) \mid b\in B\}$. By definition, $|K|\leq \Lambda(n)$, indeed, $K$ can only contain values in $\{z-1, \dots, z-\Lambda(n)\}$. Thus,
\begin{align}\label{eq:non-adaptive-real-distant-bound-B}
    \sum_{b\in B} \frac{w_b}{w_{c_z}} \overset{\eqref{eq:non-adaptive-real-distant-b-cz}}{\leq} \sum_{b\in B} \frac{\alpha\cdot \epsilon_1}{7\cdot |C_{\cluster(b)}|\cdot \Lambda(n)} = \frac{\alpha\epsilon_1}{7} \sum_{k\in K} \sum_{b\in C_k} \frac{1}{|C_k|\cdot \Lambda(n)} = \frac{\alpha\epsilon_1}{7} \cdot \frac{|K|\cdot|C_k|}{|C_k|\cdot\Lambda(n)} \leq \frac{\alpha\epsilon_1}{7}.
\end{align}
Consider now any $a \in A$. Since $z - \cluster(a) \geq \Lambda(n)+1$, by the properties of $(\frac{7}{\alpha}, \frac{1}{\alpha}, \epsilon_2)$-cluster graph we have that $w_{c_z} \geq \frac{1}{\alpha^{\Lambda(n)+1}} \cdot w_{c_{\cluster(a)}}$, and therefore:
\begin{align*}
    w_a \leq \frac{7}{\alpha} \cdot w_{c_{\cluster(a)}} \leq 7\alpha^{\Lambda(n)} \cdot w_{c_z} \leq 7\cdot \alpha^{\log_{1/\alpha}(\frac{49\cdot n}{\alpha\cdot \epsilon_1})} \cdot w_{c_z} \leq \frac{\alpha\cdot \epsilon_1}{7n} \cdot w_{c_z}.
\end{align*}
Thus, 
\begin{equation}\label{eq:non-adaptive-real-distant-bound-A}
    \sum_{a\in A} \frac{w_a}{w_{c_z}} \leq \frac{|A|\alpha\epsilon_1}{7n} \leq \frac{\alpha}{7}\cdot\epsilon_1. 
\end{equation}    
Note that, by the properties of the cluster graph and because $\cluster(u)\geq z$, we have that: $w_u \geq \frac{\alpha}{7} w_{c_{\cluster(u)}} \geq \frac{\alpha}{7} w_{c_z}$, and therefore,
\[
    \sum_{s\in H_v} \frac{w_s}{w_u} \leq \sum_{s\in H_v} \frac{7}{\alpha} \cdot \frac{w_s}{w_{c_z}} = \sum_{a\in A} \frac{7}{\alpha}\cdot \frac{w_a}{w_{c_z}} + \sum_{b\in B} \frac{7}{\alpha}\cdot\frac{w_b}{w_{c_z}} \overset{\eqref{eq:non-adaptive-real-distant-bound-B},\eqref{eq:non-adaptive-real-distant-bound-A}}{\leq} 2\epsilon_1\leq \epsilon.
\]

We now consider the case where $u$ and $v$ are in different connected components. Let $z$ be the minimum index such that $c_z$ is in the same connected component as $u$. By construction, $\cluster(u) \geq z > \cluster(v)$. Similarly to before, we can partition $H_v$ into $A= \{s\in H_v \mid z-\cluster(s) > \Lambda(n)\}$ and $B=H_v \setminus A$. By construction, for each $b\in B$, we observed $\AlgBalancedEstimateRatio(\dots, z,\cluster(b), \dots, {\beta_{\cluster(b)}},\dots)=\infty$. Thus, with an argument identical to before, we can prove $\frac{w_b}{w_{c_z}} \leq \frac{\alpha \epsilon_1}{7|C_{\cluster(b)}|\Lambda(n)}$ for each $b\in B$ and $\frac{w_a}{w_{c_z}} \leq \frac{\alpha\epsilon_1}{7n}$ for each $a\in A$, and this concludes the proof as before.
\end{proof}

We now show that even if we compute estimates along long paths, the estimates are still well-behaved. Again, \Cref{fig:proof-tree} can be used as a reference for the disposition of the vertices. 

\begin{lemma}\label{lem:non-adaptive-distant-vertices-estimates}
Let $\epsilon, \alpha, \delta\in(0,1)$, and suppose that each call to $\AlgBalancedEstimateRatio$ as well as the call to $\AlgQuicksortClustering$ is successful. Let $\mathcal{F}=(F=([n], E), r, (C_1, \dots, C_T), (c_1, \dots, c_T))$ be the output of $\AlgBuildForestNonAdaptive(\alpha, \epsilon, \delta)$. Suppose that $u,v\in[n]$ are in the same connected component $\mathcal{C}$ of $F$, and $\cluster(u)>\cluster(v)$ and $d(u,v)>5$. Then, $\sum_{s\in K_v} r(P(s,u)) \leq \epsilon$, where $K_v = \{s\in \mathcal{C} \mid \cluster(s)\leq \cluster(v)\}$.
\end{lemma}
\begin{proof}
Since $d(u,v)\geq 6$, $d(c_{\cluster(u)}, c_{\cluster(v)})\geq 4$. Note that $\mathcal{C}$ is a tree. Let $R\in[T]$ be the maximum index such that $c_R\in \mathcal{C}$ and suppose the tree is rooted at $c_R$. Let $c_z$ be the ancestor of $c_{\cluster(v)}$ at distance 2 from $c_{\cluster(v)}$. Let $c_x$ be a vertex such that, (i) $\cluster(u) \geq x \geq z$, (ii) $d(c_{\cluster(u)}, c_x) \leq 2$, and (iii) $c_x$ is an ancestor of $c_{\cluster(v)}$. Note that there is always a sibling of $c_{\cluster(u)}$ with these properties (potentially, it might also be $c_x=c_{\cluster(u)}$ or $c_x=c_z$). Since $w_{c_{\cluster(u)}} \geq w_{c_x}$ and $d(u,c_x)\leq 3$, we have,
\begin{equation}\label{eq:non-adaptive-distant-est-P-cx-u}
    r(P(c_{x}, u)) \overset{\text{\Cref{lem:non-adaptive-close-vertices}}}{\leq}  (1+\epsilon_1)^3 \frac{w_{c_x}}{w_u}\leq (1+\epsilon_1)^3 \frac{w_{c_{\cluster(u)}}}{w_{u}} \leq \frac{7(1+\epsilon_1)^3}{\alpha},
\end{equation}
 where the last inequality follows by the properties of $(\frac{7}{\alpha}, \frac1\alpha, \epsilon_2)$-cluster graph. 

Note now that, if $c_i$ is an ancestor of $c_j$ for $i> j$, then, by construction, $r(P(c_i, c_j)) \geq \frac{1}{\alpha^{i-j}}$. Since $c_x$ is an ancestor of $c_z$, we have:
\begin{equation}\label{eq:ancestor-alpha-inequality}
    r(P(c_z, c_x)) = 1/r(P(c_x,c_z)) \leq \alpha^{x-z}.
\end{equation}

Let $K_v = \{s\in \mathcal{C} \mid \cluster(s) \leq \cluster(v)\}$. Let $A=\{s\in K_v \mid z-\cluster(s) > \Lambda(n)\}$ and $B= K_v \setminus A$. Note that $c_z$ is an ancestor of any vertex in $K_v$. Consider any $a\in A$, by using the properties of the $(\frac{7}{\alpha}, \frac{1}{\alpha},\epsilon_2)$-cluster graph, we have
\begin{align}\label{eq:non-adaptive-distant-est-P-a-cz}
    r(P(a, c_z)) &= r(a,c_{\cluster(a)})\cdot r(P(c_{\cluster(a)},c_z)) \leq (1+\epsilon_2)\cdot \frac{w_{a}}{w_{c_{\cluster(a)}}}\cdot \alpha^{z-\cluster(a)}  \leq 7(1+\epsilon_2) \alpha^{z-\cluster(a)-1} \nonumber\\
    & \leq 7(1+\epsilon_2) \alpha^{\Lambda(n)}  \leq 7(1+\epsilon_2) \alpha^{\log_{1/\alpha}(\frac{49\cdot n}{\alpha\cdot \epsilon_1})} \leq \frac{\alpha(1+\epsilon_1)\epsilon_1}{7n}.
\end{align}
Then, we have,
\begin{align}\label{eq:non-adaptive-distant-est-P-a-u}
    r(P(a, u)) &= r(P(a, c_z))\cdot r(P(c_z, c_x)) \cdot r(P(c_x, u)) \overset{\eqref{eq:ancestor-alpha-inequality}}{\leq} r(P(a, c_z))\cdot \alpha^{x-z} \cdot r(P(c_x, u)) \nonumber\\
    &\leq r(P(a, c_z)) \cdot r(P(c_x, u)) \overset{\eqref{eq:non-adaptive-distant-est-P-cx-u},\eqref{eq:non-adaptive-distant-est-P-a-cz}}{\leq} \frac{7\alpha(1+\epsilon_1)^4\epsilon_1}{7\alpha n} \leq \frac{2\epsilon_1}{n}, 
\end{align}
where we used that $\epsilon_1 \leq \frac{1}{10}$, and thus $(1+\epsilon_1)^4 < 2$.

Consider now any $b\in B$. Let $c_y$ be the ancestor of $c_{\cluster(b)}$ at distance 2 from $c_{\cluster(b)}$. Note that $c_y$ exists because $d(c_z, c_{\cluster(b)}) \geq 2$, and we have $z \geq y$ (it might also be $c_z=c_y$). In particular, $c_y$ is still a descendant of $c_z$. Since $z-\cluster(b) \leq \Lambda(n)$, we also have $y-\cluster(b) \leq \Lambda(n)$. Since $c_y$ is an ancestor of  $c_{\cluster(b)}$, $d(c_y, c_{\cluster(b)}) =2$ and $y-\cluster(b) \leq \Lambda(n)$, during the construction of the tree, it must have happened that $\AlgBalancedEstimateRatio(\dots, y, \cluster(b), \dots, \beta_{\cluster(b)}, \dots)=\infty$. But then, by \Cref{lem:correctness-of-ratio-estimator},

\begin{equation}\label{eq:non-adaptive-distant-est-cy-ccb}
    \frac{w_{c_y}}{w_{c_{\cluster(b)}}} \geq \frac{1}{\beta_{\cluster(b)}}.
\end{equation}
Putting these observations together, and also by using the definition of $(\frac{7}{\alpha}, \frac{1}{\alpha},\epsilon_2)$-cluster graph and \Cref{lem:non-adaptive-close-vertices}, we obtain:
\begin{align}\label{eq:non-adaptive-distant-est-P-b-cz}
    r(P(b,c_z)) & = r(P(b,c_{\cluster(b)})) \cdot r(c_{\cluster(b)}, c_y) \cdot r(c_y, c_z) \leq \frac{7(1+\epsilon_2)}{\alpha} \cdot (1+\epsilon_1)^2 \cdot \frac{w_{c_{\cluster(b)}}}{w_{c_y}} \cdot \alpha^{z-y} \nonumber\\
    & \overset{\eqref{eq:non-adaptive-distant-est-cy-ccb}}{\leq} 7(1+\epsilon_1)^3 \cdot \beta_{\cluster(b)} \cdot \alpha^{z-y-1} \leq \frac{(1+\epsilon_1)^3 \cdot \epsilon_1 \cdot \alpha^{z-y+1}}{7 \cdot |C_{\cluster(b)}| \cdot \Lambda(n)} \leq \frac{\alpha \cdot (1+\epsilon_1)^3 \cdot \epsilon_1}{7\cdot |C_{\cluster(b)}| \cdot \Lambda(n)}.
\end{align}
Observe that, since $\epsilon_1 \leq \frac{1}{10}$, it holds that $(1+\epsilon_1)^6 < 2$, and we have that,
\begin{align}\label{eq:non-adaptive-distant-est-P-b-u}
    r(P(b,u)) &= r(P(b,c_z)) \cdot r(P(c_z, c_x))\cdot r(P(c_x, u)) \overset{\eqref{eq:non-adaptive-distant-est-P-b-cz},\eqref{eq:ancestor-alpha-inequality},\eqref{eq:non-adaptive-distant-est-P-cx-u}}{\leq} \frac{\alpha \cdot (1+\epsilon_1)^3 \cdot \epsilon_1}{7\cdot |C_{\cluster(b)}| \cdot \Lambda(n)} \cdot \alpha^{x-z} \cdot \frac{7(1+\epsilon_1)^3}{\alpha} \nonumber\\
    & \leq \frac{ (1+\epsilon_1)^6 \epsilon_1}{ |C_{\cluster(b)}| \Lambda(n)} \leq \frac{2\epsilon_1}{|C_{\cluster(b)}| \Lambda(n)}.
\end{align}
Let $\mathcal{K}=\{\cluster(s) \mid s \in B\}$. Note that $|\mathcal{K}|\leq \Lambda(n)$. Finally, we have,
\begin{align*}
    \sum_{s\in K_v} r(P(s,u)) &= \sum_{a\in A} r(P(a,u))  + \sum_{b\in B} r(P(b,u)) \overset{\eqref{eq:non-adaptive-distant-est-P-a-u},\eqref{eq:non-adaptive-distant-est-P-b-u}}{\leq} \frac{2|A|\epsilon_1}{n} + \sum_{k\in \mathcal{K}} \sum_{b \in C_k} \frac{2\epsilon_1}{|C_k|\Lambda(n)} \\
    & \leq 2\epsilon_1 + \frac{2|\mathcal{K}|\epsilon_1}{\Lambda(n)} \leq 4\epsilon_1 \leq \epsilon. \qedhere
\end{align*}
\end{proof}

\begin{proof}[Proof of \Cref{thm:non-adaptive-estimation-forest-creation}] 
With probability at least $1-\frac{\delta}{4}$, $\AlgQuicksortClustering$ correctly computes a cluster graph and queries each pair at most $O(\frac{\log(n/\delta)}{\epsilon^2})$ times. Moreover, each call made to $\AlgBalancedEstimateRatio$ is correct with probability at least $1-\frac{\delta}{4n^2}$. We make no more than $n^2$ such calls, so they are all correct with probability at least $1-\frac{\delta}{4}$. Thus, by the union bound, with probability at least $1-\frac{\delta}{2}$, both the $\AlgQuicksortClustering$ call and all the $\AlgBalancedEstimateRatio$ calls are correct. Conditioning on this event, we have that \Cref{lem:non-adaptive-close-vertices}, \Cref{lem:non-adaptive-distant-vertices-real}, and \Cref{lem:non-adaptive-distant-vertices-estimates} hold and this gives the first three properties for a $(5,\epsilon)$-estimation-forest except for the last part of the third point.

Note that vertices with the same cluster index are at distance at most two and for any tree in the forest, if we let $x$ (resp. $y$) be the minimum (resp. maximum) cluster index in the tree, then the vertex set of the tree is $\cup_{i=x}^y C_i$. This ensures that all four properties are satisfied and therefore the algorithm returns a $(5,\epsilon)$-estimation-forest. Moreover, by \Cref{lem:non-adaptive-query-cost-algorithm}, each pair is queried at most $O(\frac{\log (n/\epsilon)\log(n/\delta)}{\epsilon^3})$ times with probability at least $1-\frac{\delta}{2}$. Thus, with probability at least $1-\delta$ the algorithm is both correct and has the desired query complexity.
\end{proof}

This in turn completes the proof of \Cref{thm:adaptive-balanced}.

\section{Lower Bounds}\label{sec:lower-bounds}
\newcommand{\maxindex}[1]{i(#1)}
In order to prove lower bounds on the query complexity of the MNL learning task we consider the following family of instances. 

We denote by $Sym(n)$ the set of permutations of $[n]$. For any even $n$, given a permutation $\pi\in Sym(n)$ we shall think of $\pi$ as a way of partitioning the set $[n]$ into $n/2$ pairs $P^{\pi}_1, \dots , P^{\pi}_{n/2}$ where:
\[
    P^{\pi}_i := \{\pi(2i-1), \pi(2i)\}.
\]
Given a non-empty set $S\subseteq[n]$ its \emph{highest pair} with respect to $\pi$ is the pair $P_{\maxindex{\pi,S}}^\pi$ where:
\[
    \maxindex{\pi,S} := \max_{P_i^\pi \cap S \neq \varnothing} i. 
\]
When $\pi$ is the identity permutation, we simply write $P_i$ and $\maxindex{S}$.

\begin{restatable}[Matching pseudo-MNL]{definition}{MatchingPseudoMNL}\label{def:matching-pseudo-mnl}
Given an even number $n \in \N$, a vector $\vec{p} \in [0,1]^{n/2}$, and a permutation $\pi \in Sym(n)$, the \textit{matching pseudo-MNL} $\overline{M}(n, \vec{p},\pi)$ supported on $[n]$ has the following $\maxsample$ distributions. The winner of a slate $S$ is always an item of its highest pair $P_{\maxindex{\pi,S}}^\pi = \{\pi(2\cdot \maxindex{\pi,S}-1), \pi(2\cdot \maxindex{\pi,S})\}$. If only one of these items belongs to $S$ then that item is the winner. Otherwise the winner is chosen to be $\pi(2\cdot \maxindex{\pi,S})$ with probability $\vec{p}_{\maxindex{\pi,S}}$ and $\pi(2\cdot \maxindex{\pi,S}-1)$ with the remaining probability.
\end{restatable}

As a shorthand, we will denote by $\overline{M}(n,\vec{p}):=\overline{M}(n,\vec{p},\identityperm)$ where $\identityperm:[n]\to [n]$ is the identity permutation.

\begin{figure}
    \centering
    \tikzset{every picture/.style={line width=0.75pt}} %

\begin{tikzpicture}[x=0.75pt,y=0.75pt,yscale=-1,xscale=1]
\draw  [fill={rgb, 255:red, 0; green, 0; blue, 0 }  ,fill opacity=1 ] (109,106) .. controls (109,102.69) and (111.69,100) .. (115,100) .. controls (118.31,100) and (121,102.69) .. (121,106) .. controls (121,109.31) and (118.31,112) .. (115,112) .. controls (111.69,112) and (109,109.31) .. (109,106) -- cycle ;
\draw  [fill={rgb, 255:red, 0; green, 0; blue, 0 }  ,fill opacity=1 ] (109,55) .. controls (109,51.69) and (111.69,49) .. (115,49) .. controls (118.31,49) and (121,51.69) .. (121,55) .. controls (121,58.31) and (118.31,61) .. (115,61) .. controls (111.69,61) and (109,58.31) .. (109,55) -- cycle ;
\draw  [fill={rgb, 255:red, 0; green, 0; blue, 0 }  ,fill opacity=1 ] (189,106) .. controls (189,102.69) and (191.69,100) .. (195,100) .. controls (198.31,100) and (201,102.69) .. (201,106) .. controls (201,109.31) and (198.31,112) .. (195,112) .. controls (191.69,112) and (189,109.31) .. (189,106) -- cycle ;
\draw  [fill={rgb, 255:red, 0; green, 0; blue, 0 }  ,fill opacity=1 ] (189,55) .. controls (189,51.69) and (191.69,49) .. (195,49) .. controls (198.31,49) and (201,51.69) .. (201,55) .. controls (201,58.31) and (198.31,61) .. (195,61) .. controls (191.69,61) and (189,58.31) .. (189,55) -- cycle ;
\draw  [fill={rgb, 255:red, 0; green, 0; blue, 0 }  ,fill opacity=1 ] (349,106) .. controls (349,102.69) and (351.69,100) .. (355,100) .. controls (358.31,100) and (361,102.69) .. (361,106) .. controls (361,109.31) and (358.31,112) .. (355,112) .. controls (351.69,112) and (349,109.31) .. (349,106) -- cycle ;
\draw  [fill={rgb, 255:red, 0; green, 0; blue, 0 }  ,fill opacity=1 ] (349,55) .. controls (349,51.69) and (351.69,49) .. (355,49) .. controls (358.31,49) and (361,51.69) .. (361,55) .. controls (361,58.31) and (358.31,61) .. (355,61) .. controls (351.69,61) and (349,58.31) .. (349,55) -- cycle ;
\draw  [fill={rgb, 255:red, 0; green, 0; blue, 0 }  ,fill opacity=1 ] (439,106) .. controls (439,102.69) and (441.69,100) .. (445,100) .. controls (448.31,100) and (451,102.69) .. (451,106) .. controls (451,109.31) and (448.31,112) .. (445,112) .. controls (441.69,112) and (439,109.31) .. (439,106) -- cycle ;
\draw  [fill={rgb, 255:red, 0; green, 0; blue, 0 }  ,fill opacity=1 ] (439,55) .. controls (439,51.69) and (441.69,49) .. (445,49) .. controls (448.31,49) and (451,51.69) .. (451,55) .. controls (451,58.31) and (448.31,61) .. (445,61) .. controls (441.69,61) and (439,58.31) .. (439,55) -- cycle ;
\draw    (70,50) -- (70,118) ;
\draw [shift={(70,120)}, rotate = 270] [color={rgb, 255:red, 0; green, 0; blue, 0 }  ][line width=0.75]    (10.93,-3.29) .. controls (6.95,-1.4) and (3.31,-0.3) .. (0,0) .. controls (3.31,0.3) and (6.95,1.4) .. (10.93,3.29)   ;
\draw [shift={(70,48)}, rotate = 90] [color={rgb, 255:red, 0; green, 0; blue, 0 }  ][line width=0.75]    (10.93,-3.29) .. controls (6.95,-1.4) and (3.31,-0.3) .. (0,0) .. controls (3.31,0.3) and (6.95,1.4) .. (10.93,3.29)   ;
\draw    (103,176) -- (467,176) ;
\draw [shift={(469,176)}, rotate = 180] [color={rgb, 255:red, 0; green, 0; blue, 0 }  ][line width=0.75]    (10.93,-3.29) .. controls (6.95,-1.4) and (3.31,-0.3) .. (0,0) .. controls (3.31,0.3) and (6.95,1.4) .. (10.93,3.29)   ;
\draw (260,61.4) node [anchor=north west][inner sep=0.75pt]  [font=\LARGE]  {$...$};
\draw (99,121.4) node [anchor=north west][inner sep=0.75pt]    {$\pi ( 1)$};
\draw (98,21.4) node [anchor=north west][inner sep=0.75pt]    {$\pi ( 2)$};
\draw (179,121.4) node [anchor=north west][inner sep=0.75pt]    {$\pi ( 3)$};
\draw (178,21.4) node [anchor=north west][inner sep=0.75pt]    {$\pi ( 4)$};
\draw (327,121.4) node [anchor=north west][inner sep=0.75pt]    {$\pi ( n-3)$};
\draw (327,21.4) node [anchor=north west][inner sep=0.75pt]    {$\pi ( n-2)$};
\draw (414,123.4) node [anchor=north west][inner sep=0.75pt]    {$\pi ( n-1)$};
\draw (431,23.4) node [anchor=north west][inner sep=0.75pt]    {$\pi ( n)$};
\draw (0,50) node [anchor=north west][inner sep=0.75pt]   [align=left] {\begin{minipage}[lt]{37.32pt}\setlength\topsep{0pt}
\begin{center}
the winner is \\sampled\\ according to $\vec{p}$
\end{center}

\end{minipage}};
\draw (162,188) node [anchor=north west][inner sep=0.75pt]   [align=left] {the winner is always a right-most element};

\end{tikzpicture}
    \caption{The instance $\overline{M}(n,\vec{p},\pi)$ used for the lower bounds. Among two items in the same pair $\pi(2i-1)$ and $\pi(2i)$ the latter wins with probability $\vec{p}_i$. The winner is always an item of the right-most pair.}
    \label{fig:lower-bound}
\end{figure}
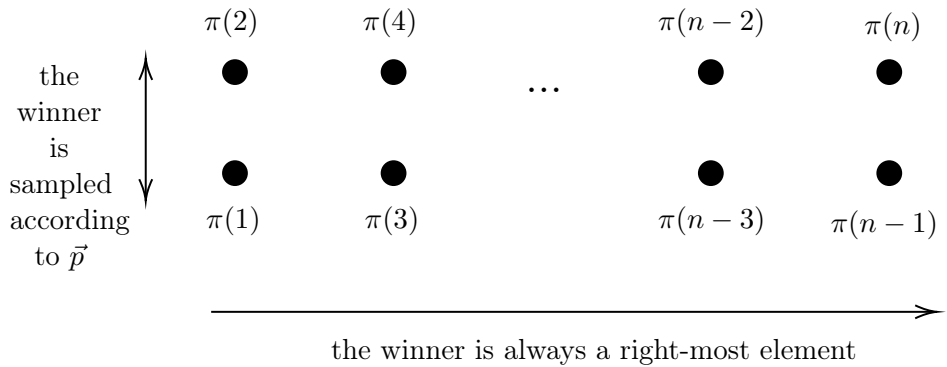

Note that the objects introduced in Definition~\ref{def:matching-pseudo-mnl} are not technically MNLs, but they are limits of a sequences of MNLs, and in particular, any lower bound on algorithms to learn objects of this kind immediately implies a lower bound on learning MNLs, as per the following result (which we prove in Appendix~\ref{sec:extension-to-pseudo-mnl}).

\begin{restatable}{proposition}{LearningMatchingPseudoMNL}\label{prop:learning-matching-psuedo-mnl}
    Suppose there exists a potentially randomized algorithm $\mathcal{A}$ that takes as input $\varepsilon, \delta\in(0,1)$ and access to a $\maxsample$ oracle for an MNL $M$ supported on $[n]$, and after making at most $m(n,\varepsilon,\delta)$ queries outputs an MNL $\hat{M}$ such that:
    \[
        \Pr[d_\infty(M,\hat{M}) \leq \varepsilon] \geq 1-\delta.
    \]
    Then, the same algorithm, when given as input $\varepsilon,\delta\in(0,1)$ and access to a $\maxsample$ oracle for a matching pseudo-MNL $\overline{M}$ makes at most $m(n,\epsilon,\delta)$ queries and outputs an MNL $\hat{M}$ such that:
    \[
        \Pr[d_\infty(\overline{M},\hat{M}) \leq \varepsilon] \geq 1-\delta.
    \] 
\end{restatable}

An analogous result can be proved for $\dmul$, up to a factor of four in front of $\epsilon$: any algorithm that learns any MNL within $\dmul$ error of $\epsilon$, also learns any matching pseudo-MNL within $\dmul$ error of $4\epsilon$.

    \subsection{Lower Bound for Approximate Coin Selection}
    
    \begin{definition}[Approximate top-$n\over 2$ coin selection problem]
        Let $n\geq 2$ be even, $\varepsilon\in(0,1)$ and $\delta\in (0,1)$. The $(\varepsilon, \delta)$-approximate top-${n\over 2}$ coin selection problem is defined as follows. Given access to $n$ Bernoulli distributions with unknown parameters $p_1, \dots , p_n$, find, with probability at least $1-\delta$ a subset $C\subseteq [n]$, with $|C| = n/2$ such that:
        \[
            \forall c \in C : \hspace{3mm} p_c \geq p_{*} - \varepsilon,
        \]
        where $p_*$ is the $(n/2)$th largest element of $ \{p_1, \dots ,p_n\}$.
    \end{definition} 

    This problem is also known, in the multi-armed bandits literature, as the Explore-$m$ problem or PAC subset selection in Stochastic Bandits \citep{ktas12} for $m = n/2$. We can therefore leverage the lower bounds on this latter problem to obtain the following theorem.
    
    \begin{theorem}[Theorem 8 from \cite{ktas12}---Paraphrased]\label{thm:lower-bound-on-bandit-problem}
        Let $n\geq 2$ be even, $\varepsilon \in \left(0, \sqrt{1\over 32}\right)$, $\delta \in\left(0,{1\over 4e}\right)$ then, for any (potentially randomized and adaptive) algorithm $\mathcal{A}$ for the $(\varepsilon,\delta)$-approximate top-{$n\over 2$} coin selection problem there is an instance on which the algorithm requires at least ${1\over 18375}\cdot  {n\over \varepsilon^2} \ln\left( {n\over 16\cdot \delta}\right)$ coin tosses.
    \end{theorem}

    The previous result holds even for the expected number of coin tosses of the algorithm and moreover the hard instance is guaranteed to have biases in $\{\frac{1}{2}, \frac{1}{2}-2\epsilon, \frac{1}{2}+2\epsilon\}$ (see \cite[Theorem 3.1]{ck19}). We also observe the following simple corollary, that we will use for our lower bound for $\dmul$.

    \begin{theorem}\label{thm:rare-coin-lower-bound}
        Let $n\geq 2$ be even, $\eta\in\left(0,\sqrt\frac{1}{32}\right)$, $\rho\in(0,1]$, and $\delta\in\left(0,\frac{1}{4e}\right)$. Consider the $(\rho\eta,\delta)$-approximate top-$\frac{n}{2}$ coin selection problem restricted to instances with biases $q_j=\rho \cdot p_j$, where $p_j \in \{\frac{1}{2}, \frac{1}{2}-2\eta,\frac{1}{2}+2\eta\}$ for each $j\in[n]$. Any (possibly randomized and adaptive) algorithm that solves this problem with probability at least $1-\delta$ must make at least $\frac{1}{18375}\cdot \frac{n}{\rho\eta^2} \ln\left(\frac{n}{16\delta}\right)$ coin tosses in the worst case.
    \end{theorem}
    \begin{proof}
    Suppose there exists an algorithm $\mathcal A$ that solves the stated problem with at most $T$ coin tosses in the worst case.  We construct an algorithm $\mathcal{B}$ for the $(\eta,\delta)$-approximate top-$n/2$ coin selection problem on the original biases $p_1,\ldots,p_n$ having expected sample complexity $\rho \cdot T$.
    
    Algorithm $\mathcal B$ simulates $\mathcal A$. Whenever $\mathcal A$ requests a toss of coin $j$, algorithm $\mathcal B$ independently samples $b\sim\Bern(\rho)$. If $b=0$, it returns $0$ to $\mathcal A$. Otherwise, it tosses the original coin $j$ and returns its outcome. At the end, $\mathcal B$ returns the set identified by $\mathcal A$.
    
    Each simulated toss is therefore distributed as $\Bern(\rho \cdot p_j)$, exactly as required by $\mathcal A$. Therefore, with probability at least $1-\delta$, the set $C$ returned satisfies, for each $c\in C$: 
    \begin{align*}
        & q_c \geq q_* - \rho\eta\\
        \iff & \rho \cdot p_c \geq \rho \cdot p_* - \rho \eta \\
        \iff & p_c \geq p_* - \eta.
    \end{align*}
    Thus, $\mathcal B$ succeeds
    with probability at least $1-\delta$. By \Cref{thm:lower-bound-on-bandit-problem} (and Theorem 3.1 of \cite{ck19}), the expected number of coin tosses made by $\mathcal{B}$ on an instance where each bias is in the set $\{1/2, 1/2+2\eta, 1/2-2\eta\}$ must be at least: $\frac{1}{18375}\cdot \frac{n}{\eta^2} \ln\left(\frac{n}{16\delta}\right)$. 

    Since $\mathcal{A}$ makes at most $T$ coin tosses, the expected number of coins that get actually tossed by $\mathcal{B}$ is $\rho \cdot T$. Therefore, we have:
    \begin{align*}
        &\rho \cdot T \geq \frac{1}{18375}\cdot \frac{n}{\eta^2} \ln\left(\frac{n}{16\delta}\right) \\
        \iff &T \geq \frac{1}{18375}\cdot \frac{n}{\rho \eta^2} \ln\left(\frac{n}{16\delta}\right),
    \end{align*}
    which concludes the proof.
    \end{proof}

    \subsection{Lower Bounds for Adaptive Algorithms}\label{sec:adaptive-lower-bound}
    In this section, we prove lower bounds for adaptive algorithms. We start with the following lower bound.

    \AdaptiveLowerBound*

    The lower bound follows directly from the following lemma as well as \Cref{thm:lower-bound-on-bandit-problem}.
    \begin{lemma}\label{lem:adaptive-bandit-reduction}    
        Let $n\in\mathbb{N}$ be a multiple of $4$, and let $\varepsilon\in(0,\frac12)$,  $\delta\in(0,1)$. Suppose there is a (potentially randomized and adaptive) algorithm $\mathcal{A}$ that for any MNL $M$ on $[n]$ makes at most $m(n,\varepsilon,\delta)$ queries and then outputs an MNL $\hat{M}$ which satisfies $d_\infty(M,\hat{M})\leq \varepsilon$ with probability at least $1-\delta$. Then there exists an algorithm $\mathcal{B}$ for the $(2\varepsilon,\delta)$-approximate top-${n\over 4}$ coin selection problem with worst-case query complexity $m(n,\varepsilon,\delta)$.
    \end{lemma}
    \begin{proof}
    Let $\vec{p}=p_1, \dots, p_{n/2}$ be the unknown parameters of the $n/2$ Bernoulli distributions for an instance of the $(\epsilon,\delta)$-approximate top-$\frac{n}{4}$ coin selection problem. By \Cref{prop:learning-matching-psuedo-mnl}, algorithm $\mathcal{A}$ must also satisfy $\Pr[d_\infty(\overline{M}(n,\vec{p}),\hat{M}) \leq \epsilon] \geq 1-\delta$, where $\overline{M}(n,\vec{p})$ is the matching pseudo-MNL of \Cref{def:matching-pseudo-mnl}. Moreover, $\mathcal{A}$ makes at most $m(n,\varepsilon,\delta)$ queries in this latter setting too.

    We now describe an algorithm $\mathcal{B}$ that solves the instance of $(\epsilon,\delta)$-approximate top-$\frac{n}{4}$ coin selection with at most $m(n,\epsilon,\delta)$ queries.
    Intuitively, $\mathcal{B}$ simulates algorithm $\mathcal{A}$ with access to $\overline{M}(n,\vec{p})$. 
    
    Let $P_{\maxindex{S}}$ be the highest pair in $S$. Whenever $\mathcal{A}$ makes a query $S_j \subseteq [n]$, if $|P_{\maxindex{S_j}} \cap S_j|=1$ then $\mathcal{B}$ returns the unique item in $P_{\maxindex{S_j}} \cap S_j$; if instead $|P_{\maxindex{S_j}} \cap S_j|=2$, then $\mathcal{B}$ samples from the $\maxindex{S_j}$th Bernoulli distribution (which has parameter $\vec{p}_{\maxindex{S_j}}$), and obtains a bit $b_j$ and then it returns $2\cdot \maxindex{S_j} -1+b_j$ to $\mathcal{A}$.

    When $\mathcal{A}$ terminates, it outputs the weights of an MNL $\hat{M}$. $\mathcal{B}$ then sorts the elements of $[n/2]$ into a sequence $(s_1, \dots, s_{n/2})$ so that:
    \[
        \hat{M}_{\{2s_{1}-1,2s_1\}}(2s_1) \geq \cdots \geq \hat{M}_{\{2s_{n/2}-1,2s_{n/ 2}\}}(2s_{n / 2})
    \]
    and returns $\{s_1, \dots ,s_{n/4}\}$.

    In order to see that this satisfies the required guarantees, we first note that, by construction, $\mathcal{B}$'s responses to $\mathcal{A}$'s queries are distributed like the responses of a $\maxsample$ oracle for $\overline{M}(n,\vec{p})$. Therefore, under the conditioning that $d_{\infty}(\overline{M}(n,\vec{p}), \hat{M}) \leq \varepsilon$, we have:
    
    \[
        \hat{M}_{\{2s_{i}-1,2s_i\}}(2s_i) \in \vec{p}_{s_i} \pm \varepsilon.
    \]
    Hence, for all $i\in [n/4]$:
    \[
        \vec{p}_{s_i} \geq \hat{M}_{\{2s_{i}-1,2s_i\}}(2s_i)-\varepsilon\geq \hat{M}_{\{2s_{n/4}-1,2s_{n/4}\}}(2s_{n/4})-\varepsilon \geq \vec{p}_{s_{n/4}}-2\varepsilon.
    \]
    Since $\Pr[d_\infty(\overline{M}(n,\vec{p}),\hat{M}) \leq \epsilon] \geq 1-\delta$, we have that $\mathcal{B}$ solves the $(2\varepsilon,\delta)$-approximate top-$n\over 4$ coin selection problem with $m(n,\epsilon,\delta)$ queries.
    \end{proof}
    
    For the $\dmul$-error we can prove a stronger lower bound in terms of $\epsilon$. Specifically, we have the following result, whose proof is very similar to the previous one but we use \Cref{thm:rare-coin-lower-bound} instead of \Cref{thm:lower-bound-on-bandit-problem} for our reduction.

    \AdaptiveLowerBoundDmul*
    \begin{proof}
    Let $\mathcal A$ be an algorithm that with probability $1-\delta$ learns any MNL within $\dmul$ error of $\epsilon$ by making at most $m(n,\epsilon, \delta)$ queries. By \Cref{prop:learning-matching-psuedo-mnl}, the same algorithm must learn any matching pseudo-MNL within a $\dmul$ error of $4\epsilon$ by making the same number of queries (see also \Cref{sec:extension-to-pseudo-mnl}).

    Let $n$ be a multiple of $4$. Define $k=n/2$, $\rho=16\varepsilon$, and $\eta=8\epsilon$. We construct an algorithm $\mathcal B$ that makes at most $m(n,\epsilon, \delta)$ queries and solves the $(\rho\eta, \delta)$-approximate top-$k/2$ coin selection problem restricted to biases $q_i=\rho \cdot p_i$ where each $p_i$ is in $\{1/2-2\eta, 1/2, 1/2+2\eta\}$. Note that, by \Cref{thm:rare-coin-lower-bound}, this is enough to conclude the proof. By our choice of parameters,
    $
        q_j\in
        \left\{
            8\varepsilon-256\varepsilon^2,\,
            8\varepsilon,\,
            8\varepsilon+256\varepsilon^2
        \right\}.
    $
    Since $\varepsilon\leq 1/128$, we have:
    \begin{equation*}
        6\varepsilon\leq q_j \leq 10\varepsilon \qquad \forall j\in[k].
    \end{equation*}
    
    Let $\overline M=\overline M(n,\vec q)$ be the matching pseudo-MNL of
    \Cref{def:matching-pseudo-mnl}, where item $2j$ has winning probability $q_j$ in the pair $P_j=\{2j,2j-1\}$ for each $j\in[k]$. Algorithm $\mathcal B$ simulates $\mathcal A$ with oracle access to $\overline M$, exactly as in the proof of \Cref{lem:adaptive-bandit-reduction}. By construction, this perfectly simulates the $\maxsample$ oracle of $\overline M$, using at most one coin toss per query. Thus, $\mathcal B$ uses at most $m(n,\varepsilon,\delta)$ coin tosses. When $\mathcal A$ terminates and outputs
    an MNL $\hat M$, define
    \[
        \hat q_j
        :=\hat M_{P_j}(2j)
        \qquad \forall j\in[k].
    \]
    Algorithm $\mathcal B$ sorts the coins in non-increasing order of $\hat q_j$ and returns the first $k/2$ coins.
    
    We now prove that $\mathcal B$ solves $(\rho\eta, \delta)$-approximate top-$k/2$ selection problem with probability at least $1-\delta$. Note that, with probability $1-\delta$, we have:
    \[
        \dmul(\overline M,\hat M)
        \leq 4\varepsilon.
    \]
    
    Fix any $j\in[k]$ and consider the slate $P_j=\{2j, 2j-1\}$. Since $\overline{M}_{P_j}(2j), \overline{M}_{P_j}(2j-1)\geq 6\epsilon$, we have that $\hat{M}$ must approximate multiplicatively both these winning probabilities (otherwise $\dmul$ would be at least $6\epsilon$). Therefore, we have:
    \[
        (1-4\varepsilon)q_j
        \leq \hat q_j
        \leq (1+4\varepsilon)q_j.
    \]
    Since $q_j \leq 10\epsilon$, we have:
    \begin{equation*}
        |\hat q_j-q_j|
        \leq 4\varepsilon q_j
        \leq 40\varepsilon^2
        \qquad \forall j\in[k].
    \end{equation*}
    
    Let $q_*$ denote the $k/2$-th largest
    bias among $q_1,\ldots,q_k$. Since there are at least $k/2$ coins whose true bias is at least $q_*$ and each of them is estimated within error $40\epsilon^2$, we have that each returned item satisfies $\hat{q}_j\geq q_*-40\epsilon^2$. Therefore, for each returned item we have:
    \begin{align*}
        q_j &\geq \hat{q_j} - 40\epsilon^2 \geq q_* - 80\epsilon^2 \geq q_* - \rho\eta,
    \end{align*}
    which concludes the proof as $\mathcal{B}$ returns a solution to the $(\rho\eta,\delta)$-approximate top-$k/2$ selection problem.
    \end{proof}
    
    \subsection{Lower Bounds for Non-Adaptive Algorithms}\label{sec:non-adaptive-lower-bound}
    In this section, we prove the following:
    
    \NonAdaptiveLowerBound*
    We first introduce some key definitions. Note that every $\maxsample$ query made by an algorithm can be identified with a subset $S_j \subseteq [n]$. We shall assume without loss of generality that all algorithms only make queries of cardinality at least $2$. Consider an algorithm that has access to a $\maxsample$ oracle for a matching pseudo-MNL $M(n,\vec{p}, \pi)$.

    We say that query $S_j$ \emph{highlights} the pair $P_i^\pi$ if $P_i^\pi$ is the highest pair in $S_j$ and $P_i^\pi \subseteq S_j$. We have the following result. 
    \begin{lemma}\label{lem:low-highlighting}
        Let $\mathcal{A}$ be a potentially randomized non-adaptive MNL learning algorithm that makes at most $m$ queries. Consider the process of running $\mathcal{A}$ with $\maxsample$ access to the matching pseudo-MNL $\overline{M}(n,\vec{p}, \pi)$ of \Cref{def:matching-pseudo-mnl}, where $\pi\sim Sym(n)$ is a permutation of $[n]$ chosen uniformly at random, and $\vec{p}\in[0,1]^{n/2}$. Let $\Lambda$ be the event that the number of queries made by the algorithm that highlight one of the pairs $P^\pi_1, \dots ,P^\pi_{\left\lfloor {n\over12e}\right\rfloor}$ is less than or equal to ${10m \over 7n}$, then:
        \[
            \Pr[\Lambda] \ge {9\over 10}.
        \]
    \end{lemma}
    \begin{proof}
        It is sufficient to show that the statement holds for any fixed, deterministic choice of the queries $S_1, \dots , S_m$, as this will imply it holds for random queries, since $\pi$ and $(S_1, \dots ,S_m)$ are independent of each other.
       
        For every $i\in [n/2]$ let $M_i$ be the random variable defined as:
        \[
            M_i := |\{j\in[m]\mid S_j \text{ highlights }P_i^\pi\}|.
        \]

        Let $X_{ij}$ be the indicator random variable of the event that the query $S_j$ highlights the pair $P_i^\pi$. Consider first any query $S_j$ with $|S_j|>2$ and $i \leq {n\over 12 e}$. We have:
    \begin{align*}
        \E{}{X_{ij}} &= \Pr_\pi\left[S_j \text{ highlights }P_i^\pi \right]\\
        & = \frac{\binom{2i-2}{|S_j| - 2}}{\binom{n}{|S_j|}}\\
        &\leq \left(\frac{2i-2}{|S_j|-2}\right)^{|S_j|-2}\left(\frac{|S_j|}{n}\right)^{|S_j|}e^{|S_j|-2}\\ %
        &= \left(\frac{2i-2}{|S_j|-2}\cdot \frac{|S_j|}{n}\right)^{|S_j|-2}\left(\frac{|S_j|}{n}\right)^{2}e^{|S_j|-2}\\
        &\leq \left(3 \cdot \frac{2i-2}{n}\right)^{|S_j|-2}\left(\frac{|S_j|}{n}\right)^{2}e^{|S_j|-2}\\
        &\leq \left({1\over 2}\right)^{|S_j|-2}\left(\frac{|S_j|}{n}\right)^{2} & \left(\text{since }i \leq {n\over 12 e}\right)\\
        &\leq {9\over 2n^2},
    \end{align*}

    where the first inequality follows by the fact that $\left(\frac{n}{k}\right)^k\leq \binom{n}{k}$ for $1\leq k\leq n$ and $\binom{n}{k}\leq \left(\frac{n\cdot e}{k}\right)^k$ for $k\geq 1$, and the last step follows from the fact that the function $f(x) = x^22^{2-x}$ takes its maximum value over the positive integers at $x = 3$.
    Consider now any query $S_j$ with $|S_j|=2$ we have:
    \begin{align*}
        \E{}{X_{ij}} &= \Pr_\pi[S_j\text{ highlights }P_i^\pi] = {1\over \binom{n}{2}} = {2\over n(n-1)} \leq {9 \over 2n^2}.
    \end{align*}
    where the last inequality holds for all $n \geq 2$. Summing up over all terms, we have:
    \begin{align}\label{eq:expected-number-of-useful-queries}
        \E{\pi}{\sum_{i = 1}^{\floor{n\over 12 e}}  M_i} = \E{\pi}{\sum_{i = 1}^{\floor{n\over 12 e}}  \sum_{j=1}^m{X_{ij}}}  =\sum_{i = 1}^{\floor{n\over 12 e}}  \sum_{j=1}^m \E{\pi}{{X_{ij}}} \leq {9m \over 24 e \cdot n} \leq {m \over 7n}.
    \end{align}
    Let $\Lambda$ be the event that $\sum_{i = 1}^{\floor{n\over 12 e}}  M_i \leq {10m \over 7n}$. Finally, by Markov's inequality and \eqref{eq:expected-number-of-useful-queries}: $\Pr[\Lambda] \geq {9\over 10}$.
    \end{proof}

    \begin{lemma}\label{lem:non-adaptive-bandit-reduction}
        Consider any even $n\geq150$, and let $\varepsilon\in(0,\frac12)$,  $\delta\in(0,\frac{9}{10})$. Suppose there is a (potentially randomized) non-adaptive algorithm $\mathcal{A}$ that for any MNL $M$ on $[n]$ makes at most $m(n,\varepsilon,\delta)$ non-adaptive queries and then outputs an MNL $\hat{M}$ which satisfies $d_\infty(M,\hat{M})\leq \varepsilon$ with probability at least $1-\delta$. Then there exists an algorithm $\mathcal{B}$ for the $(2\varepsilon,\delta+\frac{1}{10})$-approximate top-${n_1\over 2}$ coin selection problem with at most $\frac{10\cdot m(n,\varepsilon,\delta)}{n}$ queries, where $n_1$ is the even number in $\{\floor{\frac{n}{12e}}, \floor{\frac{n}{12e}}-1\}$.
    \end{lemma}
    \begin{proof}
        Algorithm $\mathcal{B}$ is given access to $n_1$ Bernoulli distributions with parameters $q_1, \dots, q_{n_1}$ and needs to produce a subset of $n_1 \over 2$ indices in $[n_1]$ corresponding to the approximate top-$\frac{n_1}{2}$ items. We construct $\mathcal{B}$ as follows. First, $\mathcal{B}$ samples a uniformly random permutation $\pi \in Sym(n)$. Then, it simulates $\mathcal{A}$ and obtains a set of queries, given as a multiset $S_1, \dots , S_m$ of subsets of $[n]$. Then, it constructs a set $a_1, \dots, a_m$ of responses to the queries as follows. If the query $S_j$ highlights one of the pairs (the pair $P^\pi_{\maxindex{\pi,S_j}}$) and this pair is in $\{P_1^\pi, \dots, P^\pi_{n_1}\}$, then $\mathcal{B}$ samples $x_j \in \{0,1\}$ from the $\maxindex{\pi, S_j}$th Bernoulli distribution (with parameter $q_{\maxindex{\pi,S_j}}$) and sets $a_j = \pi(2\cdot \maxindex{\pi,S_j}-1+x_j)$. If the query $S_j$ highlights another pair $P^{\pi}_{\maxindex{\pi, S_j}}$ where $\maxindex{\pi, S_j} > n_1$ then $\mathcal{B}$ samples $x_j\sim \{0,1\}$ uniformly at random and returns $\pi(2 \cdot \maxindex{\pi,S_j}-i+x_j)$. Finally, if the query $S_j$ does not highlight its highest pair with respect to $\pi$, then $\mathcal{B}$ sets $a_j$ to the unique item in $S_j\cap P^\pi_{\maxindex{\pi,S_j}}$.

        Then, $\mathcal{B}$ feeds the responses $a_1, \dots, a_m$ back to $\mathcal{A}$, and $\mathcal{A}$ outputs an MNL $\hat{M}$. Finally, $\mathcal{B}$ sorts the elements of $[n_1]$ into a sequence $s_1, \dots, s_{n_1}$ satisfying:
        \[
            \hat{M}_{\{\pi(2s_1-1), \pi(2s_1)\}}(\pi(2s_1)) \geq \cdots \geq\hat{M}_{\{\pi(2s_{n_1}-1), \pi(2s_{n_1})\}}(\pi(2s_{n_1}))
        \]
        and outputs $s_1, \dots, s_{n_1/2}$. 
        
        We now prove that $\mathcal{B}$ is correct with high probability. First, observe that $\mathcal{B}$ is simulating access to $\maxsample$ oracle for the matching pseudo-MNL $\overline{M}(n,\vec{p},\pi)$, where $\pi$ is chosen uniformly at random and:
        \[
            \vec{p_i}:= \begin{cases}
                q_i &\text{if }i \leq n_1\\
                {1\over 2} &\text{if $n_1 < i \leq \frac{n}{2}$.}
            \end{cases}
        \]
        By \Cref{prop:learning-matching-psuedo-mnl}, we have that, with probability at least $1-\delta$, $d_\infty(\overline{M}(n,\vec{p}, \pi),\hat{M})\leq \varepsilon$. If this event happens, $s_1, \dots , s_{n_1/2}$ is a correct solution to the $(2\epsilon,\delta)$-approximate top-$n_1 \over 2$ coin selection instance---this can be proved with essentially the same argument as \Cref{lem:adaptive-bandit-reduction}.%
        
        As it is written, the number of queries made by $\mathcal{B}$ to the Bernoulli distributions could be higher than $10m \over 7n$. Hence, in order to meet the requirements of the lemma, we modify $\mathcal{B}$ slightly. We introduce the following exception to the description above: if $\mathcal{B}$ ever needs to make more than $10m \over 7 n$ queries to the Bernoulli distributions, it will instead output a uniformly random subset of $[n_1]$ of size $n_1\over 2$ and terminate. By \Cref{lem:low-highlighting} this happens with probability at most ${1\over 10}$, and the lemma follows.
    \end{proof}
    \Cref{thm:non-adaptive-lower-bound} then follows from \Cref{lem:non-adaptive-bandit-reduction} and \Cref{thm:lower-bound-on-bandit-problem}.

\section{Conclusions and Open Problems}\label{sec:conc}

In this paper, we considered the problem of learning an unknown MNL by making queries to a $\maxsample$ oracle so that the learned weights can be used to provide an estimate to the distribution of each slate within an $\ell_1$-error of $\epsilon$. We developed two algorithms for this task: one for the adaptive setting and one for the non-adaptive setting. 

Our adaptive algorithm has a query complexity of $O(\frac{n \log n}{\epsilon^2})$ for $\delta=\frac{1}{\poly(n)}$, which is matched by our lower bound of $\Omega(\frac{n \log n}{\epsilon^2})$. We also considered a stronger notion of learning, where the goal is to provide a multiplicative estimate to the distribution of each slate (except for sub-slates whose total probability is less than $\epsilon$), and showed that for this stronger notion the optimal query complexity is $\Theta(\frac{n \log n}{\epsilon^3})$.

Our non-adaptive algorithm has a query complexity of $O(\frac{n^2 \cdot  \log n\cdot  \log(n/\epsilon)}{\epsilon^3})$ nearly matching our $\Omega(\frac{n^2 \log n}{\epsilon^2})$ non-adaptive lower bound. Closing this gap is a natural open problem. %

Finally, our non-adaptive algorithm is based on an adaptive algorithm that queries each pair at most polylogarithmic many times. However, the latter is different from the two other adaptive algorithms we specifically designed for the adaptive setting. A possible direction for future work would be to find a single algorithm which can be used to match the query complexity of our algorithms in both the adaptive and non-adaptive setting. That is, an $O(n \log n)$ adaptive algorithm that queries each pair at most $\polylog(n)$ many times. 

\section*{AI Usage Disclosure}
All the results except for the $O(n \log n / \epsilon^2)$ algorithm for $d_1$ (\Cref{sec:adaptive-optimal-eps-algorithm}) and the $\Omega(n \log n /\epsilon^3)$ lower bound for $\dmul$ (\Cref{thm:rare-coin-lower-bound} and \Cref{thm:adaptive-lower-bound-dmul}) were obtained by the authors and were published in the conference version of this paper \citep{cgklptt26}. We obtained the other two results after extensive interactions with ChatGPT 5.6 Sol, which was given access to our conference paper and our attempted ideas. The final proofs were interpreted and written by the authors (with LLMs used for polishing the wording). The authors take full responsibility for all the work. 

\appendix
\section{Missing Proofs for \Cref{sec:algorithmic-primitives}}\label{sec:appendix-algo-primitives}

We first recall the following standard concentration result %
 (see, e.g.,\ \cite{blm13} Equation 2.10, page 36).
\begin{theorem}[Bernstein's Inequality]
Let $X_1, \dots, X_N$ be i.i.d. r.v.'s in $[0, 1]$ and each with mean $\mu$ and variance $\sigma^2$. Then, for $\lambda>0$,
\begin{align*}
    \Pr\left[ \frac{1}{N} \sum_{i=1}^N X_i - \mu \ge \lambda \right] &\leq \exp\left(-\frac{\lambda^2 N}{2\sigma^2+\frac{2}{3}\lambda}\right) \\ \Pr\left[ \frac{1}{N} \sum_{i=1}^N X_i - \mu \le -\lambda \right] &\leq \exp\left(-\frac{\lambda^2 N}{2\sigma^2+\frac{2}{3}\lambda}\right).
\end{align*}
\end{theorem}

We will also use the following Chernoff-Bound (see, e.g., \cite[Theorem 1.1]{dp09}):
\begin{theorem}[Multiplicative Chernoff Bound]\label{thm:multiplicative-chernoff}
Let $X_1, \dots, X_N$ be i.i.d.\ r.v.'s in $[0,1]$ and each with mean $\mu$. Then, for $0<\delta<1$,
\[
    \Pr\left[\left|\frac{1}{N}\sum_{i=1}^N X_i - \mu \right| \geq \delta \cdot \mu \right] \leq 2\exp\left( -\frac{\delta^2 \mu  \cdot N}{3} \right).
\]
\end{theorem}

\CompareGuarantees*
\begin{proof}%
The bound on the number of queries follows directly from the pseudocode of the algorithm. We argue that the guarantees of the lemma hold with probability at least $1-{\delta\over 2}$ for $\hat{p}_i$. By symmetry and the union bound, this implies that they hold for both $\hat{p}_i$ and $\hat{p}_j$ with probability at least $1-\delta$.

Let $p = M_{\{i,j\}}(i)$ and $X_1, \dots, X_t\sim\Bern(p)$ be the indicator random variables of the events that each call to $\maxsample(\{i,j\})$ returns $i$, so that $\hat{p}_i = \frac{1}{m}\sum_{i=1}^m X_i$ and $\E{}{\hat{p}_i}= p$.

We divide the proof into three cases depending on the value of $p$, for all cases we prove that the guarantees hold with probability at least $1-{\delta \over 2}$. 

First, suppose that $p\leq c/4$. In this case, the algorithm can only fail if it returns $\hat{p}_i \neq 0$. By Bernstein's inequality, we find that the probability that $\hat{p}_i\neq 0$ is at most:
\begin{align*}
    \Pr[\hat{p}_i \geq c/2] &\leq \Pr[\hat{p}_i \geq p  + c/4] \leq \exp\left(-\frac{m(c/4)^2}{2p(1-p) + \frac{c}{6}} \right) \leq \exp\left(-\frac{m(c/4)^2}{\frac{c}{2} + \frac{c}{6}} \right) \\
    & = \exp\left(-\frac{3\cdot m\cdot c}{32}\right) \leq \frac{\delta}{6} \leq {\delta \over 2}.
\end{align*}
Second, suppose $c/4 < p < c$. In this case, the algorithm can fail if it returns a value of $\hat{p}_i$ that simultaneously satisfies $\hat{p}_i \ne 0$ and $\hat{p}_i\notin(1\pm \epsilon)p$. We now show that, if $p\geq c/4$, then the probability that $\hat{p}_i\notin(1\pm \epsilon)p$ is at most $\delta/2$. By the multiplicative Chernoff bound (\Cref{thm:multiplicative-chernoff}):
\begin{align}\label{eq:epsilon-bound-from-chernoff}
\Pr[|p - \hat{p}_i| \geq \epsilon p] \leq 2e^{-\frac{\epsilon^2\cdot m\cdot p}{3}} \leq 2e^{-\frac{\epsilon^2\cdot m\cdot c}{12}} \leq \frac{\delta}{3} \leq \frac{\delta}{2}.
\end{align}
So if $c/4 < p < c$ the guarantees hold with probability at least $1-{\delta\over 2}$.

Finally, consider the case $p\geq c$. In this case, the algorithm can fail either if $\hat{p}_i=0$ or if $\hat{p}_i\notin(1\pm \epsilon)p$. By \eqref{eq:epsilon-bound-from-chernoff}, the second event happens with probability at most $\delta/3$. By Bernstein's inequality, the first event happens with probability:
\begin{align}\label{eq:compare-general-bernstein-bound}
    \Pr[\hat{p}_i \leq c/2] & = \Pr[\hat{p}_i \leq p - (p-c/2)] \leq \exp\left( - \frac{m(p - c/2)^2}{2p(1-p) + \frac{2}{3}\cdot (p - c/2)} \right).
\end{align}
Here we consider two possibilities. Suppose that $c \leq p \leq 2c$. From \eqref{eq:compare-general-bernstein-bound}, we obtain:
\begin{align*}
    \Pr[\hat{p}_i \leq c/2] &\leq \exp\left( - \frac{m(c/2)^2}{4c + c} \right) = \exp\left( - \frac{m \cdot c}{20} \right) \leq \frac{\delta}{6}.
\end{align*}
If, instead, we have $p\geq 2c$, then $c/2 \leq p/4$. By \eqref{eq:compare-general-bernstein-bound}, we have:
\begin{align*}
    \Pr[\hat{p}_i \leq c/2] &\leq \exp\left( - \frac{m(p- p/4)^2}{2p(1-p) + \frac{2}{3} p} \right) \leq \exp\left( - \frac{m(\frac{3}{4}\cdot p)^2}{2p + p} \right) = \exp\left(-\frac{3\cdot m\cdot p^2}{16\cdot p}\right)\\
    &= \exp\left(-\frac{3\cdot m\cdot p}{16}\right) \leq \exp\left(-\frac{3\cdot m\cdot c}{8}\right) \leq \frac{\delta}{6}.
\end{align*}

Thus, when $p\geq c$, the algorithm can fail with probability at most $\frac{\delta}{3}+\frac{\delta}{6} = \frac{\delta}{2}$. %
\end{proof}

\EstimateRatioGuarantees*
\begin{proof}%
    We begin by noting that the query complexity bound for $\AlgEstimateRatio$ follows directly from the query complexity of $\AlgCompare$. 
   We now prove the rest of the guarantees. We will assume that $(\hat{p}_i,\hat{p}_j)$ satisfy the three guarantees in Lemma~\ref{lem:compare-guarantees} and show that under this assumption, $r(i,j)$ and $1/r(i,j)$ satisfy the four conditions in the statement of this lemma. Since the former happens with probability at least $1-\delta$ (by Lemma~\ref{lem:compare-guarantees}), Lemma~\ref{lem:weight-pair-ratio-estimate} will then follow.

    If ${w_i \over w_j} \leq {\alpha \over {3\alpha +4}}$ then:
    \[
        M_{\{i,j\}}(i)={w_i \over w_i + w_j} = {1 \over 1 + {w_j\over w_i}}\leq  {1 \over 1 + {3 \alpha +4\over \alpha}} = {\alpha \over 4(\alpha+1)} = {c\over 4},
    \]
    and hence $\hat{p}_i = 0$, which implies $r(i,j) = 0$. If ${w_i \over w_j}\ge {3\alpha +4\over \alpha}$ then ${w_j \over w_i} \leq {\alpha \over 3\alpha +4}$ and the same argument shows $r(i,j) = \infty$, yielding the first two conditions of this lemma.

    If $\alpha \le {w_i \over w_j } $, then:
    \[
        M_{\{i,j\}}(i) = {w_i \over w_i+w_j} \ge {w_i \over w_i+{1\over \alpha}w_i} = {\alpha \over \alpha +1} = c,
    \]
    and hence $\hat{p}_i$ satisfies:
    \begin{equation}\label{eq:pi-guarantee}
        \hat{p}_i \in \left(1\pm \frac{\varepsilon}{3}\right) \frac{w_i}{w_i + w_j},
    \end{equation}
    so that  $\hat{p}_i\neq 0$ giving that $r(i,j)\neq 0$.
    Similarly if ${w_i \over w_j } \le {1\over \alpha}$, then:
    \[
        M_{\{i,j\}}(j) = {w_j \over w_i+w_j} \ge {w_j \over {1\over \alpha }w_j+w_j} = {\alpha\over \alpha+1} = c,
    \]
    
    and hence we obtain estimates $\hat{p}_j$ of the winning probability of $j$  in the slate $\{i,j\}$ satisfying:
    
    \begin{equation}\label{eq:pj-guarantee}
        \hat{p}_j \in {\left(1\pm \frac{\varepsilon}{3}\right)} \frac{w_j}{w_i + w_j}.
    \end{equation}
    so that $r(i,j) \neq \infty$. This implies condition 3. 
    
    Moreover, if the algorithm ever outputs $r(i,j)$ different from $0$ and $\infty$, it must have been the case that both $\hat{p}_i$ and $\hat{p}_j$ were not zero. By Lemma~\ref{lem:compare-guarantees} this implies $\hat{p}_i$ and $\hat{p}_j$ satisfy \eqref{eq:pi-guarantee} and \eqref{eq:pj-guarantee}, and hence, %
    $r(i,j) = {\hat{p}_i \over \hat{p}_j}$ satisfies:
    \[ 
        r(i,j) = {\hat{p}_i \over \hat{p}_j} \leq {(1+\frac{\varepsilon}{3}) \over (1-\frac{\varepsilon}{3})}\cdot {w_i \over w_j} \leq \left(1 + 3\cdot \frac{\varepsilon}{3}\right) {w_i \over w_j} = (1 + \varepsilon) {w_i \over w_j}
    \]
    and:
    \[
         r(i,j) = {\hat{p}_i \over \hat{p}_j} \ge {(1-\frac{\varepsilon}{3}) \over (1+\frac{\varepsilon}{3})}\cdot {w_i \over w_j} \ge \left(1 - 2\cdot \frac{\varepsilon}{3}\right) {w_i \over w_j} \geq (1 - \varepsilon) {w_i \over w_j},
    \]
    where we used that $\frac{1+a}{1-a} \leq 1+3a$ and $\frac{1-a}{1+a} \geq 1-2a$ for $a\in(0,1/3)$. Similarly, ${1\over r(i,j)} = {\hat{p}_j \over \hat{p}_i}$ satisfies:
    \[ 
        {1\over r(i,j)} = {\hat{p}_j \over \hat{p}_i} \leq {(1+\frac{\varepsilon}{3}) \over (1-\frac{\varepsilon}{3})}\cdot {w_j \over w_i} \leq \left(1 + 3\cdot \frac{\varepsilon}{3}\right) {w_j \over w_i} = (1 + \varepsilon) {w_j \over w_i}
    \]
    and:
    \[
         {1\over r(i,j)} = {\hat{p}_j \over \hat{p}_i} \ge {\left(1-\frac{\varepsilon}{3}\right) \over (1+\frac{\varepsilon}{3})}\cdot {w_j \over w_i} \ge \left(1 - 2\cdot \frac{\varepsilon}{3}\right) {w_j \over w_i} \geq (1 - \varepsilon) {w_j \over w_i},
    \]
    yielding condition 4 above.
\end{proof}

\section{Missing Proofs for Section~\ref{sec:adaptive-algorithm}: Computing Approximate Orderings}\label{sec:computing-epsilon-orderings}
In this section we prove the \Cref{cor:falahatgar-for-MNLs}. Specifically, this is a simple corollary of a result of \cite{fjopr18} which considered the following notion of ordering: 

\begin{definition}[Additive $\beta$-ordering]
    An additive $\beta$-ordering for an MNL $M$ supported on $[n]$ is an ordering $(s_1, \dots , s_n)$ of the items of $[n]$ such that, for any pair $i,j$ with $i<j$:
    \[
        M_{\{s_i, s_j\}}(s_i) \leq \frac{1}{2} + \beta.  
    \]
\end{definition}
They showed that such an ordering can be computed efficiently. The following is an adaptation of their result, where we boost the success probability and make the running time explicit.

\newcommand{\AlgFalBinarySearchRanking}{\AlgName{Binary-Search-Ranking}}
\newcommand{\AlgFalRankCheck}{\AlgName{Rank-Check}}
\newcommand{\AlgFalBuildBinarySearchTree}{\AlgName{Build-Binary-Search-Tree}}
\newcommand{\AlgFalIntervalBinarySearch}{\AlgName{Interval-Binary-Search}}
\begin{theorem}[Adaptation of Theorem 9 of \cite{fjopr18}]\label{thm:additive-ordering-fal-adaptation}
Choose any $\beta\in(0,1/2)$, $\delta\in(0,1)$. There exists an algorithm that, with probability at least $1-\delta$, makes $O\left(\frac{n \cdot \log(n/\delta) }{\beta^2}\cdot \left(1+\frac{\log(1/\delta)}{\log n}\right)\right)$ queries and returns an additive $\beta$-ordering. Moreover, the running time of the algorithm is proportional to the query complexity and the algorithm only queries pairs.
\end{theorem}
\begin{proof}
    Theorem 9 of \cite{fjopr18} provides an algorithm $\AlgFalBinarySearchRanking([n], \beta)$ that, with probability at least $1-1/n$, makes at most $O(\frac{n\log n}{\beta^2})$ queries and returns an additive $\beta$-ordering.\footnote{To be precise, \cite{fjopr18} showed this result in another model which generalizes MNLs (see, e.g., Appendix A of \cite{ybkj12}).} Lemma 23 of \cite{fjopr18} provides an algorithm $\AlgFalRankCheck(\pi, \beta, \delta)$ such that, with error probability at most $\delta$: (i) if $\pi$ is an additive $\beta$-ordering of $[n]$ it returns true, (ii) if $\pi$ is not an additive $3\beta$-ordering of $[n]$ it returns false. If $\pi$ is an additive $\beta'$-ordering, with $\beta'\in(\beta,3\beta)$, \AlgFalRankCheck{} can return either true or false. Moreover, \AlgFalRankCheck{} always makes $O\left(\frac{n}{\beta^2} \log({n\over \delta})\right)$ queries. 
    
    We now use these two subroutines to boost the success probability of the algorithm. Let $\eta=\ceiling{\frac{\ln(2/\delta)}{\ln(n)}}$. We run $\eta$ independent instances of $\AlgFalBinarySearchRanking([n], \beta/3)$ in parallel. Whenever an instance terminates and outputs an ordering $\pi$, we run $\AlgFalRankCheck(\pi, \frac{\beta}{3}, \frac{\delta}{2\eta})$ and if it returns true, we give $\pi$ in output, otherwise we ignore $\pi$ and continue running the other instances. 

    Observe that \AlgFalRankCheck{} is run at most $\eta$ times, and therefore, with probability at least $1-\delta/2$ all the outputs given by this subroutine are correct. Observe also that, with probability at least $1-\delta/2$, at least one of the $\eta$ instances of \AlgFalBinarySearchRanking{} makes at most $O(\frac{n\log n}{\beta^2})$ queries and returns an additive $\beta/3$-ordering. When this happens, \AlgFalRankCheck{} will surely return true. Moreover, if \AlgFalRankCheck{} were to return true even before that, it means that the returned $\pi$ must be an additive $\beta'$-ordering for $\beta'\in(\beta/3, \beta)$. Therefore, the output returned is correct. Moreover, the total query complexity is given by:
    \begin{align*}
        O\left(\eta \cdot \frac{n \log n}{\beta^2} + \eta\cdot \frac{n \log (\frac{n \eta}{\delta})}{\beta^2}\right) \leq O\left(\eta\cdot \frac{n \log (n/\delta)}{\beta^2} \right) \leq O\left(\frac{n \log(n/\delta)}{\beta^2} \cdot \left(1 + \frac{\log(1/\delta)}{\log n}\right)\right).
    \end{align*}

    Regarding the running time, by inspecting the pseudocode, one can see that \AlgFalRankCheck{} runs in time proportional to the query complexity. Algorithm \AlgFalBinarySearchRanking{} consists of several subroutines, some of which are provided in \citep{fops17}. All but two points of \AlgFalBinarySearchRanking{} run in time proportional to the query complexity. 
    
    First, the subroutine \AlgFalBuildBinarySearchTree{} \citep{fops17} runs in time $O\left(\frac{n}{\log^3n}\right)$ and therefore its running time can be charged to the query complexity. Specifically, note that, despite the fact that this subroutine is called multiple times in the pseudocode; the binary search tree generated is always the same, and hence it is sufficient to call it once. More precisely, their algorithm takes in input an integer $N = O\left(\frac{n}{\log^3n}\right)$ and builds a complete binary tree $T$, where each vertex maintains three values $(l, m, r)$, where it always holds $m = \lceil \frac{l + r}{2} \rceil$. The root starts with $l=1$ and $r=N$, and the generic vertex $(l,m,r)$, with $r-l>1$, generates a left child with values $(l, \lceil \frac{l+m}{2} \rceil, m)$ and a right child with values $(m, \lceil \frac{m+r}{2} \rceil, r)$.

    Second, line 4.b.(i) of the subroutine \AlgFalIntervalBinarySearch{} \citep{fops17} sorts an array of integers $Q$ such that $|Q|=O(\log n)$. This instruction is repeated $O(n)$ times. If implemented as written, this would lead to an $O(n \log n \log\log n)$ runtime which is not chargeable to the query complexity. However, we can exploit how $Q$ is constructed to make this step more efficient. Specifically, $Q$ is constructed as follows: the algorithm starts on the root of $T$ and it always moves to an adjacent vertex (either the parent or one of the two children); each time a vertex with values $(l,m,r)$ is visited, the values $l,m,r$ are added to $Q$. This random walk goes on for at most $O(\log n)$ times. We can therefore do the following: throughout the execution of this walk in $T$, we maintain a parallel binary search tree $T'$ which is a copy of $T$ but containing only the vertices visited during the random walk. Thus, $T'$ can be constructed in time at most $O(|Q|)$. After the random walk, we can obtain a sorted array of the (unique) values in $Q$ as follows: first put a value 1 (which will always be present in $Q$), then run an in-order visit of $T'$ and output the midpoint $m$ of each visited vertex, finally append a value $N$ at the end (which will always be present in $Q$). Thus, this sorted array can be constructed in $O(|Q|)$. We note that outputting the sorted array without the values multiplicity (i.e.,\ where each distinct value appears exactly once) is sufficient to correctly perform the other steps of the algorithm. However, for completeness, we remark that it would not be difficult to adapt the algorithm to work with multiplicity. In particular, we can do so as follows: in each vertex of $T'$ we can save the multiplicity of the midpoint value $m$ and two pointers pointing to the vertices that have the left (resp.\ right) value as their midpoint (except for values $1$ and $N$ which are dealt with separately)---all these pointers are easy to maintain. Then, each time the random walk moves to a new vertex, we need to update three counters, an operation which requires constant time. In summary, since this step takes time $O(|Q|)$, its time complexity can be charged to the query complexity of \AlgFalIntervalBinarySearch{}, and in general, it requires at most $O(n\log n)$-time overall. 
    
     By inspection one can also see that the algorithm only queries pairs.
\end{proof}

We can now prove the result for $\orderingerror$-orderings:
\EpsilonOrderingAlgorithm*
\begin{proof}
Use \Cref{thm:additive-ordering-fal-adaptation} to compute an additive $\beta$-ordering $(s_1, \dots, s_n)$ for $\beta=\frac{\orderingerror}{4}$. This ordering is also an $\orderingerror$-ordering. Indeed, for $i<j$,
\begin{align*}
    {w_{s_i}\over w_{s_i} + w_{s_j}} \leq {1\over 2} + \beta & \iff {w_{s_i}} \leq \left({1\over 2} + \beta\right)\left(w_{s_i} + w_{s_j}\right) \iff  \left({1-2\beta\over 1+2\beta}\right){w_{s_i}} \leq w_{s_j}\\
    &\implies (1-4\beta) w_{s_i} \leq w_{s_j} \iff (1-\orderingerror)w_{s_i} \leq w_{s_j},
\end{align*}
where we used $1-4\beta \leq {1-2\beta\over 1+2\beta}$ for $\beta\in [0,1]$.
\end{proof}

\section{Missing Proofs for \Cref{sec:adaptive-optimal-eps-algorithm}}\label{sec:appendix-optimal-eps2}

We will use the following concentration bound for binomial random variables (see, e.g., \cite[Theorem 1.1 and Theorem 1.2]{dp09}):
\begin{theorem}[Concentration of Binomial Random Variables]\label{thm:binomial-concentration}
    Let $X \sim \Bin(n,p)$ for $p\in(0,1)$. Then, for any $t>0$, the following holds with probability at least $1-2e^{-t}$:
    \[
        \E{}{X} - \sqrt{2\cdot  \E{}{X} \cdot t} \leq X \leq \E{}{X} + \sqrt{2\cdot \E{}{X} \cdot t} + \frac{2}{3}\cdot t.
    \]
\end{theorem}

\EstimateWeakRatioGuarantees*
\begin{proof}
Define for convenience $p:=w_j / (w_i + w_j)$ and $\mu = m \cdot p$. Since $w_i\geq w_j$, we have $p\leq 1/2$ and
\begin{equation}
\label{eq:weak-p-ratio}
    \frac{w_j}{2w_i}\leq p\leq\frac{w_j}{w_i}.
\end{equation}
Note that $X \sim \Bin(m, p)$ and $\E{}{X} =\mu$. By \Cref{thm:binomial-concentration}, choosing $t:=L=\ln(2/\delta)$, we have, with probability at least $1-2e^{-L}=1-\delta$:
\begin{equation}
\label{eq:weak-chernoff}
    \mu-\sqrt{2\mu L}
    \leq X
    \leq \mu+\sqrt{2\mu L}+\frac{2L}{3}.
\end{equation}
We condition on this event and distinguish the
two possible outputs of the algorithm.

\paragraph{Case 1: $r(i,j)=\infty$.}
In this case, the algorithm observed $X<64L$. Note that it must be $\mu<128L$. Indeed, if it was
$\mu\geq128L$, then
\eqref{eq:weak-chernoff} would give: $X \geq \mu - \sqrt{2\mu L} > 64L$, which is a contradiction. By \eqref{eq:weak-p-ratio} and the choice
of $m$, we have:
\[
\begin{aligned}
    \frac{w_j}{w_i}
    &\leq 2p
    =\frac{2\mu}{m}<\frac{256L}{m}\leq
    \frac{256\cdot L\cdot \epsilon^2}{512\cdot Z\cdot L}=\frac{\epsilon^2}{2Z},
\end{aligned}
\]
which proves the second guarantee.

\paragraph{Case 2: $r(i,j)\neq\infty$.}
In this case, $X\geq64L$. Note that $\mu \geq 32 L$, otherwise by \eqref{eq:weak-chernoff} we would have: $X \leq 32L  + \sqrt{64L^2} + \frac{2}3{L < 64L}$. 

Recall that the algorithm
computes:
\[
    \hat p=\frac{X-4\sqrt{XL}}{m},
    \qquad
    r(i,j)=\frac{1-\hat p}{\hat p}.
\]
We start by showing that $\hat{p}$ is well defined, and in particular that $0 \leq \hat{p} \leq p \leq 1/2$. 

Since $X\geq64L$, we have
$
    X-4\sqrt{XL}\geq\frac{X}{2}>0,
$
and hence $\hat p>0$. Moreover, by $\mu\geq32 \cdot L$ and
\eqref{eq:weak-chernoff} we have:
\begin{equation}
\label{eq:weak-X-range}
    \frac{3\mu}{4}\leq X\leq2\mu.
\end{equation}
Note that:
\begin{align*}
    X-\mu& \overset{\eqref{eq:weak-chernoff}}{\leq} \sqrt{2\mu L}+\frac{2L}{3} \overset{\eqref{eq:weak-X-range}}{\leq} \sqrt{\frac{8XL}{3}}+\frac{2L}{3} \overset{(X \geq 64L)}{\leq} \left(\sqrt{\frac83}+\frac1{12}\right)
    \sqrt{XL}
    <4\sqrt{XL}.
\end{align*}
Therefore:
\[
    \hat p
    =\frac{X-4\sqrt{XL}}{m}
    \leq\frac{\mu}{m}
    =p.
\]
We now focus on the quantity $1-\frac{w_i}{w_j \cdot r(i,j)}$. Note that since $\hat{p}\leq p$, we have:
\[
    \frac{w_i}{w_j \cdot r(i,j)} = \frac{\hat{p}(1-p)}{p(1-\hat{p})} \leq 1,
\]
and therefore $1-\frac{w_i}{w_j \cdot r(i,j)} \geq 0$. On the other hand:
\begin{align*}
    1-\frac{w_i}{w_j \cdot r(i,j)} &= 1- \frac{\hat{p}(1-p)}{p(1-\hat{p})} = \frac{p- \hat{p}}{p(1-\hat{p})} \overset{(\hat{p}\leq 1/2)}{\leq } 2\cdot \frac{p-\hat{p}}{p} =2\cdot  \frac{\mu - X + 4\sqrt{X L}}{m\cdot p}\\
    & \overset{\eqref{eq:weak-chernoff}, \eqref{eq:weak-X-range}}{\leq} 2\cdot \frac{\sqrt{2\mu L} + 4\sqrt{2\mu L}}{m\cdot p} \overset{(\mu=m\cdot p)}{\leq} 16 \cdot \sqrt{\frac{L}{m\cdot p}} \overset{\eqref{eq:weak-p-ratio}}{\leq} 16 \cdot \sqrt{\frac{L}{m} \cdot \frac{2\cdot w_i}{w_j}}\\
    &\overset{(\text{choice of }m)}{\leq} 16 \cdot \sqrt{\frac{\epsilon^2}{256 \cdot Z } \cdot \frac{w_i}{w_j}} = \epsilon\cdot \sqrt{\frac{w_i}{Z\cdot w_j}}.
\end{align*}
Therefore:
\[
    0\leq
    1-\frac{w_i}{w_j\,r(i,j)}
    \leq
    \epsilon\sqrt{\frac{w_i}{Zw_j}}.
\]
Rearranging the terms we obtain:
\[
    \left(
        1-\epsilon\sqrt{\frac{w_i}{Zw_j}}
    \right)\frac{w_j}{w_i}
    \leq\frac1{r(i,j)}
    \leq\frac{w_j}{w_i},
\]
which concludes the proof.
\end{proof}

\HelperInequalityEpsUB*
\begin{proof}
    We have:
    \begin{align*}
        \left|1 - z\cdot \prod_{i=1}^n x_i\right| &= \left|1 - \prod_{i=1}^n x_i+ (1-z)\cdot \prod_{i=1}^n x_i\right| \\
        &\leq \left|1 - \prod_{i=1}^n x_i\right| + \left|(1-z)\cdot \prod_{i=1}^n x_i\right|\\
        &\leq \sum_{i=1}^n (1-x_i) + |1-z|,
    \end{align*}
    where we used that $1-\prod_{i=1}^n x_i \leq \sum_{i=1}^n (1-x_i)$; indeed, if $n=2$: $1-x_1x_2 = 1-x_1+x_1(1-x_2)\leq (1-x_1)+(1-x_2)$, and for $n>2$ we have: 
    \[
        1-\prod_{i=1}^n x_i \leq \sum_{i=1}^{n-2}(1-x_i) + (1-x_{n-1}x_{n}) \leq \sum_{i=1}^n (1-x_i).\qedhere
    \]
\end{proof}

\section{Missing Proofs for \Cref{sec:non-adaptive-algorithm}}\label{sec:appendix-non-adaptive}

\QuicksortClusterGraph*
\begin{proof}
We describe the algorithm in simple steps. Starting from $S=[n]$, pick a uniform at random pivot $c\in S$. Compare the pivot $c$ with all other items $s\in S\setminus \{c\}$ by calling $\AlgEstimateRatio(c, s, \alpha, \epsilon, \frac{\delta}{n^2})$ and let $r(c,s)$ be the result of the comparison. Define three sets: $C=\{s\in S\setminus \{c\} \mid r(c,s)\notin\{0,\infty\}\}$, $R=\{s\in S\setminus\{c\} \mid r(c,s)=0\}$ and $L = S \setminus(C \cup R)$. The set $C\cup \{c\}$ becomes a new cluster with center $c$. The algorithm then recurs on $R$ and places the resulting clusters after $C\cup \{c\}$ and finally it recurs on $L$ and places the resulting clusters before $C\cup \{c\}$. 

Observe that, for each pair, the algorithm calls \AlgEstimateRatio{} at most once. Therefore, each pair is compared at most $O(\frac{\log(n/\delta)}{\alpha\epsilon^2})$ times. Moreover the total number of calls to \AlgEstimateRatio{} is at most $O(|S|^2)$. Thus, by a union bound, all these calls are successful with probability at least $1-\delta$. We prove that this algorithm is correct conditioning on this event.

Let $(C_1, \dots, C_T)$ and $(c_1, \dots, c_T)$ be the resulting ordered clustering. Consider any $c_i$ and $s\in C_i \setminus \{c_i\}$. Since $r(c_i, s)\notin\{0,\infty\}$, we have by \Cref{lem:weight-pair-ratio-estimate} that $r(c_i, s)\in (1\pm\epsilon)\frac{w_{c_i}}{w_s}$ and $1/r(c_i,s) \in (1\pm \epsilon)\frac{w_s}{w_{c_i}}$. Moreover, $\frac{w_{c_i}}{w_s} \in (\frac{\alpha}{3\alpha+4}, \frac{3\alpha+4}{\alpha})$, and thus, $\frac{w_{c_i}}{w_s} \in (\frac{\alpha}{7}, \frac{7}{\alpha})$. 

Consider now a center $c$. Consider any $v\in L$, since $r(c, v)=\infty$, by \Cref{lem:weight-pair-ratio-estimate} it must be the case that $\frac{w_c}{w_v}\geq \frac{1}{\alpha}$. Similarly, for any $v\in R$, since $r(c,v)=0$, by \Cref{lem:weight-pair-ratio-estimate} it must be true that $\frac{w_c}{w_v}\leq \alpha$. Now observe that for centers $c_i$, $c_{i+1}$, it must either be that, during the execution of the algorithm, $c_i$ was a pivot and $c_{i+1}$ was in $R$ or $c_{i+1}$ was a pivot and $c_i$ was in $L$. Thus, in either case, $\frac{w_{c_{i+1}}}{w_{c_{i}}} \geq \frac{1}{\alpha}$. 
\end{proof}

\section{Missing Proofs for Section~\ref{sec:lower-bounds}: Extension to Pseudo-MNLs}\label{sec:extension-to-pseudo-mnl}
We have observed that one of the issues encountered when learning MNLs is that the ratio of the weights of items might approach $\infty$. In this section, we show by a straight-forward limiting argument that any algorithm that approximately learns MNLs is also approximately learning objects that are not exactly MNLs, but behave exactly as if some of its weights were infinitely larger than others. This happens because these objects arise as the limits of sequences of MNLs.

We define a \emph{subset distribution family} supported on $[n]$ as a collection of distributions $H = \{\nu_S\}_{S\in 2^{[n]}\setminus \{\varnothing\}}$, where each $\nu_S$ is a probability distribution over $S$ and $2^{[n]}$ is the power set of $[n]$. A $\maxsample$ oracle for a subset distribution family is defined as the oracle that on input $S\subseteq [n]$ with $S\neq \varnothing$ returns a sample from the distribution $\nu_S$. Let $\mathcal{F}([n])$ be the collection of subset distribution families supported on $[n]$. Let $\mathcal{M}([n]) \subseteq \mathcal{F}([n])$ be the set of MNLs supported on $[n]$ (where we identify an MNL with the collection of distributions it induces on the slates).

We call a function $d:\mathcal{F}(n)\times\mathcal{F} (n)\rightarrow[0,1]$ a $C$-distance metric if it is symmetric and it satisfies $d(F_1,F_3) \leq C\cdot (d(F_1,F_2)+d(F_2,F_3))$ for any $F_1,F_2,F_3\in \mathcal{F}(n)$. 

Define for convenience $d_{\text{sym}}(M,M')=\max\{\dmul(M,M'), \dmul(M',M)\}$. Note that $d_\text{sym}$ does not respect the triangle inequality; however, it holds:
\[
    d_{\text{sym}}(M_1,M_3) \leq 2\cdot \left(d_{\text{sym}}(M_1,M_2) + d_{\text{sym}}(M_2,M_3)\right).
\]
Note that $\dmul(M,M')\leq d_\text{sym}(M,M')\leq 2\cdot \dmul(M,M')$, which implies that an algorithm that can learn within arbitrary accuracy in $\dmul$ can do so also in $d_\text{sym}$. 

In the rest of the section, we denote by $d(\cdot,\cdot)$ a $C$-distance metric on $\mathcal{F}([n])$. This can be taken to be one of $d_1(\cdot,\cdot)$, or $d_\infty(\cdot,\cdot)$ (with $C=1$), or $d_\text{sym}$ (with $C=2$), and the results will apply in either case. 

\begin{definition}[Pseudo-MNL]
    A pseudo-MNL $\overline{M}\in\mathcal{F}([n])$
    is a limit of a sequence of MNLs with respect
    to $d$. That is, there exists a sequence
    $\{M^{(k)}\}_{k\in\mathbb{N}}\subseteq\mathcal{M}([n])$
    such that
    \[
    \lim_{k\to\infty}d(M^{(k)},\overline{M})=0.
    \]
\end{definition}
That is, a pseudo-MNL is a subset distribution family supported on $[n]$ that is the limit of a sequence of MNLs supported on $[n]$. Note that the class of pseudo-MNLs is the same independently whether $d$ is $d_1$, $d_\infty$ or $d_\text{sym}$ as convergence in one distance implies convergence in the other two. A direct consequence of this definition is that every MNL is also a pseudo-MNL. From this observation it is clear that the task of learning pseudo-MNLs is no easier than that of learning MNLs, but as it turns out, it is no harder either. In fact, the key result we show is the following: any algorithm that solves the MNL learning problem, must also solve the problem of learning pseudo-MNLs.
\begin{theorem}\label{thm:pseudo-mnl-approximation}
    For any $n \in\mathbb{N}$, $\delta\in(0,1)$, and $\varepsilon\in(0,1)$, let $\mathcal{A}$ be an algorithm that given access to a $\maxsample$ oracle for any MNL $M$ supported on $[n]$, makes at most $m(n,\varepsilon,\delta)$ queries and then returns an MNL $\hat{M}$ such that:
    \[
        \Pr[d(M,\hat{M})\leq \varepsilon] \ge 1-\delta.
    \]
    Then, the same algorithm $\mathcal{A}$, when given access to a $\maxsample$ oracle for a pseudo-MNL $\overline{M}$ supported on $[n]$ makes at most $m(n,\varepsilon,\delta)$ queries and returns an MNL $\hat{M}$ such that:
    \[
        \Pr[d(\overline{M},\hat{M})\leq C\cdot \varepsilon] \ge 1-\delta.
    \]
\end{theorem}
\begin{proof}
    By definition of pseudo-MNL, there exists some sequence of MNLs $\{M^{(i)}\}_{i\in \N}$ such that:
    \begin{equation}\label{eq:MNLs-tend-to-pseudo-MNL}
        \lim_{i\to \infty} d(M^{(i)}, \overline{M}) = 0. 
    \end{equation}
    (Recall here, that $d$ is either $d_1$, $d_\infty$, or $d_\text{sym}$). Without loss of generality, we assume that $\mathcal{A}$ is fully adaptive, that is, it queries a single subset $S_j$ and receives a sample $a_j \in S_j$ before choosing its next query. An analogous argument shows that the result holds for algorithms that make batches of queries before obtaining answers to all the queries made in the batch (e.g.,\ non-adaptive algorithms).
    
    Let $\hat{M}$ be the random variable representing the MNL output by the algorithm $\mathcal{A}$. For any choice of $i$ we consider the probability measure $\mathbb{P}_i$ assigning probabilities on events in the experiment in which $\mathcal{A}$ is run with access to a $\maxsample$ oracle for $M^{(i)}$. We also consider the probability measure $\mathbb{P}_\infty$, which assigns probabilities to events based on the experiment in which $\mathcal{A}$ is run with access to a $\maxsample$ oracle for $\overline{M}$.

    Let $T = \{(S_j,a_j)\}_{j\in [m]}$ be the the sequence of query-response pairs produced by the interaction between the algorithm and the oracle (i.e.,\ $T$ is the \emph{transcript}), where for brevity we use $m=m(n,\epsilon,\delta)$. Note that $T$ is random. Moreover, the number of possible transcripts is upper bounded by $f(n,m) = (2^{n-1}\cdot n)^m$, and since both $n$ and $m$ do not change with $i$, the cardinality of the set of possible transcripts is at most a constant with respect to $i$. Let:
    \[
        \mathcal{T}_\infty = \{\tau \text{ is a transcript} \mid \mathbb{P}_\infty[T=\tau] > 0\}.
    \]
    Observe that, since for each slate $S$ and $s\in S$, $M^{(i)}_S(s)>0$ for each $i$, we have that for any $\tau\in \mathcal{T}_\infty$ it holds that $\mathbb{P}_i[T=\tau]>0$. Note that, since any transcript possible under $\overline{M}$ is also possible with any MNL $M^{(i)}$ and the algorithm has a worst-case complexity of $m$ queries for MNLs, it must also make at most $m$ queries when it interacts with $\overline{M}$. 

    Let $\mathcal{S}=\{(S,a) \mid \overline{M}_S(a)>0\}$ and let $D=\min_{(S,a)\in \mathcal{S}} \{\overline{M}_S(a)\}$. Note that $D>0$ and it might possibly depend on $n$. By \Cref{eq:MNLs-tend-to-pseudo-MNL}, there exists an $N_0$ such that for each $i>N_0$, $d(M^{(i)}, \overline{M})\leq \frac{D}{2}$. Therefore, for each $(S,a)\in \mathcal{S}$ and $i>N_0$:
    \begin{equation}\label{eq:mnl-term-greater-than-distance}
        M_S^{(i)}(a) \geq \overline{M}_S(a) - d(M^{(i)}, \overline{M}) \geq \frac{D}{2} \geq d(M^{(i)}, \overline{M}).
    \end{equation}

    We will soon be interested in referring to specific parts of a transcript. To this end, we denote by $T_{\ell:k}$ the pairs $\{(S_j,a_j)\}_{j=\ell}^k$, by $T_{j}^q = S_j$ (the $j$th query in $T$) and by $T_{j}^a = a_j$ the $j$th answer or response in $T$. Fix a transcript $\tau = \{(S_j, a_j)\}_{j\in [m]} \in \mathcal{T}_\infty$. We have, for each $i>N_0$:  

    \begin{align*}
        \mathbb{P}_\infty[ T=\tau] &= \prod_{j\in[m]}\mathbb{P}_\infty[T_j^a = a_j \mid T_{1:j-1} = \tau_{1:j-1} \land T_j^q=S_j]\cdot \mathbb{P}_\infty[T_j^q = S_j \mid T_{1:j-1} = \tau_{1:j-1}]\\
        &=\prod_{j\in[m]}\overline{M}_{S_j}(a_j)\cdot \mathbb{P}_i[T_j^q = S_j \mid T_{1:j-1} = \tau_{1:j-1}]\\
        &\geq \prod_{j\in[m]} \left(M^{(i)}_{S_j}(a_j) - d(M^{(i)}, \overline{M})\right)\cdot \mathbb{P}_i[T_j^q = S_j \mid T_{1:j-1} = \tau_{1:j-1}]\\
        &= \prod_{j\in[m]} \left(\mathbb{P}_i[T_j^a = a_j \mid T_{1:j-1} = \tau_{1:j-1} \land T_j^q=S_j] - d(M^{(i)}, \overline{M})\right)\cdot \mathbb{P}_i[T_j^q = S_j \mid T_{1:j-1} = \tau_{1:j-1}]\\
        &\geq\mathbb{P}_i[ T=\tau] - 2^m \cdot d(M^{(i)},\overline{M}),
    \end{align*}
    where the first inequality is valid because by \Cref{eq:mnl-term-greater-than-distance} $M_{S_j}^{(i)}(a_j) - d(M^{(i)},\overline{M})\geq 0$. Consider now:
    \[
        \mathcal{T}^{(i)} := \{\tau \text{ is a transcript}\mid  \mathbb{P}_\infty[T=\tau] = 0,  \mathbb{P}_i[T=\tau] > 0\}.
    \]
    For any $\tau\in \mathcal{T}^{(i)}$, since $\mathbb{P}_\infty[T=\tau]=0$ but $\mathbb{P}_i[T=\tau]>0$, there must exists $(S_j,a_j)\in \tau$ such that $\overline{M}_{S_j}(a_j)=0$. But then, $\mathbb{P}_i[T=\tau] \leq M^{(i)}_{S_j}(a_j) \leq d(M^{(i)}, \overline{M})$. Therefore, for any $\tau \in \mathcal{T}_\infty \cup \mathcal{T}^{(i)}$, we have, for each $i>N_0$:
    \begin{equation}\label{eq:transcripts-have-similar-probability}
       \mathbb{P}_\infty[T=\tau] \geq \mathbb{P}_i[ T=\tau] - 2^m \cdot d(M^{(i)},\overline{M}).
    \end{equation}

    Fix some $\varepsilon_1 > \varepsilon$. Since $\lim_{i\to\infty} d({M}^{(i)}, \overline{M}) =0$ there exists some $N_1$ such that for all $i > N_1$, we have $d({M}^{(i)}, \overline{M}) <\varepsilon_1 - \varepsilon$. 
    And hence, for all $i > \max\{N_0, N_1\}$:
    \begin{align*}
        \mathbb{P}_\infty\left[d(\hat{M}, \overline{M})\leq C \cdot \varepsilon_1 \right] &\ge \mathbb{P}_\infty\left[d(\hat{M}, M^{(i)}) + d({M}^{(i)}, \overline{M})\leq \varepsilon_1 \right]\\
        &=\mathbb{P}_\infty\left[d(\hat{M}, M^{(i)}) \leq\varepsilon_1- d({M}^{(i)}, \overline{M}) \right]\\
        &\ge \mathbb{P}_\infty\left[d(\hat{M}, M^{(i)}) \leq\varepsilon\right]\\
        &= \sum_{\tau \in \mathcal{T}_\infty}\mathbb{P}_\infty\left[d(\hat{M}, M^{(i)}) \leq \epsilon \,\middle|\, T= \tau\right] \mathbb{P}_\infty\left[T=\tau\right]\\
        &= \sum_{\tau \in \mathcal{T}_\infty}\mathbb{P}_i\left[d(\hat{M}, M^{(i)}) \leq \epsilon \,\middle|\, T= \tau\right] \mathbb{P}_\infty\left[T=\tau\right]\\
        & = \sum_{\tau \in \mathcal{T}_\infty \cup \mathcal{T}^{(i)}}\mathbb{P}_i\left[d(\hat{M}, M^{(i)}) \leq \epsilon \,\middle|\, T= \tau\right] \mathbb{P}_\infty\left[T=\tau\right]\\
        &\overset{\eqref{eq:transcripts-have-similar-probability}}{\geq} \sum_{\tau \in \mathcal{T}_\infty \cup \mathcal{T}^{(i)}}\mathbb{P}_i\left[d(\hat{M}, M^{(i)}) \leq \epsilon \,\middle|\, T= \tau\right] \left(\mathbb{P}_i\left[T=\tau\right] - 2^m \cdot d(M^{(i)},\overline{M})\right)\\
        & \geq \mathbb{P}_i\left[ d(\hat{M}, M^{(i)})\leq \varepsilon\right] - |\mathcal{T}_\infty \cup \mathcal{T}^{(i)}| \cdot 2^m \cdot d(M^{(i)}, \overline{M}) \\
        &\ge 1-\delta -  |\mathcal{T}_\infty \cup \mathcal{T}^{(i)}| \cdot 2^m \cdot d(M^{(i)},\overline{M}).
    \end{align*}
    The above holds for all choices of $i>\max\{N_0,N_1\}$. Note that $|\mathcal{T}_\infty \cup \mathcal{T}^{(i)}|$ can be upper bounded by the total number of valid transcripts which is upper bounded by $f(n,m)$ and does not depend on $i$. Then, by taking the limit $i\to \infty$, and by \eqref{eq:MNLs-tend-to-pseudo-MNL} and the fact that $m$ does not depend on $i$ we have:
    \[
        \mathbb{P}_\infty\left[d(\hat{M}, \overline{M})\leq C\cdot \varepsilon_1 \right] \ge 1-\delta.
    \]
    This holds for all $\varepsilon_1 > \varepsilon$. Define the sequence $\{\varepsilon^{(i)}\}_{i\in\mathbb{N}}$, by:
    \[
        \varepsilon^{(i)} := \varepsilon+{1\over i},
    \]
    and for every $i\in \mathbb{N}$,  let $\mathcal{E}_i$ be the event that $d(\hat{M},\overline{M})\in(C\cdot \varepsilon,C\cdot \varepsilon^{(i)}]$. By the union bound, we have, for all $i\in \mathbb{N}$,
    \[
    \mathbb{P}_\infty[d(\hat{M},\overline{M}) \leq C\cdot \varepsilon] \ge  \mathbb{P}_\infty[d(\hat{M},\overline{M})\leq C \cdot \varepsilon^{(i)}] - \mathbb{P}_\infty[\mathcal{E}_i] \ge 1-\delta -\mathbb{P}_\infty[\mathcal{E}_i],
    \]
    where the second step follows from the derivation above. Note that $\forall i \in \N$, $\mathcal{E}_{i+1} \subseteq \mathcal{E}_i$, and that $\bigcap_{i\in\mathbb{N}}\mathcal{E}_i = \varnothing$. Hence, by taking the limit we obtain:
    \begin{align*}
        \mathbb{P}_\infty[d(\hat{M},\overline{M}) \leq C\cdot \varepsilon] &\geq 1-\delta - \lim_{i \to \infty} \mathbb{P}_\infty[\mathcal{E}_i] = 1-\delta- \mathbb{P}_\infty[\varnothing] = 1-\delta. \qedhere
    \end{align*}
\end{proof}

Pseudo-MNLs have a very specific structure: they can be partitioned into an ordered sequence of disjoint MNLs so that the winner is always an item of the first MNL that intersects the queried slate, and the probability of winning among the items with this property is proportional to their weight in their respective MNL.

For concreteness, we now outline an example of pseudo-MNL. In \Cref{sec:lower-bounds} we introduced the following definition.

\MatchingPseudoMNL*

Where we defined: 
\[
\maxindex{S,\pi} := \max_{i:P_i^\pi\cap S\neq \varnothing} i
\]
as the index of the highest pair $P^{\pi}_1 , \dots , P^{\pi}_{n/2}$ that intersects with the slate $S$.

In \Cref{sec:lower-bounds}, we make use of \Cref{prop:learning-matching-psuedo-mnl} which states that that any algorithm that can learn MNLs, must also learn the matching pseudo-MNLs of \Cref{def:matching-pseudo-mnl}. We now show that matching pseudo-MNLs are indeed pseudo-MNLs. Therefore, \Cref{prop:learning-matching-psuedo-mnl} is an immediate corollary of  \Cref{thm:pseudo-mnl-approximation} and of \Cref{lem:matching-pseduo-MNLs-are-pseudo-MNLs} below.

\begin{lemma}\label{lem:matching-pseduo-MNLs-are-pseudo-MNLs}
    For any even $n\in \mathbb{N}$, $\vec{p} \in [0,1]^{n/2}$ and $\pi \in Sym(n)$, $\overline{M}(n,\vec{p}, \pi)$ is a pseudo-MNL.
\end{lemma}
\begin{proof}
    For simplicity we denote by $\overline{M}:= \overline{M}(n,\vec{p}, \pi)$ We show that there is a sequence $\{M^{(j)}\}_{j\in \mathbb{N}}$ of MNLs with the property that for all $\varepsilon > 0$ there exists some $j_\epsilon \in \mathbb{N}$ such that for all $j\geq j_\epsilon$ it holds that $d_\infty(M^{(j)},\overline{M}) \leq \varepsilon$. Note that this in turns implies that the same is true if $d_\infty$ is replaced by $d_1$, by simply picking a value of $\varepsilon$ that is $n$ times smaller.

    For each $j$ we define the MNL $M^{(j)}$ as the MNL induced by the weights $w^{(j)}_1, \dots , w^{(j)}_n$, defined in the following iterative way:
    \[
        w^{(j)}_{\pi(1)} = 1,
    \]
    and, for every odd $k \in [n]$, $k\geq 3$:     \[
        w^{(j)}_{\pi(k)} = j^3 \cdot w^{(j)}_{\pi(k-1)}
    \]
    and for every even $k \in [n]$:
    \begin{equation}\label{eq:even-weights}
         w^{(j)}_{\pi(k)} = \begin{cases}
         w^{(j)}_{\pi(k-1)} \cdot {\vec{p}_{k/2} \over 1- \vec{p}_{k/2}}&\text{if }\vec{p}_{k/2}\not\in \{0,1\},\\
         {w^{(j)}_{\pi(k-1)} \over j} &\text{if }\vec{p}_{k/2}=0,\\
         {w^{(j)}_{\pi(k-1)} \cdot j} &\text{if }\vec{p}_{k/2}=1.
         \end{cases}
    \end{equation}

    We now show that, for every choice of $\varepsilon >0$ for all sufficiently large $j$, we have that for every slate $S$:
    \[
        \norm{M^{(j)}_S-\overline{M}_S}_\infty \leq \varepsilon.
    \]
    In particular, given $\varepsilon$, we choose $j\in \mathbb{N}$ such that: %
    \begin{equation}\label{eq:j-is-sufficiently-large}
        j \geq \frac{2n}{\epsilon}. %
    \end{equation}

    \noindent
    Fix a slate $S$, and let $i^* = \maxindex{S,\pi}$. By construction, we have that for any $v \in P_{i^*}^\pi$, and $u \in S\setminus P_{i^*}^\pi$:
    \begin{equation}\label{eq:large-factor-ratio}
        {w^{(j)}_{v}\over w^{(j)}_u} \geq {j}.
    \end{equation}

    \noindent
    Suppose $|S \cap P_{i^*}^\pi| = 1$. Then the distribution of winners over $S$ for $\overline{M}(n,\vec{p},\pi)$ is given by:
    \[
        \forall u\in S: \quad \overline{M}_{S}(u) = \begin{cases}
            1 &\text{if }u\in  S \cap P_{i^*}^\pi,\\
            0 & \text{otherwise.}
        \end{cases}
    \]
    In this case, if $u \in S \cap P_{i^*}^\pi$:
    \[
        1 \geq M^{(j)}_{S}(u) = {w^{(j)}_u \over \sum_{\ell \in S} w^{(j)}_\ell } ={1 \over 1+ \sum_{\ell \in S\setminus \{u\}} {w^{(j)}_\ell\over w^{(j)}_u} } = {1\over  1+ \sum_{\ell \in S\setminus P_{i^*}^\pi} {w^{(j)}_\ell\over w^{(j)}_u} } \overset{\eqref{eq:large-factor-ratio}}{\geq} {j \over j+ n} \overset{\eqref{eq:j-is-sufficiently-large}}{\geq} 1-\varepsilon,
    \]
    giving:
    \[
        \left|M^{(j)}_{S}(u) - \overline{M}_{S}(u)\right| \leq \varepsilon,
    \]
    while if $u \in S \setminus P_{i^*}^\pi$, let $v$ be the unique item in $S \cap P_{i^*}^\pi$, we have:
    \begin{equation}\label{eq:match-pseudo-mnl-not-in-first-pair}
        0 \le M^{(j)}_{S}(u) = {w^{(j)}_u \over \sum_{\ell \in S} w^{(j)}_\ell } \le {w^{(j)}_u \over w^{(j)}_v } \overset{\eqref{eq:large-factor-ratio}}{\leq} {1\over j} \overset{\eqref{eq:j-is-sufficiently-large}}{\leq} {\varepsilon},
    \end{equation}
    and:
    \[
        |M^{(j)}_{S}(u) - \overline{M}_{S}(u)| \leq \varepsilon,
    \]
    giving:
    \[
        \norm{M^{(j)}_{S} - \overline{M}_{S}}_\infty \leq \varepsilon,
    \]
    as needed.

    Suppose now that $|S \cap P_{i^*}^\pi| = 2$. In this case, the distribution of winners over $S$ for $\overline{M}(n,\vec{p},\pi)$ is given by:
    \[
        \overline{M}_{S}(u) = \begin{cases}
            0 &\text{if } i \not\in P_{i^*}^\pi,\\
            \vec{p}_{i^*} &\text{if } u \in P_{i^*}^\pi \text{ and }\pi^{-1}(u)\text{ is even,}\\
            1-\vec{p}_{i^*} &\text{if } u \in P_{i^*}^\pi \text{ and }\pi^{-1}(u)\text{ is odd,}\\
        \end{cases}
    \]
    for every $u \in S$.

    We now show that for all $u \in S$, $|M^{(j)}_{S}(u) - {\overline{M}}_{S}(u)| \leq {\varepsilon}$. We divide the proof into three cases: (Case 1:) $u\not\in P_{i^*}^\pi$, (Case 2:) $u\in P_{i^*}^\pi$ and $\pi^{-1}(u)$ is even, and (Case 3:) $u\in P_{i^*}^\pi$ and $\pi^{-1}(u)$ is odd.

    \begin{description}
        \item[Case 1.] If ${u \not\in P_{i^*}^\pi}$, then, by \Cref{eq:match-pseudo-mnl-not-in-first-pair}, $|M^{(j)}_{S}(u) - {\overline{M}}_{S}(u)| = M^{(j)}_{S}(u)  \leq \varepsilon$.        
        \item[Case 2.] If $u \in P_{i^*}^\pi$ and $\pi^{-1}(u)$ is even, let $v$ be the unique item in $P_{i^*}^\pi$ such that $v\neq u$. If $\vec{p}_{i^*} = 0$, we have $\overline{M}_S(u)=0$, while:
        \[
            0\le {M}^{(j)}_S(u) = {w^{(j)}_u \over \sum_{\ell \in S} w^{(j)}_\ell} \leq {w^{(j)}_u \over  w^{(j)}_v} \overset{\eqref{eq:even-weights}}{\leq} {1\over j} \overset{\eqref{eq:j-is-sufficiently-large}}{\leq} \varepsilon
        \]
        and hence:
        \[
            |\overline{M}_S(u) - {M}^{(j)}_S(u)|\le \varepsilon.
        \]

        On the other hand, if $\vec{p}_{i^*} = 1$, we have $\overline{M}_S(u)=1$, and:
        \[
            1 \ge {M}^{(j)}_S(u) =  {w^{(j)}_u \over \sum_{\ell \in S} w^{(j)}_\ell} \overset{\eqref{eq:even-weights},\eqref{eq:large-factor-ratio}}{\ge} {j\over j+n} \overset{\eqref{eq:j-is-sufficiently-large}}{\ge} 1-\varepsilon
        \]
        and hence:
        \[
            |\overline{M}_S(u) - {M}^{(j)}_S(u)|\le \varepsilon.
        \]
        Finally, if $\vec{p}_{i^*} \not\in \{0,1\}$ we have $\overline{M}_S(u)=\vec{p}_{i^*}$ and:
        \[
            {M}^{(j)}_S(u) = {w_{u}^{(j)} \over \sum_{\ell \in S} w^{(j)}_{\ell}} = {w_{u}^{(j)} \over w_{u}^{(j)}+w_{v}^{(j)}}\cdot { w_{u}^{(j)}+w_{v}^{(j)}\over\sum_{\ell \in S} w^{(j)}_{\ell}}.
        \]

        We also have:
        \begin{equation}\label{eq:biggest-items-always-win}
            1\ge  { w_{u}^{(j)}+w_{v}^{(j)}\over\sum_{\ell \in S} w^{(j)}_{\ell}} ={ w_{u}^{(j)}+w_{v}^{(j)}\over w_{u}^{(j)}+w_{v}^{(j)} + \sum_{\ell \in S\setminus P_{i^*}^\pi} w^{(j)}_{\ell}} = { 1\over 1 + {\sum_{\ell \in S\setminus P_{i^*}^\pi} w^{(j)}_{\ell}\over w_{u}^{(j)}+w_{v}^{(j)}}} \overset{\eqref{eq:large-factor-ratio}}{\geq} { j\over j +n}\overset{\eqref{eq:j-is-sufficiently-large}}{\geq} 1-\varepsilon.
        \end{equation}
        
        Hence:
        \begin{equation}\label{eq:approximate-p_i-probability}
            \vec{p}_{i^*} =  {w_{u}^{(j)} \over w_{u}^{(j)}+w_{v}^{(j)}} \ge  {M}^{(j)}_S(u) ={w_{u}^{(j)} \over w_{u}^{(j)}+w_{v}^{(j)}}\cdot { w_{u}^{(j)}+w_{v}^{(j)}\over\sum_{\ell \in S} w^{(j)}_{\ell}} \ge  \vec{p}_{i^*}\cdot \left(1-\varepsilon\right).
        \end{equation}
        Giving:
        \[
             |\overline{M}_S(u) - {M}^{(j)}_S(u)|\le \varepsilon\vec{p}_{i^*} \leq \varepsilon.
        \]
    \item[Case 3.] If $u \in P_{i^*}^\pi$ and $\pi^{-1}(u)$ is odd, let $v$ be the unique item in $P_{i^*}^\pi$ such that $v\neq u$. We have $\overline{M}_S(u) = 1-\vec{p}_{i^*}$. Note that $\pi^{-1}(v)$ is even, thus, by using \eqref{eq:biggest-items-always-win} and \eqref{eq:approximate-p_i-probability} we get:
    \[
        M^{(j)}_{S}(u) = M_S^{(j)}(u) + M_S^{(j)}(v) -M_S^{(j)}(v)= { w_{u}^{(j)}+w_{v}^{(j)}\over\sum_{\ell \in S} w^{(j)}_{\ell}} -M_S^{(j)}(v) \in (1 - \vec{p}_{i^*}) \pm \varepsilon,
    \]
    and hence:
    \[
        |\overline{M}_S(u) - {M}^{(j)}_S(u)|\le \varepsilon.
    \]
    \end{description}

    This then gives $\norm{\overline{M}_S - {M}^{(j)}_S}_\infty \leq \varepsilon$ as needed.
\end{proof}

\section{An Additive Approximation for Pairs is not Sufficient}\label{sec:additive-approximation-is-not-sufficient}
In the next theorem we prove that, for any constant $\epsilon\in(0,1/9)$, if one has an additive $\alpha$-approximation to the distribution for all the slates of size $2$, then, one must have $\alpha \leq \frac{9\epsilon}{n}$ in order for this to guarantee an error of $\epsilon$ on all slates. 

\begin{theorem}
    For any $\varepsilon \in (0,1)$, there exist two families of MNLs $\{M_1^{(n)}\}_{n\in \N}$ and $\{M_2^{(n)}\}_{n\in \N}$ such that, for each $n$:
    \begin{enumerate}[nosep]
        \item $M_1^{(n)}$ and $M_2^{(n)}$ are supported on $[n]$,  
        \item For every pair $u,v\in[n]$:
        \[
            |M_{1,\{u,v\}}^{(n)}(u) - M_{2,\{u,v\}}^{(n)}(u)|\leq \frac{\varepsilon}{n},
        \]
         \item $d_\infty(M_1^{(n)},M_{2}^{(n)})  \geq {\varepsilon\over  9}$.
     \end{enumerate}
\end{theorem}
\begin{proof}
    Consider the following families of MNLs: in $M_1^{(n)}$ there are $n-1$ items of weight $1$ and $1$ item (Item 1) of weight $n$. In $M_2^{(n)}$ there are $n-1$ items of weight $1+\varepsilon$ and one item (Item 1) of weight $n$.

    Condition $(1)$ above is satisfied by construction. We now verify Condition $(2)$.
    For any pair $u,v$ where both $u$ and $v$ are not $1$, we have $M_{1,\{u,v\}}^{(n)}(u) = M_{2,\{u,v\}}^{(n)}(u)$. On the other hand, if one of $u$ and $v$ is equal to $1$, we can assume without loss of generality that $v =1$, we then have:
    \[
          M_{1,\{u,v\}}^{(n)}(u) = {1\over n+1},
    \]
    and:
    \[
         M_{2,\{u,v\}}^{(n)}(u) = {1+\varepsilon \over n+1+\varepsilon} \geq {1\over n+1}.
    \]
    On the other hand:
    \begin{align*}
        M_{2,\{u,v\}}^{(n)}(u) &= {1\over n+1+\varepsilon} + {\varepsilon \over n+1+\varepsilon}\leq {1\over n+1} + {\varepsilon \over n}
    \end{align*}

    and hence:
    \[
        |M_{1,\{u,v\}}^{(n)}(u) - M_{2,\{u,v\}}^{(n)}(u)| \leq {\varepsilon\over n},
    \]
    as needed. Finally, we verify Condition $(3)$. Consider the full slate $[n]$. We have:
    \begin{align*}
        |M_{1,[n]}^{(n)}(1) - M_{2,[n]}^{(n)}(1)| &= M_{1,[n]}^{(n)}(1) - M_{2,[n]}^{(n)}(1) = {n \over 2n-1} - {n \over (2+\varepsilon)n -(1+\varepsilon)}\\
        &={ n \left((2+\varepsilon)n -(1+\varepsilon)\right) - n( 2n-1) \over (2n-1)((2+\varepsilon)n -(1+\varepsilon))} ={\varepsilon n^2 - \varepsilon n  \over (2n-1)((2+\varepsilon)n -(1+\varepsilon))} \\
        & = \varepsilon\cdot { n-1\over 2n-1}\cdot {n \over (2+\varepsilon)n -(1+\varepsilon)} \ge \varepsilon\cdot {1\over 3}\cdot {n \over (2+\varepsilon)n} = \varepsilon\cdot {1\over 3}\cdot {1 \over (2+\varepsilon)} \\
        &\ge \varepsilon\cdot {1\over 3}\cdot {1 \over 3} = {\varepsilon\over 9},
    \end{align*}
where we used the fact that for any $x\ge 2$, ${x-1 \over 2x-1} \geq {1\over 3}$.
\end{proof}

This entails that, if one were to use \citep[Theorem 12]{fjopr18} to approximate the winning distribution of all the slates within $\epsilon$, then, one would need to run their algorithm with $\epsilon' \leq \frac{9\epsilon}{n}$ incurring a cost of $\Omega(\frac{n^4 \log n}{\epsilon^2})$ queries.

\bibliographystyle{plainnat}
\bibliography{main}

\end{document}